\documentclass[acmsmall,screen,manuscript]{acmart}

\AtBeginDocument{%
  \providecommand\BibTeX{{%
    \normalfont B\kern-0.5em{\scshape i\kern-0.25em b}\kern-0.8em\TeX}}}

\usepackage{xspace}
\usepackage[all]{xy} 
\usepackage[capitalise]{cleveref} 

\usepackage{pkg/math-gen}
\usepackage{pkg/is}
\usepackage{pkg/sem}
\usepackage{pkg/languages}
\usepackage{pkg/monoid}

\usepackage{listings}
\newif\ifsubmit
\submittrue
\ifsubmit
\newcommand{\FDComm}[1]{} 
\newcommand{\EZComm}[1]{} 
\newcommand{\DAComm}[1]{} 
\newcommand{\DA}[1]{#1} 
\def\bfd{} \def\efd{}
\def\bez{} \def\eez{}
\else 
\newcommand{\FDComm}[1]{{\scriptsize\color{orange} [\textbf{Francesco}: #1]}}
\newcommand{\EZComm}[1]{{\scriptsize\color{blue} [\textbf{Elena}: #1]}}
\newcommand{\DAComm}[1]{{\scriptsize\color{red} [\textbf{Davide}: #1]}}
\newcommand{\DA}[1]{#1} 
\def\bfd{\begin{color}{orange}}
\def\efd{\end{color}} 
\def\bez{\begin{color}{blue}} 
\def\eez{\end{color}} 

\fi

\setcopyright{acmcopyright}
\copyrightyear{2020}
\acmYear{2020}
\acmDOI{10.1145/1122445.1122456}

\acmJournal{TOCL}
\acmVolume{37}
\acmNumber{4}
\acmArticle{111}
\acmMonth{8}

\def\eg{{\em e.g.},\xspace} 
\def\cf{cf.~}

\newcommand{\maxElem}[2]{\mathsf{maxElem}(#1,#2)}

\theoremstyle{acmplain} 
\newtheorem{theorem}{Theorem}[section]

\theoremstyle{acmdefinition}
\newtheorem{rmk}[theorem]{Remark}

\begin{document}

\title{A meta-theory for big-step semantics}

\author{Francesco Dagnino}
\email{francesco.dagnino@dibris.unige.it}
\orcid{0000-0003-3599-3535}
\affiliation{%
  \department{DIBRIS} 
  \institution{Universit\` a di Genova}
  \city{Genova}
  \country{Italy}
}


\begin{abstract}

It is well-known that big-step semantics is not able to distinguish stuck and non-terminating computations.
This is a strong limitation as it makes very difficult to reason about properties involving infinite computations, such as type soundness, which cannot even be expressed. 

We show that this issue is only apparent: 
the distinction between stuck and diverging computations is implicit in any big-step semantics and it just needs to be uncovered. 
To achieve this goal, 
we develop a systematic study of big-step semantics: 
we introduce an abstract definition of what a big-step semantics is, 
we define a notion of computation by formalising the evaluation algorithm implicitly associated with any big-step semantics, and 
we show how to canonically extend a big-step semantics to characterise stuck and diverging computations. 

Building on these notions, 
we describe a general proof technique to show that a predicate is sound, that is, it prevents stuck computation, with respect to a big-step semantics. 
One needs to check three properties relating the predicate and the semantics and, if they hold, the predicate is sound. 
The extended semantics are essential to establish this meta-logical result, but are of no concerns to the user, who only needs to prove the three properties of the initial big-step semantics. 
%
Finally, we illustrate the technique by several examples, showing that it is applicable also in cases where subject reduction does not hold, hence the standard technique for small-step semantics cannot be used. 

\end{abstract}



\maketitle

  \section{Introduction}
\label{sect:intro} 

The operational semantics of programming languages or software systems specifies, for each program/system configuration, its final result, if any. 
In the case of non-existence of a final result,  there are  two possibilities:
\begin{itemize}
\item either the computation stops with no final result: \emph{stuck computation},
\item or the computation never stops: \emph{non-termination}.
\end{itemize}

There are two main styles to define operationally  a semantic relation:  the \emph{small-step} style \citep{Plotkin81,Plotkin04}, on top of a transition relation representing single computation steps, or directly by a set of rules as in the \emph{big-step} style \citep{Kahn87}. 
Within a small-step semantics it is straightforward to make the distinction between stuck and non-terminating computations, while a typical drawback of the big-step style is that  they  are not distinguished (no judgement is derived in both cases).
Actually, in big-step style, it is not even clear \emph{what a computation is}, because the only available notion is derivability of judgements, which does not convey the dynamics of computation. 

For  this reason, even  though  big-step semantics is  generally more abstract, and sometimes  more intuitive to design and therefore to debug and extend, in the literature much more effort has been devoted to study the meta-theory of small-step semantics, providing properties, and related proof techniques. Notably, the \emph{soundness} of a type system (typing prevents stuck computation) can be proved by  \emph{progress} and \emph{subject reduction}, also called \emph{type preservation}, \citep{WrightF94}. 
Note that soundness cannot even be \emph{expressed} with respect to a big-step semantics, since non-termination and stuckness are confused, as they are both modelled by the absence of a final result. 

Our quest in this paper is to develop a meta-theory of big-step operational semantics, to enable formal reasoning also on non-terminating computations. 
More precisely, we will address the following problems: 
\begin{enumerate}
\item\label{it:bss:1} Defining, in a formal way, \emph{computations} in a given arbitrary  big-step semantics.
\item\label{it:bss:2} According to this definition, describing \emph{extensions of a given arbitrary big-step semantics}, where the difference  between stuckness and non-termination is made explicit.
\item\label{it:bss:3} Providing a general proof technique by  identifying  \emph{three sufficient conditions} on the original big-step rules to prove soundness of a predicate.
\end{enumerate}
All these points rely on the same fundamental cornerstone: 
a general definition of big-step semantics.
Such a definition captures the essential features of a big-step semantics, independently from the particular language or system. 

To address \cref{it:bss:1}, we rely on the intuition that every big-step semantics implicitly defines an evaluation algorithm. 
Then, we identify computations in the big-step semantics with computations of such algorithm. 
Formally, we extend the big-step semantics to  model \emph{partial evaluations}, representing intermediate states of the evaluation process, 
and we formalise the evaluation algorithm by a transition relation between such intermediate states. 
Then, computations are just sequences of transition steps. 
Note that the use of a transition relation is somehow necessary to define computations since they are related to the dynamics of the evaluation and it cannot be captured by derivability in big-step semantics, as it is too abstract. 
In this way, we get a reference model of computations in big-step semantics, 
where we can easily distinguish stuck and non-terminating computations, thus showing that this distinction is actually present, but hidden, in any big-step semantics. 

To deal with  \cref{it:bss:2}, we describe extensions of a given big-step semantics capable to distinguish between stuck and non-terminating computations, as defined in  \cref{it:bss:1}, but abstracting away single computation steps. 
In this way, we show that such a distinction can be made directly in a big-step style. 
More in detail, 
  starting from an arbitrary big-step judgment $\eval{\conf}{\res}$ that evaluates \emph{configurations} $\conf$ into \emph{results} $\res$, the first construction produces an enriched judgement $\evaltr{\conf}{\trres}$ where $\trres$ is either a pair \ple{\bstr,\res} consisting of a finite trace $\bstr$ and a result $\res$, or an infinite trace $\bstrinf$. 
Finite and infinite traces model  the (finite or infinite) sequences of all the configurations encountered during the evaluation. 
In this way, by interpreting coinductively the rules of the extended semantics,  an infinite trace models divergence (whereas no result corresponds to stuck computation). 
Furthermore, we will show that, by using \emph{coaxioms} \cite{AnconaDZ@esop17,Dagnino19}, we can get rid of traces, modelling divergence just by a judgmeent $\eval{\conf}{\divres}$.   
The second construction is in a sense dual. 
It is the general version of the well-known technique presented in Exercise 3.5.16 by \citet{Pierce02} of adding a special result $\Wrong$ explicitly modelling stuck computations (whereas no result corresponds to divergence). 
We will show that these constructions are correct, proving that they represent the intended class of computations as defined in \cref{it:bss:1}. 

Three sufficient conditions in \cref{it:bss:3} are \emph{local preservation}, \emph{$\exists$-progress}, and \emph{$\forall$-progress}. 
For \emph{proving} the result that the three conditions actually ensure soundness, we crucially rely on the extended big-step semantics of \cref{it:bss:2},  since otherwise, as said above, we could not even express the property. 

\emph{However, the three conditions deal only with  the original rules of the given big-step semantics. }  
This means that, practically, in order to use the technique there is no need to deal with the meta-theory (computations and extended semantics). 
This implies, in particular, that our approach does \emph{not} increase the original number of rules.  
  Moreover, the sufficient conditions are  checked only  
on  \emph{single rules}, hence neither induction nor coinduction is needed. 
In a sense, 
they make explicit elementary fragments of the soundness proof, embedding  such semantic-dependent fragments in a semantic-independent (co)inductive proof, which we carry out once and for all (\cf \cref{thm:sound-wrong,thm:sound-div}). 

We support our approach by presenting  several  examples, demonstrating that: on the one hand, soundness proofs can be easily rephrased in terms of our technique, that is, by directly reasoning on big-step rules; 
on the other hand, our technique works also when the property to be checked (for instance, well-typedness) is \emph{not preserved} by intermediate computation steps, whereas it holds for the final result. 
On a side note, our examples concern type systems, but the meta-theory we present in this work holds for any predicate.

Actually, 
we can express two flavours of soundness, depending on whether we make explicit stuckness or non-termination. 
In the former case we express \emph{soundness-must}, which is the notion of soundness we have considered so far, preventing all stuck computations, while 
in the latter case we express \emph{soundness-may}, a weaker notion only ensuring \emph{the existence of} a non-stuck computation. 
Of course, this distinction is relevant only in presence of non-determinism, otherwise the two notions coincide. 
We define a proof technique for soundness-may as well, showing it is correct. 
In the end, it should be noted that \emph{we define soundness  with respect to  a big-step semantics within a big-step formulation}, without resorting to a small-step style (indeed, the extended semantics are themselves big-step). 

\bfd This paper is extracted from the PhD \bez thesis of the author \eez \cite{Dagnino21} and extends the work presented in \citet{DagninoBZD20} in several ways: 
first, we consider a more natural and general notion of big-step semantics;
we provide a detailed analysis of computations in big-step semantics;
we define an additional construction based on coaxioms generalising the approach in \citet{AnconaDZ@oopsla17};
finally, we improve examples considering also imperative languages. \efd 

The rest of the paper is organised as follows. 
\cref{sect:is} recalls basic notions about inference systems and corules. 
\cref{sect:bss-def} provides a definition of big-step semantics. 
\cref{sect:bss-pev} 
\bez defines computations in big-step semantics as possibly infinite sequences of steps in a transition relation on partial evaluation trees. \eez
In this way we get a reference semantic model. 
\cref{sect:constructions} defines two constructions extending  a given big-step semantics: 
one, based on traces, which explicitly models diverging computations and another which explicitly models stuck computations. 
\cref{sect:bss-div} defines a third construction, modelling divergence just as a special result, by using  appropriate corules. 
\cref{sect:bss-soundness} shows how we can express two flavours of soundness against big-step semantics and provides proof techniques to show this property. 
\cref{sect:bss-ex} \bez illustrates \eez the proof technique on several examples. 
Finally, \cref{sect:conclu} concludes the paper, discussing related and future work.

\section{Preliminaries on inference systems and corules}
\label{sect:is} 

In this section, we recall standard notions about (co)inductive definitions by inference systems \cite{Aczel77,LeroyG09,Sangiorgi11}, which are used throughout the paper, and 
also their generalisation by corules, introduced by \citet{AnconaDZ@esop17,Dagnino19,Dagnino21},
which enable more flexible coinductive definitions. 
Corules will be only used in \cref{sect:bss-div,sect:bss-soundness} to properly model and reason about diverging computations in a big-step semantics. 

Assume a set $\universe$, named \emph{universe}, whose elements are called \emph{judgements}. 
An \emph{inference system} $\is$ 
is a set of \emph{(inference) rules}, which are pairs $\RulePair{\prem}{\conclu}$, where $\prem\subseteq\universe$ the set of \emph{premises} and $\conclu\in\universe$ the \emph{conclusion} (a.k.a. \emph{consequence}). 
As it is customary, rules are often written as fractions $\Rule{\prem}{\conclu}$. 
A rule with an empty set of premises is an \emph{axiom}. 
A \emph{proof tree} (a.k.a. \emph{derivation}) for a judgement $\judg$ in $\is$ is a tree whose nodes are (labeled with) judgements in $\universe$, $\judg$ is the root, and there is a node $\conclu$ with set of children $\prem$ only if there is a rule $\RulePair{\prem}{\conclu}$ in $\is$. 
The \emph{inductive} and the \emph{coinductive interpretations} of $\is$, denoted $\Ind{\is}$ and $\CoInd{\is}$, respectively,  are the sets of judgements with, respectively a well-founded\footnote{It is finite when sets of premises are finite.} and an arbitrary (well-founded or not)  proof tree.
We will write $\validInd{\is}{\judg}$ and $\validCo{\is}{\judg}$ when $\judg\in\Ind{\is}$ and $\judg\in\CoInd{\is}$, respectively. 
Set-theoretically, we say that a subset $X\subseteq\universe$ is 
\emph{($\is$-)closed} if, for every rule $\RulePair{\prem}{\judg}\in\is$, $\prem \subseteq X$ implies $\judg\in X$, and 
\emph{($\is$-)consistent} if, for every $\judg\in X$, there is a rule $\RulePair{\prem}{\judg}\in \is$ such that $\prem \subseteq X$.
Then, it can be proved that $\Ind{\is}$ is the least closed subset and $\CoInd{\is}$ is the largest consistent subset and this provides us with the following proof principles:  
\begin{description}
\item [induction principle] if $X\subseteq\universe$ is closed then $\Ind{\is}\subseteq X$
\item [coinduction principle] if $X\subseteq\universe$ is consistent then $X\subseteq \CoInd{\is}$
\end{description}

We recall now the notion of inference system with corules \cite{AnconaDZ@esop17,Dagnino19,Dagnino21},
which mixes induction and coinduction in a specific way.

For a set $X \subseteq \universe$, let 
 $\Restrictis{\is}{X}$ denote the inference system obtained from $\is$ by keeping only rules with  conclusion  in $X$.

\begin{definition}[Inference system with corules] \label{def:cois} 
An \emph{inference system with corules}, or \emph{generalised inference system}, is a pair $\Pair{\is}{\cois}$ where $\is$ and $\cois$ are inference systems, whose elements are called \emph{rules} and \emph{corules}, respectively.
A corule with empty set of premises is a \emph{coaxiom}. 
The \emph{interpretation} $\FlexCo{\is}{\cois}$ of such a pair is defined by
$\FlexCo{\is}{\cois}=\CoInd{\Restrictis{\is}{\Ind{\is\cup\cois}}}$.
\end{definition}

Thus, the interpretation $\FlexCo{\is}{\cois}$ is basically \emph{coinductive}, but restricted to
a universe of judgements which is \emph{inductively defined} by the (potentially) larger system $\is \cup \cois$. 
In proof-theoretic terms, $\FlexCo{\is}{\cois}$ is the set of judgements which have an arbitrary (well-founded or not)  proof tree in $\is$ whose nodes all have a well-founded proof tree in $\is\cup\cois$, that is, the (standard) inference system consisting of both rules and corules. 
We will write $\validFCo{\is}{\cois}{\judg}$ when $\judg$ is derivable in $\Pair{\is}{\cois}$, that is, $\judg\in\FlexCo{\is}{\cois}$.

We illustrate these notions by a simple example. 
As usual, sets of rules are expressed by \emph{meta-rules} with side conditions, and analogously sets of corules are expressed by \emph{meta-corules} with side conditions.
(Meta-)corules will be written with thicker lines, to be distinguished from (meta-)rules.
The following inference system defines the maximal element of a list of natural numbers, where $\EList$ is the empty list, and $\x\cons u$ the list with head $\x$ and tail $u$.
\[
\Rule{}{\maxElem{x\cons \EList}{x}}\BigSpace\Rule{\maxElem{u}{y}}{\maxElem{x\cons u}{z}}z=\max(x,y) 
\]
The inductive interpretation is defined only on finite lists, since for infinite lists an infinite proof is needed. 
However, the coinductive interpretation allows the derivation of wrong judgements. 
For instance, let $L = 1:2:1:2:1:2:\ldots$. 
Then, any judgement $\maxElem{L}{x}$ with $x\geq 2$ can be derived, as illustrated by the following examples.
\[
\Infer{
  \Infer{
    \Infer{\ldots}{\maxElem{L}{2}}{} 
  }{\maxElem{2\cons L}{2}}{}
}{\maxElem{1\cons 2\cons L}{2}}{}
\BigSpace
\Infer{
  \Infer{
    \Infer{\ldots}{\maxElem{L}{5}}{}
  }{\maxElem{2\cons L}{5}}{} 
}{\maxElem{1\cons 2\cons L}{5}}{} 
\]
By adding a corule (in this case a coaxiom), we add a constraint which forces the greatest element
 to belong to the list,  so that wrong results are ``filtered out'':
\[
\Rule{}{\maxElem{x\cons \EList}{x}}\BigSpace\Rule{\maxElem{u}{y}}{\maxElem{x\cons u}{z}}z=\max(x,y) 
\BigSpace\CoRule{}{\maxElem{x\cons u}{x}}
\]
Indeed, the judgement $\maxElem{1\cons 2\cons L}{2}$ has the infinite proof tree shown above, and each node has a finite proof tree in the inference system extended by the corule:
\[
\Infer{
  \Infer{
    \Infer{\ldots}{\maxElem{L}{2}}{} 
  }{\maxElem{2\cons L}{2}}{} 
}{\maxElem{1\cons 2\cons L}{2}}{} 
\BigSpace
\Infer{
  \CoInfer{ }{\maxElem{2\cons L}{2}}{} 
}{\maxElem{1\cons 2\cons L}{2}}{} 
\]
On the other hand, the judgement $\maxElem{1\cons 2\cons L}{5}$ has the infinite proof tree shown above, but has \emph{no finite proof tree} in the inference system extended by the corule. Indeed, since $5$ does not belong to the list, the corule can never be applied. Hence, this judgement cannot be derived in the inference system with corules.
Finally, note that the judgement $\maxElem{1\cons 2\cons L}{1}$ has a finite proof tree in the inference system extended by the corule, but has no proof tree in the system with no corules, as $1$ is not an upper bound of the list. 
We refer to \cite{AnconaDZ@esop17,AnconaDZ@oopsla17,AnconaDZ@ecoop18,Dagnino19,DagninoAZ20,Dagnino21} for other examples.

Let $\Pair{\is}{\cois}$ be a generalised inference system. 
The interpretation $\FlexCo{\is}{\cois}$ can be characterised as the largest $\is$-consistent subset of $\Ind{\is\cup\cois}$, and this provides us with 
the \emph{bounded coinduction principle}, a generalisation of the standard coinduction principle.  

\begin{theorem}[Bounded coinduction]\label{thm:bcoind}
Let $X\subseteq\universe$. 
If  $X$ is $\is$-consistent and ${X\subseteq\Ind{\is\cup\cois}}$, 
then $X \subseteq\FlexCo{\is}{\cois}$.
\end{theorem}

In other words, to prove that every judgement in $X$ is derivable in \ple{\is,\cois}, 
we have to prove that 
every judgement in $X$ has a well-founded proof tree in $\is\cup\cois$ and 
every judgement in $X$ is the conclusion of a rule whose premises are all in $X$.

  \section{Defining big-step semantics}
\label{sect:bss-def}

As mentioned in the introduction, the corner stone of this paper is 
a formalisation of what a big-step semantics is, that captures its essential features, subsuming 
a large class of examples. \FDComm{riferimento a sezione}  
This 
 enables a general formal reasoning on an arbitrary big-step semantics. 

\begin{definition} \label{def:bss}
A \emph{big-step semantics} is a triple $\Triple{\ConfSet}{\ResSet}{\RuleSet}$ where:
\begin{itemize} 
\item $\ConfSet$ is a set of \emph{configurations} $\conf$.
\item $\ResSet$ is a set of \emph{results} $\res$. 
A \emph{judgment} $\judg$ is a pair written  $\eval{\conf}{\res}$, meaning that  configuration $\conf$ evaluates to result $\res$.
 Set  $\ConfSet(\judg) = \conf$ and $\ResSet(\judg) = \res$. 
\item $\RuleSet$ is a set of \emph{(big-step) rules} $\rho$ of shape
\begin{center}
\begin{tabular}{cl}
$\Rule{ \judg_1\Space\ldots\Space \judg_n}{ \eval{\conf}{\res}}$ & 
also written in \emph{inline format}: $\inlinerule{\judg_1\ldots\judg_n}{\conf}{\res}$
\end{tabular}
\end{center}
where $\judg_1\ldots\judg_n$, with $n\ge 0$, is a sequence of premises. 
Set  $\ConfSet(\rho) =\conf$, $\ResSet(\rho) = \res$  and, for $i \in 1..n$, $\ConfSet(\rho, i) = \ConfSet(\judg_i)$ and $\ResSet(\rho, i) = \ResSet(\judg_i)$. 

We require $\RuleSet$ to satisfy the \emph{bounded premises} condition: 
\begin{description} 
\item[\BP] for every $\conf\in\ConfSet$, there exists $b_\conf \in \N$ such that, for each $\rho = \inlinerule{\judg_1\ldots\judg_n}{\conf}{\res}$, $n \le b_\conf$.
\end{description} 
\end{itemize}
\end{definition}
We will use the inline format, more concise and manageable, for the development of the meta-theory, e.g., in constructions.

Big-step rules, as defined above, are very much like inference rules (\cf \cref{sect:is}),
but they carry slightly more structure with respect to them. 
Notably, premises are a sequence rather than a set, that is, they are ordered and there can be repeated premises. 
Such additional structure, however, does not affect derivability, namely, the inference operator and so the interpretations of such rules. 
Therefore, given a big-step semantics $\Triple{\ConfSet}{\ResSet}{\RuleSet}$, slightly abusing the notation, 
we denote by $\RuleSet$ the inference system obtained by forgetting such additional structure, and
define, as usual, the \emph{semantic relation} as the inductive interpretation of $\RuleSet$. 
Then,  we write $\validInd{\RuleSet}{\eval{\conf}{\res}}$ when the judgment $\eval{\conf}{\res}$ is derivable in $\RuleSet$.  

Even though the additional structure of big-step rules does not affect the semantic relation they define, 
it is crucial to develop the meta-theory, allowing abstract reasoning about an arbitrary big-step semantics. 
It will be used in all results in this paper: to define computations in big-step semantics, then to provide constructions yelding extended semantics able to distinguish stuck and diverging computations and, finally, to define proof techniques for soundness. 
Indeed, as premises are a sequence, we know in which order configurations in the premises should be evaluated. 

In practice, the (infinite) set of rules $\RuleSet$ is described by a finite set of meta-rules, each one with a finite number of premises.
As a consequence, for each configuration, the number of premises of rules with such a configuration in the conclusion is not only finite but \emph{bounded}. 
Since we have no notion of meta-rule, we explicitly require this feature (relevant in the following) by the bounded premises (\BP) condition. 

We end this section by illustrating the above definitions and conditions on a simple example: a $\lambda$-calculus with constants for natural numbers, successor and non-deterministic choice, shown in \cref{fig:bss-ex-lambda}.
We denote by $\x$ variables and by $\natconst$ natural number constants. 
\begin{figure}
\begin{small}
\begin{math}
\begin{array}{rcll}
\e &::=& \x \mid \val \mid \e_1\appop \e_2 \mid \SuccExp e \mid \Choice{\e_1}{\e_2}& \text{expression}\\
\val &::=& \natconst \mid \Lam{\x}\e & \text{value}
\end{array}
\end{math}

\HSep

\begin{math}
\begin{array}{c}
\MetaRule
{val}
{}
{\eval{\val}{\val}}{}
\qquad
\MetaRule
{ app }
{\eval{\e_1}{\lambda x.e}\quad\eval{\e_2}{\val_2}\quad\eval{\subst{\e}{\val_2}{\x}}{\val}}
{\eval{\AppExp{\e_1}{\e_2}}{\val}}{}\\[3ex]
\MetaRule{succ}{\eval{\e}{\natconst}}{\eval{\SuccExp\e}{\natconst+1}}{} \Space
\MetaRule{choice}{\eval{\e_i}{\val}}{\eval{\Choice{\e_1}{\e_2}}{\val}}{i=1,2}
\end{array}
\end{math}

\HSep

\begin{math} 
\begin{array}{l}
\metainlinerule{val}{\EList}{\val}{\val} \\
\metainlinerule{app}{\eval{\e_1}{\Lam{\x}\e}\ \ \eval{\e_2}{\val_2}\ \ \eval{\subst{\e}{\val_2}{\x}}{\val}}{\e_1 \appop \e_2}{\val} \\
\metainlinerule{succ}{\eval{\e}{\natconst}}{\SuccExp \e}{\natconst+1} \\
\metainlinerule{choice}{ \eval{\e_i}{\val} }{\Choice{\e_1}{\e_2}}{\val}\ i =1,2 
\end{array}
\end{math}
\end{small}
\caption{Example of big-step semantics}\label{fig:bss-ex-lambda}
\end{figure}
It is immediate to see this example as an instance of \cref{def:bss}:
\begin{itemize}
\item Configurations and results are expressions, and values, respectively.\footnote{In general, configurations may include additional components and results are not necessarily particular configurations, see, \eg \cref{sect:bss-fjl}.}
\item  To have the set of (meta-)rules in our required shape, abbreviated in inline format in the bottom section of the figure, we have only to assume an order on premises of rule \rn{app}.
\end{itemize}

\label{page:bss-strategy} 
\begin{rmk}
The order of premises chosen for rule \rn{app} in \cref{fig:bss-ex-lambda} formalises the evaluation strategy for 
an application $\e_1\appop\e_2$ where  
first (1) evaluates $\e_1$, then (2) checks that the value of $\e_1$ is a $\lambda$-abstraction, finally (3) evaluates $\e_2$. 
That is, left-to-right evaluation with early error detection.
Other strategies can be obtained by choosing a different order or by adjusting  big-step rules. 
Notably, right-to-left evaluation (3)-(1)-(2)  
can be expressed by just swapping the first two premises, that is:
\[
\metainlinerule{app-r}{\eval{\e_2}{\val_2}\ \ \eval{\e_1}{\Lam{\x}\e}\ \ \eval{\subst{\e}{\val_2}{\x}}{\val}}{\e_1\appop\e_2}{\val} 
\]
Left-to-right evaluation with late error detection  (1)-(3)-(2)  
can be expressed as follows:
\[
\metainlinerule{app-late}
  {\eval{\e_1}{\val_1}\ \ \eval{\e_2}{\val_2}\ \ \eval{\val_1}{\Lam{\x}\e}\ \ \eval{\subst{\e}{\val_2}{\x}}{\val}}
  {\e_1\appop\e_2}{\val} 
\]
We can even opt for a non-deterministic approach by taking more than one rule among \rn{app}, \rn{app-r} and \rn{app-late}. 
As said above, these different choices do not affect the semantic relation inductively defined by the inference system, which is always the same. 
However, they will affect computations and thus the extended semantics distinguishing stuck computation and non-termination. 
Indeed, if the evaluation of $\e_1$ and $\e_2$ is stuck and non-terminating, respectively, we should obtain a stuck computation with rule  \rn{app}  and non-termination with  rule \rn{app-r}; 
further, if $\e_1$ evaluates to a natural constant and $\e_2$ diverges, we should obtain a stuck computation with rule \rn{app} and non-termination with rule \rn{app-late}. 
\end{rmk} 
 
In summary, to see a typical big-step semantics as an instance of our definition, it is enough to identify \EZComm{define?} configurations and results and to assume an order (or more than one) on premises. 

  \section{Computations in big-step semantics}
\label{sect:bss-pev}

Intuitively, the evaluation of a configuration $\conf$ is a \emph{dynamic} process and, as such, it may either successfully terminate producing the final result, or get stuck, or never terminate. 
However, a big-step semantics just tells us whether a configuration $\conf$ evaluates to a certain result $\res$, without describing the dynamics of such evaluation process. 
This is nice, because it allows us to abstract away details about intermediate states in the evaluation process, but it makes quite difficult to reason about concepts like non-termination and stuckness, since they refer to computations and we do not even know what a computation is in a big-step semantics. 

In this section, we show that, given a big-step semantics as defined in \cref{def:bss}, 
we can recover the dynamics of the evaluation, by defining \emph{computations}, which, in a sense, are implicit in a big-step specification. 
To this end, we extend the big-step semantics, so that we can represent partial (or incomplete) evaluations, modelling intermediate states of the evaluation process. 
Then, we model the dynamics by a transition relation beween such partial evaluations, 
hence, as usual, a computation will be a (possibly infinite) sequence of transitions. 

Let us assume a big-step semantics $\Triple{\ConfSet}{\ResSet}{\RuleSet}$. 
As said above, the first step is to extend such semantics to model partial evaluations. 
To this end, first of all, we introduce a special result $\Unknown$, so that a judgment $\eval{\conf}{\Unknown}$ (called \emph{incomplete}, whereas a judgment $\eval{\conf}{\res}$  is \emph{complete})  means that  the evaluation of $\conf$ is not completed yet.  
Set $\UnResSet = \ResSet + \{\Unknown\}$ whose elements are ranged over by $\unres$. 
We now define an augmented set of rules $\UnRuleSet$ to properly handle the new result $\Unknown$:

\begin{definition}[Rules for partial evaluation] \label{def:pev-rules}
The set of rules $\UnRuleSet$ is obtained from $\RuleSet$ by adding the following rules: 
\begin{description}
\item [start rules] For each configuration $\conf \in \ConfSet$, define rule $\startrule{\conf}$ as $\Rule{}{\eval{\conf}{\Unknown}}$.
\item [partial rules] For each rule $\rho = \inlinerule{\judg_1\ldots\judg_n}{\conf}{\res}$ in $\RuleSet$, index $i \in 1..n$, and $\unres \in \UnResSet$,  
define rule $\partialrule{\rho}{i}{\unres}$ as 
\[
\Rule{\judg_1 \Space \ldots \Space \judg_{i-1} \Space \eval{\ConfSet(\judg_i)}{\unres} }{ \eval{\conf}{\Unknown} }
\]
\end{description}
\end{definition} 
Intuitively, start rules allow us to begin the evaluation of any configuration, while 
partial rules allow us to partially apply a rule from $\RuleSet$ to derive a partial judgement. 
Note that the last premise of a partial rule can be either complete ($\unres\in\ResSet$) or incomplete ($\unres = \Unknown$), in the latter case we also call it a \emph{$\Unknown$-propagation} rule, since it propagates $\Unknown$ from premises to the conclusion. 

It is important to observe that the construction described above yields 
a triple $\Triple{\ConfSet}{\UnResSet}{\UnRuleSet}$, which is a big-step semantics according to \cref{def:bss}. 
\footnote{The condition (\BP) is satisfied as the number of premises of the additional rules is bounded by that of a rule in the original semantics. }
In \cref{fig:bss-ex-lambda-unknown} we report rules added by the construction in \cref{def:pev-rules} to the big-step semantics of the $\lambda$-calculus in \cref{fig:bss-ex-lambda}. 

\begin{figure}
\begin{math}
\begin{array}{c}
\Rule{}{\eval{\e}{\Unknown}} 
\BigSpace
\Rule{\eval{\e}{\val_\Unknown}}{\eval{\SuccExp\e}{\Unknown}} 
\BigSpace
\Rule{\eval{\e_i}{\val_\Unknown}}{\eval{\Choice{\e_1}{\e_2}}{\Unknown}}\ i=1,2 
\\[3ex]
\Rule
{\eval{\e_1}{\val_\Unknown} }
{\eval{\AppExp{\e_1}{\e_2}}{\Unknown}} 

\Rule
{\eval{\e_1}{\Lam{\x}\e}\Space\eval{\e_2}{\val_\Unknown} }
{\eval{\AppExp{\e_1}{\e_2}}{\Unknown}} 

\Rule
{\eval{\e_1}{\Lam{\x}\e}\Space\eval{\e_2}{\val_2}\Space\eval{\subst{\e}{\val_2}{\x}}{\val_\Unknown}}
{\eval{\AppExp{\e_1}{\e_2}}{\Unknown}} 
\end{array}
\end{math}
\caption{Rules for $\Unknown$ for the $\lambda$-calculus in \cref{fig:bss-ex-lambda}.}\label{fig:bss-ex-lambda-unknown}
\end{figure}

Given a big-step semantics $\Triple{\ConfSet}{\ResSet}{\RuleSet}$, 
using rules in $\RuleSet$, we can build trees called \emph{evaluation trees}. 
Such trees are very much like proof trees for an inference system, 
with the only difference that evaluation trees are \emph{ordered} trees, because premises of big-step rules are a sequence. 
Roughly, an evaluation tree is an ordered tree with nodes labelled by semantic judgements, such that for each node labelled by $\eval{\conf}{\res}$ with sequence of children $\judg_1,\ldots,\judg_n$, there is a rule $\inlinerule{\judg_1\ldots\judg_n}{\conf}{\res}$ in $\RuleSet$. 

An evaluation tree for $\Triple{\ConfSet}{\UnResSet}{\UnRuleSet}$ is called a \emph{partial evaluation tree}, as it can contain incomplete judgements.
We say that a partial evaluation tree is \emph{complete} if it only contains complete judgments, it is \emph{incomplete} otherwise. 
Finite partial evaluation trees indeed model possibly incomplete evaluation of configurations, namely, the intermediate states of the evaluation process, because big-step rules can be partially applied. 
Hence, they are the fundamental building block, which will allow us to define computations in big-step semantics. 

%
In the next subsection we will give a formal definition of (partial) evaluation trees, similar to that of proof trees \cite{Dagnino19,Dagnino20,Dagnino21}. 
This formal definition is needed to state some results and to carry out proofs in a rigorous way, and it is not essential to follow the rest of the paper, hence the reader not interested in formal details can skip it, relying on the above semiformal definition.

\subsection{The structure of partial evaluation trees}
\label{sect:bss-pet}

We give a formal account of (partial) evaluation trees, which is useful to state and prove technical results in the next sections. 
This development is based on the definition and properties of trees provided by \citet{Courcelle83}, adjusted to our specific setting. 

Let $\NPos$ be the set of positive natural numbers, $\List{\NPos}$ the set of finite sequences of positive natural numbers and $\LabelSet$ a set of labels. 
An  \emph{ordered tree} labelled in $\LabelSet$ is a partial function $\fun{\tr}{\List{\NPos}}{\LabelSet}$ such that 
$\dom(\tr)$ is not empty, and, for each $\alpha \in \List{\NPos}$ and $n \in \NPos$, 
if $\alpha n \in \dom(\tr)$ then 
$\alpha \in \dom(\tr)$ and, for all $k \le n$, $\alpha k \in \dom(\tr)$. 
Given an ordered tree $\tr$ and $\alpha \in \dom(\tr)$, set $\trbr_\tr(\alpha)=\sup\{n \in \N \mid \alpha n \in \dom(\tr)\}$ the \emph{branching} of $\tr$ at $\alpha$, and 
 $\subtr{\tr}{\alpha}$ the \emph{subtree} of $\tr$ rooted at $\alpha$, that is, $\subtr{\tr}{\alpha}(\beta) = \tr(\alpha\beta)$, for all $\beta\in\List{\NPos}$. 
The \emph{root} of $\tr$ is $\rt(\tr) = \tr(\EList)$ and obviously we have $\tr = \subtr{\tr}{\EList}$. 
Finally, we write $\Rule{\tr_1\Space \ldots \Space \tr_n}{x}$ for the tree $\tr$ defined by $\tr(\EList)=x$, and $\tr(i\alpha) = \tr_{i}(\alpha)$  for all $i \in 1..n$.
%
Since in the following we will only deal with ordered trees, we will refer to them just as trees. 
 
Let us assume a big-step semantics $\Triple{\ConfSet}{\ResSet}{\RuleSet}$. 
Assume also that labels in $\LabelSet$ are semantic judgments $\eval{\conf}{\res}$,  
then we can define evaluation trees as follows: 

\begin{definition} \label{def:eval-tree}
A tree $\fun{\tr}{\List{\NPos}}{\LabelSet}$  is an \emph{evaluation tree}  in $\Triple{\ConfSet}{\ResSet}{\RuleSet}$, 
if, for each $\alpha \in \dom(\tr)$ with $\tr(\alpha) = \eval{\conf}{\res}$, there is \mbox{$\inlinerule{\tr(\alpha 1)\ldots\tr(\alpha\trbr_\tr(\alpha))}{\conf}{\res} \in \RuleSet$.}
\end{definition}
Note that, starting from an evaluation tree $\tr$, we can construct a proof tree for the inference system denoted by $\RuleSet$, by forgetting the order on sibling nodes and removing duplicated children. 
Therefore, if $\tr$ is a finite evaluation tree with $\rt(\tr) = \eval{\conf}{\res}$, then $\validInd{\RuleSet}{\eval{\conf}{\res}}$ holds. 

\begin{definition} \label{def:pet}
A \emph{partial evaluation tree} in $\Triple{\ConfSet}{\ResSet}{\RuleSet}$ is an evaluation tree in $\Triple{\ConfSet}{\UnResSet}{\UnRuleSet}$. 
\end{definition}

The following proposition assures two key properties of partial evaluation trees. 
First, if there is some $\Unknown$, then it is propagated to ancestor nodes. 
Second, for each level of the tree there is at most one $\Unknown$. We set $\Len{\alpha}$ the length of $\alpha \in \List{\NPos}$.

\begin{proposition} \label{prop:tree-unknown}
Let $\tr$ be a partial evaluation tree, then the following hold:
\begin{enumerate}
\item\label{prop:tree-unknown:1} for all $\alpha n \in \dom(\tr)$, if $\ResSet_\Unknown(\tr(\alpha n)) = \Unknown$ then $\ResSet_\Unknown(\tr(\alpha)) = \Unknown$.
\item\label{prop:tree-unknown:2} for all $n \in \N$, there is at most one $\alpha \in \dom(\tr)$ with $\Len{\alpha} = n$ such that $\ResSet_\Unknown(\tr(\alpha)) = \Unknown$.
\end{enumerate}
\end{proposition}
\begin{proof}
To prove \cref{prop:tree-unknown:1}, we just have to note that the only rules having a premise $\judg$ with $\ResSet_\Unknown(\judg) = \Unknown$ are $\Unknown$-propagation rules, which also have conclusion $\judg'$ with $\ResSet_\Unknown(\judg') = \Unknown$; hence the thesis is immediate.
To prove \cref{prop:tree-unknown:2}, we proceed by induction on $n$. 
For $n = 0$, there is only one  $\alpha \in \List{\NPos}$ with $\Len{\alpha} = 0$ (the empty sequence), hence the thesis is trivial.
Consider $\alpha = \alpha' k \in \dom(\tr)$ with $\Len{\alpha} = n+1$.
If $\ResSet_\Unknown(\tr(\alpha)) = \Unknown$, then, by \cref{prop:tree-unknown:1}, $\ResSet_\Unknown(\tr(\alpha')) = \Unknown$, and, by induction hypothesis,  $\alpha'$ is the only sequence of length $n$ in $\dom(\tr)$ with this property. 
Therefore, another node $\beta \in \dom(\tr)$, with $\Len{\beta} = n+1$ and $\ResSet_\Unknown(\tr(\beta)) = \Unknown$, must satisfy $\beta = \alpha' h$ for some $h \in \NPos$;
hence, since $\tr$ is a partial evaluation tree, $\tr(\alpha)$ and $\tr(\beta)$ are two premises of the same rule with $\Unknown$ as result, thus they must coincide, since all rules in $\UnRuleSet$ have at most one premise with $\Unknown$.
\end{proof}

\begin{corollary} \label{cor:complete-tree}
Let $\tr$ be a partial evaluation tree, then 
$\ResSet_\Unknown(\rt(\tr)) \in \ResSet$ if and only if  $\tr$ is complete. 
\end{corollary}

We can define a relation\footnote{This is a slight variation of similar relations on trees considered by \citet{Courcelle83,Dagnino19}.}, denoted by $\TreeOrd$, on trees labelled by possibly incomplete judgements,  as follows: 

\begin{definition}\label{def:tree-ord}
Let $\tr$ and $\tr'$ be trees labelled by possibly incomplete semantic judgements. 
Define $\tr \TreeOrd \tr'$ if and only if 
$\dom(\tr) \subseteq \dom(\tr')$ and,
for all $\alpha \in \dom(\tr)$, $\ConfSet(\tr(\alpha)) = \ConfSet(\tr'(\alpha))$  and 
$\ResSet_\Unknown(\tr(\alpha)) \in \ResSet$ implies  $\subtr{\tr}{\alpha} = \subtr{\tr'}{\alpha}$. 
\end{definition}

Intuitively, $\tr\TreeOrd\tr'$ means that $\tr'$ can be obtained from $\tr$ by adding new branches or replacing some $\Unknown$s with results. 
We use $\StrTreeOrd$ for the strict version of $\TreeOrd$. 
\bfd Note that, if $\tr\TreeOrd\tr'$, then, for all $\alpha\in\List{\NPos}$, $\tr'(\alpha)$ is \emph{more defined} than $\tr(\alpha)$, because, 
either $\tr(\alpha)$ is undefined, or $\tr(\alpha)$ is incomplete and $\ConfSet(\tr(\alpha)) = \ConfSet(\tr'(\alpha))$, or $\tr(\alpha) = \tr'(\alpha)$.
\FDComm{spero che si capisca} \efd \EZComm{s\`i}

It is easy to check that $\TreeOrd$ is a partial order and, if $\tr\TreeOrd\tr'$, then, for all $\alpha \in \dom(\tr)$, $\subtr{\tr}{\alpha} \TreeOrd \subtr{\tr'}{\alpha}$. 
The following proposition shows some, less trivial, properties of $\TreeOrd$. 

\begin{proposition} \label{prop:tree-ord}
The following properties hold:
\begin{enumerate}
\item\label{prop:tree-ord:1} for all trees $\tr$ and $\tr'$, if $\tr \TreeOrd \tr'$ and $\ResSet_\Unknown(\rt(\tr)) \in \ResSet$, then $\tr = \tr'$
\item\label{prop:tree-ord:2} for each increasing sequence $(\tr_i)_{i \in \N}$ of trees, there is a least upper bound $\tr = \Treelub \tr_n$.
\end{enumerate}
\end{proposition}
\begin{proof}
\cref{prop:tree-ord:1} is immediate by definition of $\TreeOrd$.
To prove \cref{prop:tree-ord:2}, first note that, since for all $n \in \N$, $\tr_n \TreeOrd \tr_{n+1}$, 
for all $\alpha \in \List{\NPos}$ we have that, for all $n \in \N$, if $\tr_n(\alpha)$ is defined, then, 
for all $k \ge n$, $\ConfSet(\tr_k(\alpha)) = \ConfSet(\tr_n(\alpha))$, and, 
if $\ResSet_\Unknown(\tr_n(\alpha)) \in \ResSet$, then $\tr_k(\alpha) = \tr_n(\alpha)$. 
Hence, for all $n \in \N$, there are only three possibilities for $\tr_n(\alpha)$: it is either undefined, or equal to $\eval{\conf}{\Unknown}$, or equal to $\eval{\conf}{\res}$, where $\conf$ and $\res$ are always the same. 
Let us denote by $k_\alpha$ the least index $n$ where $\tr_n(\alpha)$ is most defined, 
that is, 
if $\tr_n(\alpha)$ is always undefined, then $k_\alpha = 0$, 
if $\tr_n(\alpha)$ is eventually always equal to $\eval{\conf}{\Unknown}$, then $k_\alpha$ is the least  $n$ where $\tr_n(\alpha)$ is defined, and, 
if $\tr_n(\alpha)$ is eventually always equal to $\eval{\conf}{\res}$, then $k_\alpha$ is the least $n$ where $\tr_n(\alpha) = \eval{\conf}{\res}$. 
Therefore, for all $n \ge k_\alpha$, we have that $\tr_n(\alpha) = \tr_{k_\alpha}(\alpha)$.

Consider a tree $\tr$ defined by  $\tr(\alpha) = \tr_{k_\alpha}(\alpha)$. 
It is easy to check that $\dom(\tr) = \bigcup_{n \in \N} \dom(\tr_n)$. 
We now check that, for all $n \in \N$, $\tr_n \TreeOrd \tr$. 
For all $\alpha \in \dom(\tr_n)$, we have $\alpha \in \dom(\tr)$ and we distinguish two cases:
\begin{itemize}
\item if $\tr_n(\alpha) = \eval{\conf}{\Unknown}$, then, since either $\tr_n \TreeOrd \tr_{k_\alpha}$ or $\tr_{k_\alpha}\TreeOrd\tr_n$ and $\alpha\in\dom(\tr_{k_\alpha})$, we get $\ConfSet(\tr(\alpha))=\ConfSet(\tr_{k_\alpha}(\alpha)) = \ConfSet(\tr_n(\alpha)) = \conf$; 
\item if $\tr_n(\alpha) = \eval{\conf}{\res}$, then $k_\alpha \le n$, hence, since $\tr_{k_\alpha} \TreeOrd \tr_n$, we get $\ConfSet(\tr(\alpha)) = \ConfSet(\tr_{k_\alpha}(\alpha))  = \ConfSet(\tr_n(\alpha)) = \conf$, thus we have only to check that $\subtr{\tr_n}{\alpha} = \subtr{\tr}{\alpha}$. 
To prove this point, 
consider $\beta \in \dom(\subtr{\tr}{\alpha})$, 
then, by \cref{cor:complete-tree}, 
we have $\subtr{\tr}{\alpha}(\beta) = \tr(\alpha\beta) = \eval{\conf'}{\res'}$, 
hence, since for all $h\ge k_\alpha$ we have $\subtr{\tr_{k_\alpha}}{\alpha} = \subtr{\tr_h}{\alpha}$, we get $\alpha\beta \in \dom(\tr_h)$ and $\subtr{\tr_h}{\alpha}$ is complete, thus $k_{\alpha\beta} \le k_\alpha$. 
Therefore, $\tr_{k_{\alpha\beta}}\TreeOrd \tr_{k_\alpha} \TreeOrd \tr_n$ and so we get $\subtr{\tr_{k_{\alpha\beta}}}{\alpha\beta} = \subtr{\tr_n}{\alpha\beta}$, which implies that 
$\subtr{\tr_n}{\alpha}(\beta) = \tr_{k_{\alpha\beta}}(\alpha\beta) = \tr(\alpha\beta)$, as needed. 
\end{itemize}
This proves that $\tr$ is an upper bound of the sequence, we have still to prove that it is the least one. 
To this end, let $\tr'$ be an upper bound of the sequence: we have to show that $\tr \TreeOrd \tr'$. 
Since $\tr'$ is an upper bound, for all $n \in \N$ we have $\dom(\tr_n) \subseteq \dom(\tr')$, hence $\dom(\tr) \subseteq \dom(\tr')$, and, especially, for all $\alpha \in \List{\NPos}$ we have $\tr_{k_\alpha} \TreeOrd \tr'$. 
Hence, for all $\alpha \in \dom(\tr)$, we have $\ConfSet(\tr(\alpha)) = \ConfSet(\tr_{k_\alpha}(\alpha)) = \ConfSet(\tr'(\alpha))$, and, if $\ResSet_\Unknown(\tr(\alpha)) = \res$, since $\tr_{k_\alpha} \TreeOrd \tr$ and $\tr_{k_\alpha} \TreeOrd \tr'$, we have $\subtr{\tr_{k_\alpha}}{\alpha}  = \subtr{\tr}{\alpha}$ and $\subtr{\tr_{k_\alpha}}{\alpha} =  \subtr{\tr'}{\alpha}$, hence $\subtr{\tr}{\alpha} = \subtr{\tr'}{\alpha}$, as needed.  
\end{proof}

Obviously, this relation restricts to partial evaluation trees and, more importantly, the set of partial evaluation trees is closed  with respect to least upper bound for $\TreeOrd$, as the next proposition shows. 

\begin{proposition}\label{prop:pet-lub}
For each increasing sequence $(\tr_n)_{n \in \N}$ of partial evaluation trees, the least upper bound $\Treelub \tr_n$ is a partial evaluation tree as well. 
\end{proposition}
\begin{proof}
Set $\tr = \Treelub \tr_n$. 
We have to show that for every node $\alpha \in \dom(\tr)$ there is a rule in $\UnRuleSet$ with conclusion $\tr(\alpha)$ and premises the (labels of) the children of $\alpha$ in $\tr$. 

Recall from \refItem{prop:tree-ord}{2} that $\tr(\alpha) = \tr_{k_\alpha}(\alpha)$, where $k_\alpha \in \N$ is the least index $n$ where $\tr_n(\alpha)$ is most defined. 
Note that, for all $\alpha \in \dom(\tr)$, $\trbr_\tr(\alpha)$ is finite.  
Indeed, 
by definition of $\tr$ and since the sequence is increasing, we have $\trbr_\tr(\alpha) = \sup\{ \trbr_{\tr_n}(\alpha) \mid n\ge k_\alpha  \}$, and $\trbr_{\tr_n}(\alpha)$ is the number of premises of a rule, for all $n\ge k_\alpha$; 
all such rules have the same configuration in the conclusion $\ConfSet(\tr_n(\alpha)) = \ConfSet(\tr(\alpha))$, hence, 
by condition (\BP) in \cref{def:bss}, there is $b \in \N$ such that $\trbr_{\tr_n}(\alpha) \le b$, thus $\trbr_\tr(\alpha) \le b$. 
Then, 
the set $K = \{k_\alpha\}\cup\{ k_{\alpha i}\mid i\in 1..\trbr_\tr(\alpha)\}$ is finite and 
$n = \max K$ is finite, hence, as $n \ge k_\alpha$ and $n\ge k_{\alpha i}$, for all $i \in 1..\trbr_\tr(\alpha)$, we have 
$\tr_n(\alpha) = \tr(\alpha)$ and $\tr_n(\alpha i) = \tr(\alpha i)$, for all $i\in 1..\trbr_\tr(\alpha)$. 
Therefore, $\bsrule{\tr(\alpha 1)\ldots\tr(\alpha \trbr_\tr(\alpha))}{\tr(\alpha)} = \bsrule{\tr_n(\alpha 1)\ldots\tr_n(\alpha\trbr_\tr(\alpha))}{\tr_n(\alpha)} \in \UnRuleSet$, since $\tr_n$ is a partial evaluation tree. 
Thus, by \cref{def:pet}, $\tr$ is a partial evaluation tree. 
\end{proof}

As already mentioned, finite partial evaluation trees model possibly incomplete evaluations. 
Then, the relation $\TreeOrd$ models refinement of the evaluation, because if $\tr\TreeOrd\tr'$, where $\tr$ and $\tr'$ are finite partial evaluation trees, 
$\tr'$ is ``more detailed'' than $\tr$. 
In a sense, $\TreeOrd$ on finite partial  evaluation trees abstracts the process of evaluation itself, as we will make precise in the next section. 

What about infinite trees? 
Similarly to what we have discussed in the introduction, 
there are many infinite partial evaluation trees which are difficult to interpret. 
For instance, using rules in \cref{fig:bss-ex-lambda} and \cref{fig:bss-ex-lambda-unknown}, we can construct the following infinite tree for all $\val_\Unknown$, where $\Omega = (\Lam{\x}\x\appop\x)\appop(\Lam{\x}\x\appop\x)$: \label{page:Omega}
\[
\Rule{
  \Rule{}{\eval{\Lam{\x}\x\appop\x}{\Lam{\x}\x\appop\x} }\Space
  \Rule{}{\eval{\Lam{\x}\x\appop\x}{\Lam{\x}\x\appop\x} }\Space
  \Rule{ \vdots  }{ \eval{\Omega = \subst{(\x\appop\x)}{\Lam{\x}\x\appop\x}{\x}}{\val_\Unknown} } 
}{ \eval{\Omega}{\val_\Unknown} }
\]
Among all such trees there are some ``good'' ones, we call them \emph{well-formed}. 
Well-formed infinite partial evaluation trees arise as limits of strictly increasing sequences of finite partial evaluation trees, 
hence, in a sense, they model the limit of the evaluation process. namely, non-termination. 

\begin{definition}\label{def:pet-wf}
An infinite partial evaluation tree  $\tr$ is \emph{well-formed} if, 
for all $\alpha \in \dom(\tr)$, if $\ResSet_\Unknown(\tr(\alpha)) \in \ResSet$, then $\subtr{\tr}{\alpha}$ is finite. 
\end{definition}

In other words, in a well-formed partial evaluation tree all complete subtrees are finite. 
The next proposition, together with \cref{prop:tree-unknown}, implies that 
a well-formed tree  
contains a unique infinite path, which is  entirely labelled by incomplete judgments. 
A similar property on infinite derivations will be enforced by corules in the semantics for divergence in \cref{sect:bss-div}. 

\begin{proposition}\label{prop:pet-wf-unknown}
If $\tr$ is a well-formed infinite partial evaluation tree then, 
for all $n \in \N$, there is $\alpha \in \dom(\tr)$ such that $\Len{\alpha} = n$ and $\ResSet_\Unknown(\tr(\alpha)) = \Unknown$. 
\end{proposition}
\begin{proof}
The proof is by induction on $n$. 
For $n=0$, we have $\ResSet_\Unknown(\rt(\tr)) = \Unknown$, since, otherwise, we would have $\UnResSet(\rt(\tr)) = \res$, hence, 
by \cref{def:pet-wf}, $\tr = \subtr{\tr}{\EList}$ would be finite, while $\tr$ is infinite by hypothesis. 

For $n = k+1$, by induction hypothesis, we know there is $\alpha \in \dom(\tr)$ such that $\Len{\alpha} = k$ and $\UnResSet(\tr(\alpha)) = \Unknown$. 
For all $\beta\in\dom(\tr)$ with $\beta=\alpha'h$, $\Len{\alpha'} = k$, $\alpha'\ne \alpha$, we have $\UnResSet(\tr(\beta)) \in\ResSet$, because, if $\UnResSet(\tr(\beta)) = \Unknown$, then also $\UnResSet(\tr(\alpha')) = \Unknown$, by \cref{prop:tree-unknown},  and, again by \cref{prop:tree-unknown}, this implies $\alpha' = \alpha$, which is absurd. 
As a consequence, for all such $\beta$, we have that $\subtr{\tr}{\beta}$ is finite, as $\tr$ is well-formed. 

Then, we focus on children of $\alpha$, splitting cases over $\trbr_\tr(\alpha)$. 
If $\trbr_\tr(\alpha) = 0$, then $\alpha$ has no children and so $\tr$ is finite, which is absurd. 
If $h = \trbr_\tr(\alpha) > 0$, then, if $\UnResSet(\tr(\alpha h)) \in \ResSet$, since $\tr$ is a partial evaluation tree, we get $\UnResSet(\tr(\alpha h')) \in \ResSet$ for all $h'\le h$, hence $\tr$ is again finite, which is absurd. 
Therefore, $\UnResSet(\tr(\alpha h)) = \Unknown$, as needed. 
\end{proof} 

The following result shows that well-formed partial evaluation trees are exactly the least upper bounds of strictly increasing sequences of finite partial evaluation trees. 

\begin{proposition} \label{prop:pt-seq}
The following properties hold:
\begin{enumerate}
\item\label{prop:pt-seq:2} for each strictly increasing sequence $(\tr_n)_{n \in \N}$ of finite partial evaluation trees, the least upper bound $\Treelub \tr_n$ is infinite and well-formed;
\item\label{prop:pt-seq:3} for each well-formed infinite partial evaluation tree $\tr$, there is a strictly increasing sequence $(\tr_n)_{n \in \N}$ of finite partial evaluation trees such that $\tr = \Treelub \tr_n$.
\end{enumerate}
\end{proposition}
\begin{proof}
To prove \cref{prop:pt-seq:2}, set $\tr = \Treelub \tr_n$, then, by \cref{prop:pet-lub}, we have that $\tr$ is a partial evaluation tree, hence we have only to check that it is infinite and well-formed. 

Note that $\tr$ is infinite if and only if $\dom(\tr) = \bigcup_{n\in \N} \dom(\tr_n)$ is infinite. 
To prove this, it suffices to observe that, 
for all $n \in \N$, there is $h > n$ such that $\dom(\tr_n) \subset \dom(\tr_h)$, namely, there is $\alpha \in \dom(\tr_h)$ such that $\alpha\notin\dom(\tr_n)$. 
This can be proved by induction on the number of $\Unknown$ in $\tr_n$, denoted by $N_\Unknown(\tr_n)$, which is finite as $\tr_n$ is finite. 
This follows because,  
if $\dom(\tr_n) = \dom(\tr_{n+1})$, we have $N_\Unknown(\tr_{n+1}) < N_\Unknown(\tr_n)$, since $\tr_n \StrTreeOrd \tr_{n+1}$ implies that there is at least one node $\alpha \in \dom(\tr_n)$ such that $\ResSet_\Unknown(\tr_n(\alpha)) = \Unknown$ and $\ResSet_\Unknown(\tr_{n+1}(\alpha)) = \res$, thus we can apply the induction hypothesis. 

To show that $\tr$ is well-formed, first recall that, for all $\alpha \in \dom(\tr)$, 
we have $\tr(\alpha) = \tr_n(\alpha)$ for some $n\in\N$. 
Then, for all $\alpha \in \dom(\tr)$ such that $\tr(\alpha) = \eval{\conf}{\res}$, since $\tr_n \TreeOrd \tr$ and $\tr_n(\alpha) = \tr(\alpha)$, by definition of $\TreeOrd$, we get  $\subtr{\tr_n}{\alpha} = \subtr{\tr}{\alpha}$; 
hence,  $\subtr{\tr}{\alpha}$ is finite and so $\tr$ is well-formed. 

To prove \cref{prop:pt-seq:3}, for all $n \in \N$, consider the partial evaluation tree $\tr_n$ obtained by ``cutting'' $\tr$ at level $n$ and  defined as follows. 
Let $\alpha_n \in \dom(\tr)$ be the node such that $\Len{\alpha_n} = n$ and $\UnResSet(\tr(\alpha_n)) = \Unknown$ (which exists by \cref{prop:pet-wf-unknown} as $\tr$ is well-formed and it is unique thanks to \refItem{prop:tree-unknown}{2}), 
then define $\tr_n(\beta) = \tr(\beta)$ for all $\beta \ne \alpha_n \beta'$, with $\beta' \in \NEList{\NPos}$, and undefined otherwise.  
We have $\tr_n \TreeOrd \tr_{n+1}$, since, by \refItem{prop:tree-unknown}{1}, $\alpha_{n+1} = \alpha_n i$ for some $i \in \NPos$. 
Finally, by construction, we have $\tr = \Treelub \tr_n$, as needed. 
\end{proof}

This important result will be used in the next sections to prove correctness of extended big-step semantics explicitly modelling divergence.

\subsection{The transition relation}
\label{sect:bss-ts}

As already mentioned, finite partial evaluation trees nicely model intermediate states in the evaluation process of a configuration. 
We now make this precise by defining a transition relation $\ruleAr$ between them,
such that, starting from the initial partial evaluation tree $\Rule{}{\eval{\conf}{\Unknown}}$, we derive a sequence where, intuitively, at each step we detail the evaluation. 
In this way, a sequence ending with a complete tree (a tree containing no $\Unknown$)  models succesfully terminating computation, whereas an infinite sequence (tending to an infinite partial evaluation tree) models divergence, and a sequence reaching an incomplete tree which cannot further move models a stuck computation. 

The one-step transition relation $\ruleAr$  is inductively defined by the rules in \cref{fig:bss-?-rules}.  
To make the definition clearer, we explicitly annotate the tree with the rule in $\UnRuleSet$ applied to derive the root of the tree from its children. 
In the figure,  $\#\rho$ denotes the number of premises of $\rho$. 
Finally, $\sim_i$ is the \emph{equality up-to an index} of rules, defined below: 

\begin{definition}\label{def:rule-eq}
Let $\rho = \inlinerule{\judg_1\ldots\judg_n}{\conf}{\res}$ and $\rho' = \inlinerule{\judg'_1\ldots \judg'_m}{\conf'}{\res'}$ be rules in $\RuleSet$.
Then, for any index $i \in 1..\min (n, m)$, define $\rho \sim_i \rho'$ if and only if  
\begin{itemize}
\item $\conf = \conf'$, 
\item for all $k < i$, $\judg_k = \judg'_k$, and
\item $\ConfSet(\judg_i) = \ConfSet(\judg'_i)$. 
\end{itemize}
\end{definition}
In other words, this equivalence models the fact that rules $\rho$ and $\rho'$ represent the same computation until the $i$-th configuration included. 

\begin{figure}
\begin{small}
\begin{math}
\begin{array}{l}
\rn{tr-1}\Space 
\MetaRuleTree{\startrule{\conf}}{\Rule{}{\eval{\conf}{\Unknown}}}{\rho}{\Rule{}{\eval{\conf}{\res}}}{
\#\rho = 0 \\
\ConfSet(\rho) = \conf \\
\ResSet(\rho) = \res
} \\[3ex]
\rn{tr-2}\Space 
\MetaRuleTree{\startrule{\conf}}{\Rule{}{\eval{\conf}{\Unknown}}}{\partialrule{\rho}{1}{\Unknown}}{\Rule{\eval{\conf'}{\Unknown}}{\eval{\conf}{\Unknown}}}{
\#\rho > 0 \\
\ConfSet(\rho) = \conf \\
\ConfSet(\rho, 1) = \conf'
} \\[3ex]
\rn{tr-3}\Space 
\MetaRuleTree{\partialrule{\rho}{i}{\res}}{\Rule{\tr_1\ \ldots\ \tr_i}{\eval{\conf}{\Unknown}}}{\rho'}{\Rule{\tr_1\ \ldots\ \tr_i}{\eval{\conf}{\res'}}}{
\rho' \sim_i \rho  \\
\ResSet(\rho',i) = \res \\
\#\rho' = i \\
\ResSet(\rho') = \res'
}\\[3ex]
\rn{tr-4}\Space 
\MetaRuleTree{\partialrule{\rho}{i}{\res}}{\Rule{\tr_1\ \ldots\ \tr_i}{\eval{\conf}{\Unknown}}}{\partialrule{\rho'}{i+1}{\Unknown}}{\Rule{\tr_1\ \ldots \ \tr_i\ \eval{\conf'}{\Unknown}}{\eval{\conf}{\Unknown}}}{
\rho' \sim_i \rho \\
\ResSet(\rho',i) = \res \\
\ConfSet(\rho', i+1) = \conf' 
} \\[3ex]
\rn{tr-5}\Space 
\MetaRuleTree{\partialrule{\rho}{i}{\Unknown}}{\Rule{\tr_1\ \ldots\ \tr_{i-1}\ \tr_i}{\eval{\conf}{\Unknown}}}{\partialrule{\rho}{i}{\unres}}{\Rule{\tr_1\ \ldots\ \tr_{i-1}\ \tr'_i}{\eval{\conf}{\Unknown}}}{
\tr_i \ruleAr \tr'_i 
} 
\end{array}
\end{math}
\end{small}
\caption{Transition relation between partial evaluation trees.}\label{fig:bss-?-rules}
\end{figure}

Intuitively, each transition step makes ``less incomplete'' the partial evaluation tree. 
Notably, transition rules apply only to nodes labelled by incomplete judgements ($\eval{\conf}{\Unknown}$), whereas subtrees whose root is a complete judgement ($\eval{\conf}{\res}$) cannot move. 
In detail:
\begin{itemize}
\item If the last applied rule is $\startrule{\conf}$, 
we have to find a rule $\rho$ with $\conf$ in the conclusion and, if it has no premises we just return $\ResSet(\rho)$ as result, 
otherwise we  start evaluating the first premise of such rule.
\item If the last applied rule is $\partialrule{\rho}{i}{\res}$, then all subtrees are complete, hence, to continue the evaluation, we have to find another rule $\rho'$, having, for each $k \in 1..i$, as $k$-th premise the root of $\tr_k$. 
Then there are two possibilities: if there is an $i+1$-th premise, we start evaluating it, otherwise, we return $\ResSet(\rho')$ as result.
\item If the last applied rule is a propagation rule $\partialrule{\rho}{i}{\Unknown}$, then we simply propagate the step made by $\tr_i$ (the last subtree), which is necessarily incomplete. 
After the step, $\tr'_i$ may be complete, hence the last applied rule is $\partialrule{\rho}{i}{\unres}$. 
\end{itemize}
In \cref{fig:ex-pev-reduction} we report an example of evaluation of a term according to rules in \cref{fig:bss-ex-lambda}, using partial evaluation trees and $\ruleAr$.  

\begin{figure}
\begin{small}
\begin{math}
\begin{array}{l}
\Infer{ }{\eval{(\Lam{\x}\x)\appop n}{\Unknown}}{}  
\ruleAr \Infer{\eval{\Lam{\x}\x}{\Unknown}}{\eval{(\Lam{\x}\x)\appop n}{\Unknown}}{}
 \ruleAr \Infer{\eval{\Lam{\x}\x}{\Lam{\x}\x}}{\eval{(\Lam{\x}\x)\appop n}{\Unknown}}{} \\[2ex] \BigSpace
 \ruleAr \Infer{\eval{\Lam{\x}\x}{\Lam{\x}\x}\quad \eval n \Unknown}{\eval{(\Lam{\x}\x)\appop n}{\Unknown}}{} 
\ruleAr \Infer{\eval{\Lam{\x}\x}{\Lam{\x}\x}\quad \eval n n}{\eval{(\Lam{\x}\x)\appop n}{\Unknown}}{} \\[2ex] \BigSpace
 \ruleAr \Infer{\eval{\Lam{\x}\x}{\Lam{\x}\x}\quad \eval n n \quad \eval n \Unknown}{\eval{(\Lam{\x}\x)\appop n}{\Unknown}}{} 
\ruleAr \Infer{\eval{\Lam{\x}\x}{\Lam{\x}\x}\quad \eval n n \quad \eval n n}{\eval{(\Lam{\x}\x)\appop n}{\Unknown}}{} \\[2ex] \BigSpace
 \ruleAr \Infer{\eval{\Lam{\x}\x}{\Lam{\x}\x}\quad \eval n n \quad \eval n n}{\eval{(\Lam{\x}\x)\appop n}{n}}{}
\end{array}
\end{math} 
\end{small}
\caption{The evaluation  of $(\Lam{\x}\x) \appop n$ using $\ruleAr$ for rules in \cref{fig:bss-ex-lambda}. } \label{fig:ex-pev-reduction} 
\end{figure}

As mentioned above, the definition of $\ruleAr$ given in \cref{fig:bss-?-rules} nicely models as a transition system an interpreter driven by the big-step rules. 
In other words, the one-step transition relation between finite partial evaluation trees  specifies an algorithm of incremental evaluation.\footnote{Non-determinism can only be caused by intrinsic non-determinism of the big-step semantics, if any.} 
On the other hand, also the partial order relation $\TreeOrd$ (\cf \cref{def:tree-ord}) models a refinement relation between finite partial evauation trees, even if in a more abstract way. 
The next proposition formally proves that these two descriptions agree, namely, $\TreeOrd$ is indeed an abstraction of $\ruleAr$. 

\begin{proposition} \label{prop:tree-ord-arrow}
Let $\tr$ and $\tr'$ be finite partial evaluation trees, then the following hold:
\begin{enumerate}
\item\label{prop:tree-ord-arrow:1} if $\tr\ruleAr \tr'$ then $\tr \StrTreeOrd \tr'$, and
\item\label{prop:tree-ord-arrow:2} if $\tr \TreeOrd \tr'$ then $\tr \ruleArStar \tr'$.
\end{enumerate}
\end{proposition}
\begin{proof}
Point 1 can be easily proved by induction on the definition of $\ruleAr$. 
The proof of point 2 is by induction on $\tr'$, denote by $IH$ the induction hypothesis. 
This is possible as $\tr'$ is finite by hypothesis. 
We can assume $\ResSet_\Unknown(\rt(\tr)) = \Unknown$, since in the other case, by \refItem{prop:tree-ord}{1}, we have $\tr = \tr'$, hence the thesis is trivial. 
We can further assume $\ResSet_\Unknown(\rt(\tr')) = \Unknown$, since, if $\tr' = \Rule{\tr'_1\Space\ldots\Space\tr'_n}{\eval{\conf}{\res}}$, then we always have $\tr'' = \Rule{\tr'_1\Space\ldots\Space\tr'_n}{\eval{\conf}{\Unknown}} \ruleAr \tr'$ and $\tr \TreeOrd \tr''$, 
because $\tr \TreeOrd \tr'$ and we have $\dom(\tr') = \dom(\tr'')$, $\tr'(\alpha) = \tr''(\alpha)$ for all $\alpha\ne\EList$, $\ConfSet(\rt(\tr')) = \ConfSet(\rt(\tr''))$ and $\UnResSet(\rt(\tr)) = \Unknown$. 
Now, if $\tr' = \Rule{}{\eval{\conf}{\Unknown}}$ (base case), then, since $\dom(\tr) \subseteq \dom(\tr')$ and $\ConfSet(\rt(\tr)) = \ConfSet(\rt(\tr'))$ by definition of $\TreeOrd$, we have $\tr = \tr'$, hence the thesis is trivial. 

Let us assume $\tr = \Rule{\tr_1\Space\ldots\tr_k}{\eval{\conf}{\Unknown}}$ and $\tr' = \Rule{\tr'_1\Space\ldots\Space\tr'_i}{\eval{\conf'}{\Unknown}}$, with, necessarily,  $k \le i$ and  $\conf = \conf'$  by definition of $\TreeOrd$. 
We have $\tr_h \TreeOrd \tr'_h$, for all $h\le k$, and 
by \refItem{prop:tree-unknown}{2}, at most $\tr_k$ is incomplete, that is, for all $h < k$, $\tr_h$ is complete, namely, $\ResSet_\Unknown(\rt(\tr_h)) \in \ResSet$, thus, by definition of $\TreeOrd$, we have $\tr_h = \tr'_h$. 
Furthermore, since $\tr_k \TreeOrd \tr'_k$, by $IH$, 
we get $\tr_k \ruleArStar \tr'_k$, hence \mbox{$\tr \ruleArStar \tr'' = \Rule{\tr'_1\Space\ldots\Space\tr'_k}{\eval{\conf}{\Unknown}} \TreeOrd \tr'$.}
We now show, concluding the proof, by arithmetic induction on $i-k$, that $\tr''\ruleArStar \tr'$. 
If $i-k = 0$, hence $i = k$, we have $\tr'' = \tr'$, hence the thesis is immediate. 
If $i-k > 0$, hence $i > k$, setting $\conf'' = \ConfSet(\rt(\tr'_{k+1}))$, by $IH$, we get $\Rule{}{\eval{\conf''}{\Unknown}}\ruleArStar \tr'_{k+1}$; 
moreover, again by \refItem{prop:tree-unknown}{2}, we have $\ResSet_\Unknown(\rt(\tr'_k)) \in \ResSet$, hence we get 
\[ \tr'' \ruleAr \Rule{\tr'_1\Space\ldots\Space\tr'_k\Space\eval{\conf''}{\Unknown}}{\eval{\conf}{\Unknown}} \ruleArStar \Rule{\tr'_1\Space\ldots\Space \tr'_k\Space\tr'_{k+1}}{\eval{\conf}{\Unknown}} =  \hat{\tr} \]
Finally, by arithmetic induction hypothesis, we get $\hat{\tr} \ruleArStar \tr'$, as needed. 
\end{proof}

We conclude this section by showing that the transition relation $\ruleAr$ agrees with the semantic relation (inductively) defined by $\RuleSet$, 
namely, the semantic relation captures exactly successful terminating computations in $\ruleAr$.

\begin{theorem} \label{thm:pev-conservative} 
$\validInd{\RuleSet}{\eval{\conf}{\res}}$ iff $\Rule{}{\eval{\conf}{\Unknown}} \ruleArStar \tr$, where $\rt(\tr) = \eval{\conf}{\res}$.
\end{theorem}
\begin{proof}
$\validInd{\RuleSet}{\eval{\conf}{\res}}$ implies $\Rule{}{\eval{\conf}{\Unknown}} \ruleArStar \tr$ where $\rt(\tr) = \eval{\conf}{\res}$. 
By definition, if $\validInd{\RuleSet}{\eval{\conf}{\res}}$ holds, then there is a finite evaluation tree $\tr$ in $\RuleSet$ such that $\rt(\tr) = \eval{\conf}{\res}$. 
Since $\RuleSet \subseteq \UnRuleSet$ by \cref{def:pev-rules},  $\tr$ is a (complete)  partial evaluation tree as well;
furthermore, $\Rule{}{\eval{\conf}{\Unknown}} \TreeOrd \tr$, hence, by \refItem{prop:tree-ord-arrow}{2}, we get the thesis. \\
$\Rule{}{\eval{\conf}{\Unknown}} \ruleArStar \tr$ where $\rt(\tr) = \eval{\conf}{\res}$ implies $\validInd{\RuleSet}{\eval{\conf}{\res}}$. 
Since $\rt(\tr) = \eval{\conf}{\res}$, by \cref{cor:complete-tree}, $\tr$ is complete, hence, 
it is an evaluation tree in $\RuleSet$, thus $\validInd{\RuleSet}{\eval{\conf}{\res}}$ holds.
\end{proof}

  \section{Extended big-step semantics: two constructions} 
\label{sect:constructions}

In \cref{sect:bss-pev}, we have just shown that, given a big-step semantics as in \cref{def:bss}, it is possible to define computations in such semantics, by deriving a transition relation which formally models the evaluation algorithm guided by the rules. 
In this way, we are able to distinguish stuck and non-terminating computations as in standard small-step semantics. 
This, in a sense, shows that such a distinction is \emph{implicit} in a big-step semantics. 

In this section, we aim at showing that we can make such distinction explicit directly by a big-step semantics, without introducing any transition relation modelling single computation steps. 
To this end, we  describe two constructions that, starting from a big-step semantics, yield extended ones where non-terminating and stuck computations are  explicitly distinguished. 
These two constructions are in some sense dual to each other, because one explicitly models non-termination, while the other one explicitly models stuckness, and 
they are based on well-know ideas: divergence is modelled by \emph{traces}, as suggested by \citet{LeroyG09}, while stuckness by an additional special result, as described, for instance, by \citet{Pierce02}. 
The novel contribution is that, thanks to the general definition of big-step semantics in \cref{sect:bss-def} (\cf \cref{def:bss}), we can provide \emph{general} constructions working on an arbitrary big-step semantics, rather than discussing specific examples, as it is customary in the literature. 

In the following, we assume a big-step semantics $\Triple{\ConfSet}{\ResSet}{\RuleSet}$.

\subsection{Adding traces}
\label{sect:bss-traces}

The set of \emph{traces} in the big-step semantics is the set $\FIList{\ConfSet}$ of finite and infinite sequences of configurations. 
Finite traces are ranged over by $\bstr$, while infinite traces by $\bstrinf$. 

The judgement of trace semantics has shape $\evaltr{\conf}{\trres}$, where $\trres \in \TrResSet = (\List{\ConfSet}\times\ResSet) + \InfList{\ConfSet}$, that is, 
$\trres$ is either a pair \ple{\bstr,\res} of a finite trace and a result, modelling a converging computation,  or an infinite trace $\bstrinf$, modelling divergence. 
Intuitively, traces $\bstr$ keep track of all the configurations visited during the evaluation, starting from $\conf$ itself. 
To define the trace semantics, we construct, starting from $\RuleSet$, a new set of rules $\TrRuleSet$ as follows: 

\begin{definition}[Rules for traces] \label{def:tr-rules}
The set of rules $\TrRuleSet$ consists of the following rules: 
\begin{description}
\item [finite trace rules] 
For each $\rho = \inlinerule{\judg_1\ldots\judg_n}{\conf}{\res}$ in $\RuleSet$ and finite traces $\bstr_1, \ldots, \bstr_n  \in \List{\ConfSet}$, define rule $\tracerule{\rho}{\bstr_1,\ldots,\bstr_n}$ as 
\[
\Rule{
	\evaltr{\ConfSet(\judg_1)}{\ple{\bstr_1,\ResSet(\judg_1)}} \Space \ldots \Space \evaltr{\ConfSet(\judg_n)}{\ple{\bstr_n,\ResSet(\judg_n)}}
}{ \evaltr{\conf}{\ple{\conf\bstr_1\cdots\bstr_n,\res}}  }
\]
\item [infinite trace rules]  
For each $\rho = \inlinerule{\judg_1 \ldots \judg_n}{\conf}{\res}$ in $\RuleSet$, index $i \in 1..n$, finite traces $\bstr_1, \ldots, \bstr_{i-1} \in \List{\ConfSet}$,  and infinite trace $\bstrinf \in \InfList{\ConfSet}$, define rule $\divtracerule{\rho}{i}{\bstr_1,\ldots,\bstr_{i-1}}{\bstrinf}$ as follows: 
\[
\Rule{
	\evaltr{\ConfSet(\judg_1)}{\ple{\bstr_1,\ResSet(\judg_1)}}
	\Space \ldots \Space
	\evaltr{\ConfSet(\judg_{i-1})}{\ple{\bstr_{i-1},\ResSet(\judg_{i-1})}} 
	\Space  
	\evaltr{\ConfSet(\judg_i)}{\bstrinf}
}{ \evaltr{\conf}{\conf\bstr_1\cdots\bstr_{i-1}\bstrinf} }
\]
\end{description}
\end{definition}
Finite trace rules enrich big-step rules in $\RuleSet$ by finite traces, thus modelling computations converging to a final result. 
On the other hand, infinite trace rules handle non-termination, modelled by infinite traces: 
they propagate divergence, that is, if a configuration in the premises of a rule in $\RuleSet$  diverges, namely, it evaluates to an infinite trace, then the subsequent premises are ignored and the configuration in the conclusion diverges as well.
Note that all these rules have a non-empty trace in the conclusion, hence only non-empty traces are derivable by such rules. 
Finally, observe that the triple \ple{\ConfSet,\TrResSet,\TrRuleSet} is a big-step semantics according to \cref{def:bss}. 

The standard inductive interpretation of big-step rules is not enough in this setting: 
it can only derive judgements of shape $\evaltr{\conf}{\ple{\bstr,\res}}$, because there is no axiom introducing infinite traces, hence they cannot be derived by finite derivations. 
In other words, the inductive interpretation of $\TrRuleSet$ can only capture converging computations. 
To properly handle divergence, we have to interpret rules \emph{coinductively}, namely, allowing both finite and infinite derivations. 
Then, we will write $\validCo{\TrRuleSet}{\evaltr{\conf}{\trres}}$ to say that $\evaltr{\conf}{\trres}$ is coinductively derivable by rules in $\TrRuleSet$. 
It is important to note the following proposition, stating that 
enabling infinite derivations does not affect the semantics of converging computations. 

\begin{lemma}\label{lem:evaltr-conv}
$\validCo{\TrRuleSet}{\evaltr{\conf}{\ple{\bstr,\res}}}$ iff $\validInd{\TrRuleSet}{\evaltr{\conf}{\ple{\bstr,\res}}}$. 
\end{lemma}
\begin{proof}
The right-to-left implication is trivial, because the inductive interpretation is always included in the coinductive one. 
The proof of the other direction is by induction on the length of $\bstr$, which is a finite trace. 
By hypothesis, we know that $\evaltr{\conf}{\ple{\bstr,\res}}$ is derivable by a (possibly infinite) derivation and, by \cref{def:tr-rules}, we know that the last applied rule $\rho^\trlb$ has shape $\tracerule{\rho}{\bstr_1,\ldots,\bstr_n}$, hence $\bstr = \conf\bstr_1\cdots\bstr_n$. 
If $\Len{\bstr} = 1$, then $\bstr = \conf$, and so $n = 0$, that is, $\rho = \inlinerule{\EList}{\conf}{\res}$, because only non-empty traces are derivable, hence $\validInd{\TrRuleSet}{\evaltr{\conf}{\ple{\bstr,\res}}}$ holds by $\rho^\trlb$.  
If $\Len{\bstr} > 0$, then, for all $i\in 1..n$, $\Len{\bstr_i} < \Len{\bstr}$, hence, by induction hypothesis, we get $\validInd{\TrRuleSet}{\evaltr{\ConfSet(\rho,i)}{\ple{\bstr_i,\ResSet(\rho,i)}}}$, and so $\validInd{\TrRuleSet}{\evaltr{\conf}{\ple{\bstr,\res}}}$ holds by $\rho^\trlb$.  
\end{proof}
Note that, following the same inductive strategy  as the above proof, 
we can prove that actually  a derivation for a judgement of shape $\evaltr{\conf}{\ple{\bstr,\res}}$ is necessarily finite. 
This is essentially due to the fact that rules are \emph{productive}, meaning that the trace in the conclusion is always strictly larger than those in the premises. 

We show in \cref{fig:bss-ex-lambda-tr} the rules obtained by applying \cref{def:tr-rules}, starting from meta-rule \rn{app} of the example in \cref{fig:bss-ex-lambda} (for the other meta-rules the outcome is analogous).
\begin{figure}[t]
\begin{math}
\begin{array}{c}
\MetaRule{app-tr}
{\evaltr{\e_1}{\ple{\bstr_1,\Lam{\x}\e}}\Space\evaltr{\e_2}{\ple{\bstr_2,\val_2}}\Space\evaltr{\subst{\e}{\val_2}{\x}}{\ple{\bstr,\val}}}
{ \evaltr{\e_1\appop\e_2}{\ple{(\e_1\appop\e_2) \bstr_1\bstr_2\bstr,\val}} }
{} 
\\[3ex]
\MetaRule{div-app-1}
{\evaltr{\e_1}{\bstrinf}}
{ \evaltr{\e_1\appop\e_2}{(\e_1\appop\e_2)\bstrinf} }
{} 
\Space
\MetaRule{div-app-2}
{\evaltr{\e_1}{\ple{\bstr_1,\Lam{\x}\e}}\Space\evaltr{\e_2}{\bstrinf}}
{ \evaltr{\e_1\appop\e_2}{(\e_1\appop\e_2)\bstr_1\bstrinf}}
{} 
\\[3ex]
\MetaRule{div-app-3}
{\evaltr{\e_1}{\ple{\bstr_1,\Lam{\x}\e}}\Space\evaltr{\e_2}{\ple{\bstr_2,\val_2}}\Space\evaltr{\subst{\e}{\val_2}{\x}}{\bstrinf}}
{ \evaltr{\e_1\appop\e_2}{(\e_1\appop\e_2)\bstr_1\bstr_2\bstrinf}}
{} 
\end{array}
\end{math}
\caption{Trace semantics for application}\label{fig:bss-ex-lambda-tr}
\end{figure}

For instance,  set $\omega = \Lam{\x}\x\appop\x$, hence $\Omega = \omega\appop\omega$ (\cf page \pageref{page:Omega}), and $\bstrinf_\Omega$ the infinite trace $\Omega \omega \omega \Omega \omega \omega \ldots$,
it is easy to see that the judgment $\evaltr{\Omega}{\bstr_\Omega}$ can be derived by the following infinite derivation:\footnote{To help the reader, we add  equivalent expressions with a grey background.}\\
\[
\Rule{
  \Rule{}{\evaltr{\omega}{\ple{\omega,\omega}}}\Space
  \Rule{}{\evaltr{\omega}{\ple{\omega,\omega}}}\Space
  \Rule{\vdots}{
    \evaltr{\meta{\Omega = }\subst{(\x\appop\x)}{\omega}{\x}}{\bstrinf_\Omega}
  }
}{\eval{\Omega}{\Omega \omega \omega \bstrinf_\Omega\meta{= \bstrinf_\Omega}}}
\]

Note that \emph{only} the judgment $\evaltr{\Omega}{\bstrinf_\Omega}$ can be derived, that is, the trace semantics of $\Omega$ is uniquely determined to be $\bstrinf_\Omega$, 
since the infinite derivation forces the equation $\bstrinf_\Omega=\Omega \omega \omega \bstrinf_\Omega$.  

To check that the construction in \cref{def:tr-rules} is a correct extension of the given big-step semantics, 
we have to show it is \emph{conservative}, in the sense that it does not affect the semantics of converging computations, as formally stated below. 

\begin{theorem} \label{thm:tr-conservative}
$\validCo{\TrRuleSet}{\evaltr{\conf}{\ple{\bstr,\res}}}$ for some $\bstr \in \List{\ConfSet}$ iff $\validInd{\RuleSet}{\eval{\conf}{\res}}$.
\end{theorem}
\begin{proof}
By \cref{lem:evaltr-conv}, we know that $\validCo{\TrRuleSet}{\evaltr{\conf}{\ple{\bstr,\res}}}$ iff $\validInd{\TrRuleSet}{\evaltr{\conf}{\ple{\bstr,\res}}}$. 
Then, the thesis follows by proving $\validInd{\TrRuleSet}{\evaltr{\conf}{\ple{\bstr,\res}}}$, for some $\bstr\in\List{\ConfSet}$, iff $\validInd{\RuleSet}{\eval{\conf}{\res}}$, 
by a straightforward induction on rules. 
\end{proof}

We conclude this subsection by showing a coinductive proof principle associated with trace semantics, which allows us to prove that a predicate on configurations ensures the existence of a non-terminating computation.

\begin{lemma} \label{lem:tr-proof}
Let $\Spec \subseteq \ConfSet$ be a set. 
If, for all $\conf \in \Spec$, there are $\rho = \inlinerule{\judg_1\ldots\judg_n}{\conf}{\res} \in \RuleSet$ and $i  \in 1..n$
such that 
\begin{enumerate}
\item\label{lem:tr-proof:1} for all  $k < i$, $\validInd{\RuleSet}{\judg_k}$, and 
\item\label{lem:tr-proof:2}  $\ConfSet(\judg_i) \in \Spec$ 
\end{enumerate}
then, for all $\conf \in \Spec$, there exists $\bstrinf \in \InfList{\ConfSet}$ such that $\validCo{\TrRuleSet}{\evaltr{\conf}{\bstrinf}}$. 
\end{lemma}
\begin{proof}
First of all, for each $\conf \in \Spec$, we construct a trace $\bstrinf_\conf \in \InfList{\ConfSet}$, which will be the candidate trace to prove the thesis. 
By hypothesis (\cref{lem:tr-proof:1}), there is a rule $\rho_\conf = \inlinerule{\judg_1^\conf\ldots\judg_{n_\conf}^\conf}{\conf}{\res_\conf}$ and an index $i_\conf \in 1..n_\conf$ such that, for all $k < i_\conf$, we have $\validInd{\RuleSet}{\judg_k^\conf}$. 
Therefore, by  \cref{thm:tr-conservative}, there are finite traces $\bstr_1^\conf, \ldots, \bstr_{i_\conf-1}^\conf \in \List{\ConfSet}$ such that  for all $k < i_\conf$ we have $\validCo{\TrRuleSet}{\evaltr{\ConfSet(\judg_k^\conf)}{\ple{\bstr_k^\conf,\ResSet(\judg_k^\conf)}}}$,
and, in addition (\cref{lem:tr-proof:2}), we know that $\ConfSet(\judg_{i_\conf}^\conf) \in \Spec$. 
Then, for each $\conf \in \Spec$, we can introduce a variable $X_\conf$ and define an equation $X_\conf = \conf \cdot \bstr_1^\conf \cdot \cdots \cdot \bstr_{i_\conf - 1}^\conf \cdot X_{\ConfSet(\judg_{i_\conf}^\conf)}$. 
\bfd The set of all such equations is a guarded system of equations, which thus has a unique solution, namely, a  function $\fun{s}{\Spec}{\InfList{\ConfSet}}$ such that, for each $\conf \in \Spec$ we have $s(\conf) = \conf \cdot \bstr_1^\conf \cdot \cdots \cdot \bstr_{i_\conf -1}^\conf \cdot s(\ConfSet(\judg_{i_\conf}^\conf))$.\footnote{This argument can be made more precise using coalgebras \citep{Rutten00}, in particular the fact that $\Spec$ and $\InfList{\ConfSet}$ carry, respectively, a coalgebra and a corecursive algebra \citep{CaprettaUV09} structure for the functor $X\mapsto \List{\ConfSet}\times X$. } \efd 

We now have to prove that, for all $\conf \in \Spec$, we have $\validCo{\TrRuleSet}{\evaltr{\conf}{s(\conf)}}$. 
To this end, 
consider the set $\Spec' = \{\Pair{\conf}{s(\conf)} \mid \conf \in \Spec \} \cup \{\Pair{\conf}{\ple{\bstr,\res}} \mid \validCo{\TrRuleSet}{\evaltr{\conf}{\ple{\bstr,\res}}} \}$, 
then the proof is by coinduction. 
Let $\Pair{\conf}{\trres} \in \Spec'$, then we have to find a rule $\bsrule{\judg_1\ldots\judg_n}{\evaltr{\conf}{\trres}}\in\TrRuleSet$ such that, for all $k \in 1..n$, $\Pair{\ConfSet(\judg_k)}{\TrResSet(\judg_k)} \in \Spec'$. 
We have two cases:
\begin{itemize}
\item if $\trres = s(\conf)$, then the needed rule is $\divtracerule{\rho_\conf}{i_\conf}{\bstr_1^\conf,\ldots,\bstr_{i_\conf-1}^\conf}{s(\ConfSet(\judg_{i_\conf}^\conf))}$, and 
\item if $\trres = \ple{\bstr,\res}$, then $\validCo{\TrRuleSet}{\evaltr{\conf}{\ple{\bstr,\res}}}$, by construction of $\Spec'$,  hence $\evaltr{\conf}{\ple{\bstr,\res}}$ is the conclusion of a finite trace rule, where all premises are still derivable, thus in $\Spec'$ by construction.
\end{itemize}
\end{proof}

\subsection{Adding $\Wrong$} \label{sect:bss-wrong}

A well-known technique \citep{AbadiC96,Pierce02}  
to distinguish between stuck and diverging computations, in a sense ``dual'' to the previous one, is to  add  a special result $\Wrong$, so that $\eval{\conf}{\Wrong}$ means that the evaluation of $\conf$ goes stuck.

In this case, defining a general and ``automatic'' version of the construction, starting from an arbitrary big-step semantics $\Triple{\ConfSet}{\ResSet}{\RuleSet}$, is a non-trivial problem. 
Our solution is based on the equivalence on rules defined in \cref{def:rule-eq} (equality up to an index), which allows us to define when $\Wrong$ can be introduced. 

The extended judgement has shape $\eval{\conf}{\wrres}$ where $\wrres \in \WrResSet = \ResSet + \{\Wrong\}$, that is, it is either a result or an error. 
To define the extended semantics, we construct, starting from $\RuleSet$, an extended set of rules $\WrRuleSet$ as follows: 

\begin{definition}[Rules for $\Wrong$]\label{def:wr-rules}
The set of rules $\WrRuleSet$ is obtained by adding to $\RuleSet$ the following rules: 
\begin{description}
\item [wrong configuration rules]  
For each configuration $\conf\in\ConfSet$ such that there is no rule $\rho$ in $\RuleSet$ with $\ConfSet(\rho) = \conf$, define rule $\wrongax{\conf}$ as $\Rule{}{\eval{\conf}{\Wrong}}$. 

\item [wrong result rules]
For each rule $\rho = \inlinerule{\judg_1\ldots\judg_n}{\conf}{\res}$ in $\RuleSet$, index $i \in 1..n$, and result $\res' \in \ResSet$, 
if, for all rules $\rho'$ such that $\rho \sim_i \rho'$, $\ResSet(\rho', i) \ne \res'$, then define rule $\wrongrule{\rho}{i}{\res'}$ as 
\[
\Rule{
	\judg_1\Space  \ldots \Space \judg_{i-1}
	\Space 
	\eval{\ConfSet(\judg_i)}{\res'}
}{ \eval{\conf}{\Wrong} }
\]

\item [$\Wrong$ propagation rules] These rules propagate $\Wrong$ analogously to those for divergence propagation:
For each rule $\rho = \inlinerule{\judg_1\ldots\judg_n}{\conf}{\res}$ in $\RuleSet$ and index $i \in 1..n$, define rule $\proprule{\rho}{i}{\Wrong}$ as 
\[
\Rule{
	\judg_1\Space \ldots\Space  \judg_{i-1}
	\Space
	\eval{\ConfSet(\judg_i)}{\Wrong}
}{ \eval{\conf}{\Wrong} }
\]
\end{description}
\end{definition}
Wrong configurations rules simply say that, if there is no rule for a given configuration, then we can derive $\Wrong$. 
Wrong result rules, instead, 
derive $\Wrong$ whenever the configuration in a premise of a rule reduces to a result which is not admitted in such (and any equivalent) rule. 
We also call these two kinds of rules $\Wrong$ introduction rules, as they introduce $\Wrong$ in the conclusion without having it in the premises. 
Finally, $\Wrong$ propagation rules say that, if a configuration in a premise of some rule in $\RuleSet$ goes wrong, then
the subsequent premises are ignored and the configuration in the conclusion goes wrong as well.
Note that \ple{\ConfSet,\WrResSet,\WrRuleSet} is a big-step semantics according to \cref{def:bss}. 

In this case, the standard inductive interpretation is enough to get the correct semantics, because, intuitively, an error, if any, occurs after a finite number of steps. 
Then, we write $\validInd{\WrRuleSet}{\eval{\conf}{\wrres}}$ when the judgment $\eval{\conf}{\wrres}$ is inductively derivable by rules in $\WrRuleSet$. 

We show in \cref{fig:bss-ex-lambda-wr} the meta-rules for $\Wrong$ introduction and propagation constructed starting from  those for application and successor in \cref{fig:bss-ex-lambda}. 

\begin{figure}
\begin{small}
\begin{math}
\begin{array}{c}
\MetaRule{wrong-app}{\eval{\e_1}{\natconst}}{\eval{\e_1\appop\e_2}{\Wrong}}{}\BigSpace
\MetaRule{wrong-succ}{\eval{\e}{\Lam{\x}\e}}{\eval{\SuccExp \e}{\Wrong}}{}\\[3ex]
\MetaRule{prop-app-1}{\eval{\e_1}{\Wrong}}{ \eval{\e_1\appop\e_2}{\Wrong}}{}\BigSpace
\MetaRule{prop-app-2}{\eval{\e_1}{\Lam{\x}\e}\Space\eval{\e_2}{\Wrong}}{ \eval{\e_1\appop\e_2}{\Wrong}}{}\\[3ex]
\MetaRule{prop-app-3}{\eval{\e_1}{\Lam{\x}\e}\Space\eval{\e_2}{\val_2}\Space\eval{\subst{\e}{\val_2}{\x}}{\Wrong}}{ \eval{\e_1\appop\e_2}{\Wrong}}{}\BigSpace
\MetaRule{prop-succ}{\eval{\e}{\Wrong}}{\eval{\SuccExp \e}{\Wrong}}{}
\end{array}
\end{math} 
\end{small}
\caption{Semantics with $\Wrong$ for application and successor}\label{fig:bss-ex-lambda-wr}
\end{figure}

For instance, rule \rn{wrong-app} is introduced since in the original semantics there is rule \rn{app} with $\e_1\appop\e_2$ in the conclusion and $\e_1$ in the first premise, but 
there is no equivalent rule (that is, with $\e_1\appop\e_2$ in the conclusion and $\e_1$ in the first premise) such that the result in the first premise is $\natconst$.
Intuitively, this means that $\natconst$ is a wrong result for the evaluation of the first argument of an application. 

Like the previous construction, the $\Wrong$ construction is a correct extension of $\RuleSet$, namely, it is conservative. 

\begin{theorem} \label{thm:wr-conservative}
$\validInd{\WrRuleSet}{\eval{\conf}{\res}}$ iff $\validInd{\RuleSet}{\eval{\conf}{\res}}$.
\end{theorem}
\begin{proof}
The right-to-left implication is trivial, as $\RuleSet\subseteq\WrRuleSet$ by \cref{def:wr-rules}. 
The proof of the other direction is by induction on rules in $\WrRuleSet$. 
The only relevant cases are rules in $\RuleSet$, because rules in $\WrRuleSet\setminus \RuleSet$ allow only the derivation of judgements of shape $\eval{\conf}{\Wrong}$.
Hence, the thesis is immediate. 
\end{proof}

\subsection{Correctness of constructions}
\label{sect:bss-correct}

We now prove correctness of the trace and $\Wrong$ constructions, by showing they capture diverging and stuck computations, respectively, as defined by the transition relation $\ruleAr$ introduced in \cref{sect:bss-ts}. 
This provides us a coherence result for our approach. 

First of all, note that both constructions correctly capture converging computations, because, if restricted to such computations, by \cref{thm:tr-conservative,thm:wr-conservative}, the constructions are both equivalent to the original big-step semantics.
Hence, in the following, we focus only on diverging and stuck computations, respectively. 

\paragraph{Correctness of $\TrRuleSet$} 
Given a partial evaluation tree $\tr$, we write $\tr\ruleArInf$ meaning that there is an infinite sequence of $\ruleAr$-steps starting from $\tr$. 
Then, the theorem we want to prove is the following:

\begin{theorem}  \label{thm:eq-tr-pev} 
$\validCo{\TrRuleSet}{\evaltr{\conf}{\bstrinf}}$, for some $\bstrinf\in\InfList{\ConfSet}$, 
iff  $\Rule{}{\eval{\conf}{\Unknown}} \ruleArInf$.
\end{theorem}

To prove this result, we need to relate evaluation trees (a.k.a. derivations) in $\TrRuleSet$ (\cf \cref{def:tr-rules}) to partial evaluation trees in $\RuleSet$ (\cf \cref{def:pet}). 
To this end, we define a function $\fun{u_\Unknown}{\TrResSet}{\UnResSet}$, which essentially forgets traces, as follows: 
$u_\Unknown(\ple{\bstr,\res}) = \res$ and $u_\Unknown(\bstrinf) = \Unknown$. 
We can extend this function to judgements, mapping $\evaltr{\conf}{\trres}$ to $\eval{\conf}{u_\Unknown(\trres)}$, and to rules, mapping 
$\tracerule{\rho}{\bstr_1,\ldots,\bstr_n}$ to $\rho$ and $\divtracerule{\rho}{i}{\bstr_1,\ldots,\bstr_{i-1}}{\bstrinf}$ to $\partialrule{\rho}{i}{\Unknown}$. 
Finally, we get a function 
$\erasetr$ that transforms an evaluation tree $\tr^\trlb$ in $\TrRuleSet$ into a partial evaluation tree, defined by $\erasetr(\tr^\trlb) = u_\Unknown \circ \tr^\trlb$, 
that is, relying on the fact that a tree is a (partial) function, we postcompose $\tr^\trlb$ with $u_\Unknown$; 
in other words, 
this means that  we apply $u_\Unknown$ to all judgements labeling a node in $\tr^\trlb$, thus erasing traces. 
Since $u_\Unknown$ transforms rules in $\TrRuleSet$ into rules in $\UnRuleSet$, $\erasetr(\tr^\trlb)$ is indeed a partial evaluation tree and the following equalities between trees hold: 

\[\begin{split}
\EraseTr{\DecoratedTree{\tracerule{\rho}{\bstr_1,\ldots,\bstr_n}}{\Rule{\tr^\trlb_1\Space\ldots\Space\tr^\trlb_n}{\evaltr{\conf}{\ple{\bstr,\res}}}}} 
    &= \DecoratedTree{\rho}{\Rule{\EraseTr{\tr^\trlb_1}\Space\ldots\Space\EraseTr{\tr^\trlb_n}}{\eval{\conf}{\res}}}\\[2ex]
\EraseTr{\DecoratedTree{\divtracerule{\rho}{i}{\bstr_1,\ldots,\bstr_{i-1}}{\bstrinf}}{\Rule{\tr^\trlb_1\Space\ldots\Space\tr^\trlb_i}{\evaltr{\conf}{\bstrinf'}}}} 
    &= \DecoratedTree{\partialrule{\rho}{i}{\Unknown}}{\Rule{\EraseTr{\tr^\trlb_1}\Space\ldots\Space\EraseTr{\tr^\trlb_i}}{\eval{\conf}{\Unknown}}}
\end{split}\]

Note that, by construction, $\dom(\tr^\trlb) = \dom(\EraseTr{\tr^\trlb})$, hence, 
$\tr^\trlb$ is finite iff  $\EraseTr{\tr^\trlb}$ is finite and   
$\tr^\trlb$ is infinite iff $\EraseTr{\tr^\trlb}$ is infinite. 
Furthermore, since, as we have already observed, $\tr^\trlb$ is finite iff $\rt(\tr^\trlb) = \evaltr{\conf}{\ple{\bstr,\res}}$, we have that 
$\EraseTr{\tr^\trlb}$ is complete iff $\tr^\trlb$ is finite and 
$\EraseTr{\tr^\trlb}$ is well-formed iff $\tr^\trlb$ is infinite (\cf \cref{def:pet-wf}).

\begin{lemma} \label{lem:trace-to-unknown}
If $\validCo{\TrRuleSet}{\evaltr{\conf}{\trres}}$ holds by an infinite evaluation tree $\tr^\trlb$, then there is a sequence $(\tr_n)_{n \in \N}$ such that $\tr_n \ruleAr \tr_{n+1}$ for all $n\in \N$, $\tr_0 = \Rule{}{\eval{\conf}{\Unknown}}$,  and $\Treelub \tr_n = \EraseTr{\tr^\trlb}$.
\end{lemma}
\begin{proof}
Since $\tr^\trlb$ is infinite,  $\EraseTr{\tr^\trlb} = \tr$ is a well-formed infinite partial evaluation treee and, by \refItem{prop:pt-seq}{3}, there is a strictly increasing sequence $(\tr'_n)_{n \in \N}$ of finite partial evaluation trees  such that $\Treelub \tr'_n = \tr$ and $\tr'_0 = \Rule{}{\eval{\conf}{\Unknown}}$. 
By \refItem{prop:tree-ord-arrow}{2},  since for all $n \in \N$ we have $\tr'_n \StrTreeOrd \tr'_{n+1}$, we get $\tr'_n \ruleArStar \tr'_{n+1}$, and, since $\tr'_n \ne \tr'_{n+1}$, this sequence of steps is not empty. 
Hence, we can construct a sequence $(\tr_n)_{n \in \N}$ such that $\tr_0 = \tr'_0 =  \Rule{}{\eval{\conf}{\Unknown}}$, $\tr_n \ruleAr \tr_{n+1}$ and $\Treelub \tr_n = \tr$, as needed. 
\end{proof}

\begin{lemma}\label{lem:unknown-to-trace}
Let $\tr$ be a well-formed infinite partial evaluation tree with $\rt(\tr) = \eval{\conf}{\Unknown}$. 
Then,  $\validCo{\TrRuleSet}{\evaltr{\conf}{\bstrinf}}$ holds  for some $\bstrinf\in\InfList{\ConfSet}$. 
\end{lemma} 
\begin{proof}
The thesis  follows from \cref{lem:tr-proof}, applied to the set $\Spec \subseteq \ConfSet$ defined  as follows: 
$\conf\in\Spec$ iff $\ConfSet(\rt(\tr)) = \conf$, for some infinite well-formed partial evaluation tree $\tr$. 
Let $\conf\in\Spec$, then $\conf = \ConfSet(\rt(\tr))$ and the last applied rule in $\tr$ is $\partialrule{\rho}{i}{\Unknown}$, for some $\rho = \inlinerule{\judg_1\ldots\judg_n}{\conf}{\res}$ in $\RuleSet$. 
Then, we have $\validInd{\RuleSet}{\judg_k}$, for all $k < i$ and $\ConfSet(\judg_i) = \ConfSet(\rt(\subtr{\tr}{i}))$ and $\subtr{\tr}{i}$ is an infinite well-formed partial evaluation tree. 
Therefore, $\ConfSet(\judg_i) \in \Spec$, and so the hypotheses of \cref{lem:tr-proof} are satisfied. 
\end{proof}

\begin{proof}[Proof of \cref{thm:eq-tr-pev}] 
$\validCo{\TrRuleSet}{\evaltr{\conf}{\bstrinf}}$ for some $\bstrinf \in \ConfSet^\omega$ implies $  \Rule{}{\eval{\conf}{\Unknown}} \ruleArInf$. 
Since $\validCo{\TrRuleSet}{\evaltr{\conf}{\bstrinf}}$ holds and $\bstrinf$ is infinite, by (a consequence of) \cref{lem:evaltr-conv}, there is an infinite evaluation tree $\tr^\trlb$ in $\TrRuleSet$ such that $\rt(\tr^\trlb) = \evaltr{\conf}{\bstrinf}$. 
Then, by \cref{lem:trace-to-unknown} we get the thesis.  

$\Rule{}{\eval{\conf}{\Unknown}} \ruleArInf$ implies $ \validCo{\TrRuleSet}{\evaltr{\conf}{\bstrinf}}$ for some $\bstrinf \in \ConfSet^\omega$. 
By definition of $\ruleArInf$, there is an infinite sequence $(\tr_n)_{n \in \N}$ such that $\tr_0 = \Rule{}{\eval{\conf}{\Unknown}}$ and, for all $n \in \N$, $\tr_n \ruleAr \tr_{n+1}$, hence, by \refItem{prop:tree-ord-arrow}{1}, we get $\tr_n \StrTreeOrd \tr_{n+1}$. 
By \refItem{prop:pt-seq}{2}, we have that $\tr = \Treelub \tr_n$ is a well-formed infinite partial evaluation tree, hence we get the thesis by \cref{lem:unknown-to-trace}. 
\end{proof}

\paragraph{Correctness of $\WrRuleSet$} 
We now show that the construction in \cref{sect:bss-wrong} correctly models stuck computation in $\ruleAr$. 

The proof relies on the following lemma. 
We say that a (finite) partial evaluation tree $\tr$  is \emph{irreducible} if there is no $\tr'$ such that $\tr \ruleAr \tr'$, and it is \emph{stuck} if it is irreducible and $\UnResSet(\rt(\tr)) = \Unknown$. 
Note that, by \refItem{prop:tree-ord}{1} and \refItem{prop:tree-ord-arrow}{1}, a complete partial evaluation tree $\tr$  is irreducible. 

\begin{lemma} \label{lem:stuck-to-wrong}
If $\tr$ is a stuck partial evaluation tree with $\rt(\tr) = \eval{\conf}{\Unknown}$, then $\validInd{\WrRuleSet}{\eval{\conf}{\Wrong}}$ holds.
\end{lemma}
\begin{proof}
The proof is by induction on $\tr$, splitting cases on the last applied rule.
\begin{proofcases}
\item [$\startrule{\conf}$] Since $\tr$ is stuck, by definition of $\ruleAr$ (\cf \cref{fig:bss-?-rules} first and second clauses), there is no rule $\rho \in \RuleSet$ such that $\ConfSet(\rho) = \conf$, hence $\validInd{\WrRuleSet}{\eval{\conf}{\Wrong}}$ holds, by applying $\wrongax{\conf}$.
\item [$\partialrule{\rho}{i}{\res}$] Suppose $\rho = \inlinerule{\judg_1\ldots\judg_n}{\conf}{\res'}$ and $i \in 1..n$, by hypothesis, for all $k < i$, $\subtr{\tr}{k}$ is a complete partial evaluation tree of $\judg_k$, hence we know that $\validInd{\RuleSet}{\judg_k}$ holds. 
Since $\tr$ is stuck, by definition of $\ruleAr$ (\cf \cref{fig:bss-?-rules} third and fourth clauses), there is no rule $\rho' \sim_i \rho$ with $\ResSet(\rho', i) = \res$, hence $\wrongrule{\rho}{i}{\res} \in \WrRuleSet$. 
By \cref{thm:wr-conservative} we get $\validInd{\WrRuleSet}{\judg_k}$, for all $k < i$, hence 
applying $\wrongrule{\rho}{i}{\res}$, we get  $\validInd{\WrRuleSet}{\eval{\conf}{\Wrong}}$. 
\item [$\partialrule{\rho}{i}{\Unknown}$] Suppose $\rho = \inlinerule{\judg_1\ldots\judg_n}{\conf}{\res'}$ and $i \in 1..n$, by hypothesis, for all $k < i$, $\subtr{\tr}{k}$ is a complete partial evaluation tree of $\judg_k$, hence we know that $\validInd{\RuleSet}{\judg_k}$ holds. 
Set $\conf_i = \ConfSet(\rho, i)$, then, since $\tr$ is stuck, by definition of $\ruleAr$ (\cf \cref{fig:bss-?-rules} clause \rn{tr-5}), the subtree $\subtr{\tr}{i}$ is stuck as well and $\rt(\subtr{\tr}{i}) = \eval{\conf_i}{\Unknown}$. 
By \cref{thm:wr-conservative}, we get $\validInd{\WrRuleSet}{\judg_k}$, for all $k < i$, and, 
by induction hypothesis, we get $\validInd{\WrRuleSet}{\eval{\conf_i}{\Wrong}}$, hence, applying rule $\proprule{\rho}{i}{\Wrong}$, we get $\validInd{\WrRuleSet}{\eval{\conf}{\Wrong}}$. 
\qedhere 
\end{proofcases}
\end{proof}

\begin{lemma}\label{lem:wrong-to-stuck}
If $\validInd{\WrRuleSet}{\eval{\conf}{\Wrong}}$, then 
there is a stuck partial evaluation tree $\tr$ with $\rt(\tr) = \eval{\conf}{\Unknown}$. 
\end{lemma} 
\begin{proof}
The proof is by induction on rules in $\WrRuleSet$.
It is enough to consider only rules with $\Wrong$ in the conclusion, hence we have the following three cases: 
\begin{proofcases}
\item [\wrongax{\conf}] By \cref{def:wr-rules}, there is no rule $\rho \in \RuleSet$ such that $\ConfSet(\rho) = \conf$, thus $\Rule{}{\eval{\conf}{\Unknown}}$ is stuck.
\item [\wrongrule{\rho}{i}{\res}] By \cref{def:wr-rules}, assuming $\rho \equiv \inlinerule{\judg_1\ldots\judg_n}{\conf}{\res'}$, there is no rule $\rho' \sim_i \rho$ such that $\ResSet(\rho', i) = \res$;
then, by \cref{thm:wr-conservative}, for all $k \le i$, $\validInd{\RuleSet}{\judg_k}$ holds, hence there is a finite and complete partial evaluation tree $\tr_k$ with $\rt(\tr_k) = \judg_k$.
Therefore, applying rule \partialrule{\rho}{i}{\res} to $\tr_1,\ldots,\tr_i$, we get a partial evaluation tree, which is stuck, by definition of $\ruleAr$. 
\item [\proprule{\rho}{i}{\Wrong}] Suppose $\rho = \inlinerule{\judg_1\ldots\judg_n}{\conf}{\res}$ and $\conf_i = \ConfSet(\judg_i)$, then, by induction hypothesis, we get that there is a stuck tree $\tr'$ such that $\rt(\tr') = \eval{\conf_i}{\Unknown}$; 
then, by \cref{thm:wr-conservative}, for all $k<i$, $\validInd{\RuleSet}{\judg_k}$ holds, hence there is a finite and complete partial evaluation tree $\tr_k$ with $\rt(\tr_k) = \judg_k$. 
Therefore, applying \partialrule{\rho}{i}{\Unknown} to $\tr_1, \ldots, \tr_{i-1}, \tr'$, we get a stuck tree. 
\qedhere
\end{proofcases}
\end{proof}

\begin{theorem}   \label{thm:eq-wrong-pev}
$\validInd{\WrRuleSet}{\eval{\conf}{\Wrong}}$
iff $\Rule{}{\eval{\conf}{\Unknown}}\ruleArStar \tr$, where $\tr$ is stuck.
\end{theorem}
\begin{proof}
$\validInd{\WrRuleSet}{\eval{\conf}{\Wrong}}$ implies $ \Rule{}{\eval{\conf}{\Unknown}}\ruleArStar \tr$ where $\tr$ is stuck. 
By \cref{lem:wrong-to-stuck} we get a stuck partial evaluation tree $\tr$ with $\rt(\tr) = \eval{\conf}{\Unknown}$, 
hence the thesis follows by \refItem{prop:tree-ord-arrow}{2}, as we trivially have $\Rule{}{\eval{\conf}{\Unknown}} \TreeOrd \tr$. 

$\Rule{}{\eval{\conf}{\Unknown}}\ruleArStar \tr$ where $\tr$ is stuck implies $\validInd{\WrRuleSet}{\eval{\conf}{\Wrong}}$. 
It follows immediately from \cref{lem:stuck-to-wrong}, since $\rt(\tr) = \eval{\conf}{\Unknown}$ by hypothesis. 
\end{proof}

  \section{Divergence by coaxioms}
\label{sect:bss-div} 

As we have described in \cref{sect:bss-traces}, traces allow us to explicitly model divergence, provided that we interpret rules coinductively: a configuration diverges if it evaluates to an infinite trace. 
However, the resulting semantics is somewhat redundant: traces keep track of all configurations visited during the evaluation, while we are just interested in whether there is a final result or non-termination, and a configuration may evaluate to many different infinite traces, hence divergence is modelled in many ways. 
In this section we show how coaxioms (\cf \cref{def:cois} in \cref{sect:is}) can be succesfully adopted to achieve a more abstract model of divergence, removing this redundancy.  
Basically, we present a systematic definition of the approach discussed by \citet{AnconaDZ@oopsla17}. 

The key idea is to regard divergence just as a special result $\divres$, that, like infinite traces (\cf \cref{def:tr-rules}) and $\Wrong$ (\cf \cref{sect:bss-wrong}), can only be propagated by big-step rules. 
To this end, we define yet another construction, extending a given big-step semantics. 

Let us assume a big-step semantics \ple{\ConfSet,\ResSet,\RuleSet}. 
Then, the extended judgement has shape $\eval{\conf}{\infres}$ where $\infres \in \InfResSet = \ResSet + \{\divres\}$, that is, it is either a result or divergence. 
To define the extended semantics, we construct, starting from $\RuleSet$, a new set of rules $\InfRuleSet$ as follows: 

\begin{definition}[Rules for divergence]\label{def:div-rules}
The set of rules $\InfRuleSet$ is obtained by adding to $\RuleSet$ the following rules: 
\begin{description}
\item [divergence propagation rules] 
For each rule $\rho = \inlinerule{\judg_1\ldots\judg_n}{\conf}{\res}$ in $\RuleSet$ and index $i\in 1..n$, define rule $\proprule{\rho}{i}{\divres}$ as 
\[
\Rule{
  \judg_1\Space\ldots\Space\judg_{i-1} \Space \eval{\ConfSet(\judg_i)}{\divres} 
}{ \eval{\conf}{\divres} }
\]
\end{description}
\end{definition}
These additional rules 
propagate divergence, that is, if a configuration in the premises of a rule in $\RuleSet$  diverges, then the subsequent premises are ignored and the configuration in the conclusion diverges as well.
This is very similar to infinite trace rules, but here we do not need to construct traces to represent divergence. 
Note that the triple \ple{\ConfSet,\InfResSet,\InfRuleSet} is a big-step semantics according to \cref{def:bss}. 

Now the question is: how do we interpret such rules?
The standard inductive interpretation of big-step rules, as for trace semantics,  is not enough in this setting, since there is no axiom introducing $\divres$, 
hence it cannot be derived by finite derivations. 
In other words, the inductive interpretation of $\InfRuleSet$ can only capture converging computations, hence it is equivalent to the inductive interpretation of $\RuleSet$. 
On the other hand, differently from trace semantics, even the coinductive interpretation cannot provide the expected semantics: it allows the derivation of too many judgements. 
For instance, in \cref{fig:bss-ex-lambda-div}, we report the divergence propagation rules obtained starting from meta-rule \rn{app} of the example in \cref{fig:bss-ex-lambda} (for other meta-rules the outcome is analogous); 
then, using these rules (and the original ones in \cref{fig:bss-ex-lambda}), we can build the following infinite derivation for $\Omega$, which is correct for any $\infres\in\InfResSet$. 
\[
\Rule{
  \Rule{}{\eval{\omega}{\omega}}\Space
  \Rule{}{\eval{\omega}{\omega}}\Space
  \Rule{\vdots}{
    \eval{\meta{\Omega = }\subst{(\x\appop\x)}{\omega}{\x}}{\infres}
  }
}{\eval{\Omega}{\infres} }
\]

\begin{figure}
\begin{math}
\begin{array}{c}
\MetaRule{div-app-1}
{ \eval{\e_1}{\divres} }
{ \eval{\e_1\appop\e_2}{\divres} }{}
\BigSpace 
\MetaRule{div-app-2}
{ \eval{\e_1}{\Lam{\x}\e} \Space \eval{\e_2}{\divres} }
{ \eval{\e_1\appop\e_2}{\divres} }{}
\\[3ex]
\MetaRule{div-app-3}
{ \eval{\e_1}{\Lam{\x}\e} \Space \eval{\e_2}{\val_2} \Space \eval{\subst{\e}{\val_2}{\x}}{\divres} } 
{ \eval{\e_1\appop\e_2}{\divres} }{}
\end{array}
\end{math}
\caption{Divergence propagation rules for application}\label{fig:bss-ex-lambda-div}
\end{figure}

Intuitively, we would like to allow infinite derivations only to derive divergence, namely, judgments of shape $\eval{\conf}{\divres}$. 
Inference systems with corules are precisely the tool enabling this kind of refinement. 
That is, in addition to divergence propagation rules, we can add appropriate corules $\coRuleSet$ for divergence, as defined below. 

\begin{definition}[Coaxioms for divergence] \label{def:co-div}
The set of corules $\coRuleSet$ consists of the following coaxioms: 
\begin{description}
\item [coaxioms for divergence]
for each configuration $\conf\in\ConfSet$, define coaxiom $\divcorule{\conf}$ as $\CoAxiom{\eval{\conf}{\divres}}$. 
\end{description}
\end{definition}

As described in \cref{sect:is}, coaxioms impose additional conditions on infinite derivations to be considered correct: 
a judgement $\eval{\conf}{\infres}$ is derivable in \ple{\InfRuleSet,\coRuleSet} iff it has an arbitrary (finite or infinite) derivation in $\InfRuleSet$,  whose nodes all have a finite derivation in $\InfRuleSet\cup\coRuleSet$, that is, using both rules and corules. 
We will write $\validFCo{\InfRuleSet}{\coRuleSet}{\eval{\conf}{\infres}}$ when $\eval{\conf}{\infres}$ is derivable in \ple{\InfRuleSet,\coRuleSet}. 

In the above example, $\validFCo{\InfRuleSet}{\coRuleSet}{\eval{\Omega}{\infres}}$ holds iff $\infres = \infty$, because $\eval{\Omega}{\res}$ has no finite derivation in $\InfRuleSet\cup\coRuleSet$, for any $\res\in\ResSet$. 
\bfd In the case of the trace construction (\cf \cref{sect:bss-traces}), coaxioms are not needed as rules are \emph{productive}, because the trace in the conclusion is always strictly larger than those in the premises, see \cref{def:tr-rules}. \FDComm{tentativo dispiegazione come suggerito da Elena} \efd \EZComm{OK}

To check that the construction in \cref{def:div-rules} and \cref{def:co-div}  is a correct extension of the given big-step semantics, 
as for trace semantics, we have to show it is conservative, in the sense that it does not affect the semantics of converging computations, as formally stated below. 

\begin{theorem}\label{thm:div-conservative}
$\validFCo{\InfRuleSet}{\coRuleSet}{\eval{\conf}{\res}}$ iff $\validInd{\RuleSet}{\eval{\conf}{\res}}$. 
\end{theorem}
\begin{proof}
The right-to-left implication is trivial as $\RuleSet\subseteq\InfRuleSet$ by \cref{def:div-rules}. 
To get the other direction, note that if $\validFCo{\InfRuleSet}{\coRuleSet}{\eval{\conf}{\res}}$ then we have $\validInd{\InfRuleSet\cup\coRuleSet}{\eval{\conf}{\res}}$.
Hence, we prove by induction on rules in $\InfRuleSet\cup\coRuleSet$ that, if $\validInd{\InfRuleSet\cup\coRuleSet}{\eval{\conf}{\res}}$ then $\validInd{\RuleSet}{\eval{\conf}{\res}}$. 
The cases of coaxiom $\divcorule{\conf}$ and divergence propagation $\proprule{\rho}{i}{\divres}$ are both empty, as the conclusion of such rules has shape $\eval{\conf}{\divres}$. 
The only relevant case is that of a rule $\rho \in \RuleSet$, for which the thesis follows immediately. 
\end{proof}

Inference systems with corules come with the bounded coinduction principle (\cf \cref{thm:bcoind}). 
Thanks to such principle, we can define a coinductive proof principle, 
which allows us to prove that a predicate on configurations ensures the existence of a non-terminating computation.

\begin{lemma} \label{lem:div-proof}
Let $\Spec \subseteq \ConfSet$ be a set. 
If, for all $\conf \in \Spec$, there are $\rho = \inlinerule{\judg_1\ldots\judg_n}{\conf}{\res}$ in $\RuleSet$ and $i  \in 1..n$
such that 
\begin{enumerate}
\item\label{lem:div-proof:1} for all  $k < i$, $\validInd{\RuleSet}{\judg_k}$, and 
\item\label{lem:div-proof:2}  $\ConfSet(\judg_i) \in \Spec$ 
\end{enumerate}
then, for all $\conf \in \Spec$, $\validFCo{\InfRuleSet}{\coRuleSet}{\eval{\conf}{\divres}}$. 
\end{lemma}
\begin{proof}
Consider the set $\Spec' = \{\ple{\conf,\divres}\mid \conf\in\Spec\} \cup \{\ple{\conf,\res}\mid \validInd{\RuleSet}{\eval{\conf}{\res}}\}$, 
then the proof is by bounded coinduction (\cf \cref{thm:bcoind}). 
\begin{description}
\item [Boundedness] 
We have to show that, for all $\ple{\conf,\infres}\in\Spec'$, $\validInd{\InfRuleSet\cup\coRuleSet}{\eval{\conf}{\infres}}$ holds. 
This is easy because, if $\infres = \divres$, then this holds by coaxiom $\divcorule{\conf}$, otherwise $\infres \in\ResSet$ and $\validInd{\RuleSet}{\eval{\conf}{\infres}}$, hence this holds since $\RuleSet\subseteq\InfRuleSet\subseteq\InfRuleSet\cup\coRuleSet$. 

\item [Consistency] 
We have to show that, for all $\ple{\conf,\infres}\in\Spec'$, there is a rule $\bsrule{\judg_1\ldots\judg_n}{\eval{\conf}{\infres}} \in \InfRuleSet$ such that, for all $k \in 1..n$, \mbox{$\ple{\ConfSet(\judg_k),\InfResSet(\judg_k)} \in \Spec'$.} 
There are two cases:
\begin{itemize}
\item If $\infres = \divres$, then by hypothesis (\cref{lem:div-proof:1}), we have a rule $\rho = \inlinerule{\judg_1\ldots\judg_n}{\conf}{\res} \in \RuleSet$ and an index $i\in 1..n$ such that, for all $k<i$, $\validInd{\RuleSet}{\judg_k}$ and $\ConfSet(\judg_i)\in \Spec$. 
Then, the needed rule is $\proprule{\rho}{i}{\divres}$. 
\item If $\infres\in\ResSet$, then, by construction of $\Spec'$, we have $\validInd{\RuleSet}{\eval{\conf}{\infres}}$, hence, there is a rule $\rho = \inlinerule{\judg_1\ldots\judg_n}{\conf}{\infres}\in\RuleSet\subseteq\InfRuleSet$, where, for all $k\in 1..n$, $\validInd{\RuleSet}{\judg_k}$ holds, and so $\ple{\ConfSet(\judg_k),\ResSet(\judg_k)}\in\Spec'$. 
\end{itemize}
\end{description}
\end{proof}

The reader may have noticed that most definitions and results in this section are very similar to those provided for trace semantics in \cref{sect:bss-traces}. 
This is not a coincidence, indeed, we now formally prove this semantics is an \emph{abstraction} of trace semantics. 

Intuitively, if we are only interested in modelling convergence or divergence, traces are useless, in the sense that it is only relevant to know whether the trace is infinite or not and, in case it is finite, the final result. 
We can model this intuition by a (surjective) function $\fun{u}{\TrResSet}{\InfResSet}$ simply forgetting traces, that is, 
$u(\ple{\bstr,\res}) = \res$ and $u(\bstrinf) = \divres$, with $\bstr\in\List{\ConfSet}$ and $\bstrinf\in\InfList{\ConfSet}$. 

Then, we aim at proving the following result:

\begin{theorem}\label{thm:eq-div-tr}
$\validFCo{\InfRuleSet}{\coRuleSet}{\eval{\conf}{\infres}}$ iff $\validCo{\TrRuleSet}{\evaltr{\conf}{\trres}}$, for some $\trres$ such that $\infres = u(\trres)$. 
\end{theorem}

In a diagrammatic form, \cref{thm:eq-div-tr} says that the following diagram commutes:
\[\xymatrix{
\wp(\TrResSet) \ar[rr]^{\img{u}} && \wp(\InfResSet) \\
& \ar[ul]^{\sem{-}_{\trlb}} \ConfSet \ar[ur]_{\sem{-}_{\divres}} &
}\]
where $\fun{\img{u}}{\wp(\TrRuleSet)}{\wp(\InfRuleSet)}$ is the direct image of $u$, 
$\fun{\sem{-}_\trlb}{\ConfSet}{\wp(\TrResSet)}$ is defined by $\sem{\conf}_\trlb = \{\trres\in\TrResSet \mid \validCo{\TrRuleSet}{\evaltr{\conf}{\trres}} \}$, and 
$\fun{\sem{-}_{\divres}}{\ConfSet}{\wp(\InfResSet)}$ is defined by $\sem{\conf}_\trlb = \{\infres\in\InfResSet \mid \validFCo{\InfRuleSet}{\coRuleSet}{\eval{\conf}{\infres}} \}$. 

\begin{proof}
The statement can be split in the following two points: 
\begin{enumerate}
\item $\validFCo{\InfRuleSet}{\coRuleSet}{\eval{\conf}{\res}}$ iff $\validCo{\TrRuleSet}{\evaltr{\conf}{\ple{\bstr,\res}}}$, for some $\bstr\in\List{\ConfSet}$, and 
\item $\validFCo{\InfRuleSet}{\coRuleSet}{\eval{\conf}{\divres}}$ iff $\validCo{\TrRuleSet}{\evaltr{\conf}{\bstrinf}}$, for some $\bstrinf\in\InfList{\ConfSet}$.
\end{enumerate}
The first point follows immediately from \cref{thm:tr-conservative} and \cref{thm:div-conservative}, as $\validFCo{\InfRuleSet}{\coRuleSet}{\eval{\conf}{\res}}$ and $\validCo{\TrRuleSet}{\evaltr{\conf}{\ple{\bstr,\res}}}$ are both equivalent to $\validInd{\RuleSet}{\eval{\conf}{\res}}$. 
Then, we have only to prove the second point. 

The left-to-right implication follows applying \cref{lem:tr-proof} to the set $\Spec_\divres = \{ \conf \in \ConfSet \mid \validFCo{\InfRuleSet}{\coRuleSet}{\eval{\conf}{\divres}} \}$.
If $\conf\in\Spec_\divres$, then $\eval{\conf}{\divres}$ is derived by a rule $\proprule{\rho}{i}{\divres}$ for some $\rho = \inlinerule{\judg_1\ldots\judg_n}{\conf}{\res}$ in $\RuleSet$  and $i\in 1..n$, 
hence we have $\validFCo{\InfRuleSet}{\coRuleSet}{\judg_k}$, which implies $\validInd{\RuleSet}{\judg_k}$ by \cref{thm:div-conservative}, for all $k < i$, and 
$\validFCo{\InfRuleSet}{\coRuleSet}{\eval{\ConfSet(\judg_i)}{\divres}}$, that is, $\ConfSet(\judg_i)\in\Spec_\divres$, because these judgements are the premises of $\proprule{\rho}{i}{\divres}$. 
Therefore, the hypotheses of \cref{lem:tr-proof} are satisfied and we get, for all $\conf\in\Spec_\divres$,  $\validCo{\TrRuleSet}{\evaltr{\conf}{\bstrinf_\conf}}$, for some $\bstrinf_\conf\in\InfList{\ConfSet}$, hence $u(\bstrinf_\conf) = \divres$. 

Similarly, the right-to-left implication follows applying \cref{lem:div-proof} to the set $\Spec_\trlb = \{\conf\in\ConfSet \mid \validCo{\TrRuleSet}{\evaltr{\conf}{\bstrinf}}\mbox{ for some }\bstrinf\in\InfList{\ConfSet}\}$. 
If $\conf\in\Spec_\trlb$, then, for some $\bstrinf\in\InfList{\ConfSet}$, $\evaltr{\conf}{\bstrinf}$ is derived by a rule $\divtracerule{\rho}{i}{\bstr_1,\ldots,\bstr_{i-1}}{\bstrinf'}$, for some $\rho = \inlinerule{\judg_1\ldots\judg_n}{\conf}{\res}$ in $\RuleSet$  and $i\in 1..n$, 
hence we have \mbox{$\validCo{\TrRuleSet}{\evaltr{\ConfSet(\judg_k)}{\ple{\bstr_k,\ResSet(\judg_k)}}}$,} which implies $\validInd{\RuleSet}{\judg_k}$ by \cref{thm:tr-conservative}, for all $k < i$, and 
$\validCo{\TrRuleSet}{\evaltr{\ConfSet(\judg_i)}{\bstrinf'}}$, that is, $\ConfSet(\judg_i)\in\Spec_\trlb$, because these judgements are the premises of the rule $\divtracerule{\rho}{i}{\bstr_1,\ldots,\bstr_{i-1}}{\bstrinf'}$. 
Therefore, the hypotheses of \cref{lem:div-proof} are satisfied and we get, for all $\conf\in\Spec_\trlb$,  $\validFCo{\InfRuleSet}{\coRuleSet}{\eval{\conf}{\divres}}$. 
\end{proof}

As an immediate consequence of \cref{thm:eq-div-tr} and \cref{thm:eq-tr-pev}, we get the following corollary, stating that  the construction given by \cref{def:div-rules,def:co-div} correctly models diverging computations: 

\begin{corollary}\label{cor:eq-div-pev}
$\validFCo{\InfRuleSet}{\coRuleSet}{\eval{\conf}{\divres}}$ 
iff $\Rule{}{\eval{\conf}{\Unknown}} \ruleArInf$. 
\end{corollary}

\paragraph{Total semantics} 
We now briefly describe how we can combine the presented constructions in order to get a semantics modelling \emph{all} computations as defined in \cref{sect:bss-ts}. 
In particular, we will use the $\Wrong$ construction to model stuck computations and the construction in this section to model divergence, because they are more similar to each other. 

Let us consider a big-step semantics \ple{\ConfSet,\ResSet,\RuleSet}. 
We add to $\ResSet$ two special values to model stuckness and divergence, defining $\TotResSet = \ResSet + \{\Wrong\} + \{\divres\}$. 
Then, we have to add appropriate rules to handle these two special results: the idea is to add ``simultanously'' rules from \cref{def:wr-rules} and from \cref{def:div-rules}, that is, 
we define $\TotRuleSet = \WrRuleSet \cup \InfRuleSet$. 
Note that, since both $\WrRuleSet$ and $\InfRuleSet$ extend $\RuleSet$, we have $\RuleSet \subseteq \TotRuleSet$. 
In addition, the triple \ple{\ConfSet,\TotResSet,\TotRuleSet} is a big-step semantics according to \cref{def:bss}. 
Finally, to properly model divergence, we have to add corules from \cref{def:co-div}, so that infinite derivations are only allowed to prove divergence. 

Since, as we have noticed, all the presented constructions yield a big-step semantics, starting from another one, we can also try to combine them ``sequentially''. 
Of course, there are two possibilities: either we first apply the $\Wrong$ construction or the divergence construction. 
Nicely, it is not difficult to check that all these possibilities yield the same big-step semantics \ple{\ConfSet,\TotResSet,\TotRuleSet}, as depicted below: 
\[\xymatrix{
\ple{\ConfSet,\ResSet,\RuleSet} \ar@{|->}[r]^{\wrlb} \ar@{|->}[d]_{\divres} \ar@{|->}[dr]^{\totlb} 
  & \ple{\ConfSet,\WrResSet,\WrRuleSet} \ar@{|->}[d]^{\divres} \\
\ple{\ConfSet,\InfResSet,\InfRuleSet} \ar@{|->}[r]_{\wrlb} 
  & \ple{\ConfSet,\TotResSet,\TotRuleSet}
}\]
\FDComm{la prova di questa cosa è noiosa e poco utile, quindi non l'ho ifnita: 
The commutativity of the above diagram follows from simple observations. 
\begin{itemize}
\item If we apply \cref{def:div-rules} to \ple{\ConfSet,\WrResSet,\WrRuleSet}\!, we trivially get $\TotResSet$ as set of results. 
Concerning rules, we start from $\WrRuleSet$ and add divergence ropagation rules, then, since $\RuleSet \subseteq \WrRuleSet$, we get all rules in $\InfRuleSet$, hence we get all rules in $\TotRuleSet$. 
To check that we do not get any further rule, note that, if $\rho_\wrlb \in \WrRuleSet\setminus\RuleSet$, we have that \proprule{\rho_\wrlb}{i}{\divres} is undefined when $\rho_\wrlb = \wrongax{\conf}$, becuase wrong configuration rules have no premises, 
and it is equal to \proprule{\rho}{i}{\divres} both when $\rho_\wrlb = \wrongrule{\rho}{k}{\res}$ and when $\rho_\wrlb = \proprule{\rho}{k}{\Wrong}$, because $i\le k$. 
\item If we apply \cref{def:wr-rules} to \ple{\ConfSet,\InfResSet,\InfRuleSet}, we trivially get $\TotResSet$ as set of results. 
Concerning rules, we start from $\InfRuleSet$ and add wrong configuration rules, wrong result rules and $\Wrong$ propagation rules, 
\end{itemize}
}

Thanks to the commutativity of the above diagram,  we can exploit results proved for the various constructions to get properties of  this last construction, as stated below. 

\begin{proposition}\label{prop:tot-conservative}
The following facts hold:
\begin{enumerate}
\item $\validFCo{\TotRuleSet}{\coRuleSet}{\eval{\conf}{\res}}$ iff $\validInd{\RuleSet}{\eval{\conf}{\res}}$,
\item $\validFCo{\TotRuleSet}{\coRuleSet}{\eval{\conf}{\Wrong}}$ iff $\validInd{\WrRuleSet}{\eval{\conf}{\Wrong}}$, 
\item $\validFCo{\TotRuleSet}{\coRuleSet}{\eval{\conf}{\divres}}$ iff $\validFCo{\InfRuleSet}{\coRuleSet}{\eval{\conf}{\divres}}$. 
\end{enumerate}
\end{proposition}
\begin{proof}
All right-to-left implication are trivial, as $\RuleSet,\WrRuleSet,\InfRuleSet \subseteq \TotRuleSet$. 
The other implications follow from \cref{thm:wr-conservative,thm:div-conservative}, relying on the above commutative diagram. 
\end{proof}

\begin{corollary} \label{cor:total-sem}
For any configuration $\conf\in\ConfSet$, one of the following holds: 
\begin{itemize}
\item either $\validFCo{\TotRuleSet}{\coRuleSet}{\eval{\conf}{\res}}$, for some $\res\in\ResSet$, 
\item or $\validFCo{\TotRuleSet}{\coRuleSet}{\eval{\conf}{\divres}}$, 
\item or $\validFCo{\TotRuleSet}{\coRuleSet}{\eval{\conf}{\Wrong}}$. 
\end{itemize}
\end{corollary}
\begin{proof}
Straightforward from \cref{prop:tot-conservative} and \cref{thm:eq-div-tr,thm:eq-wrong-pev,thm:eq-tr-pev}, since the partial evaluation tree $\Rule{}{\eval{\conf}{\Unknown}}$, either converges to a tree, which is either complete or stuck, or diverges. 
\end{proof}

\bfd Note that these three possibilities in general are not mutually exclusive, that is, for instance, a configuration can both converge to a result and diverge. 
This is due to the fact that big-step rules can define a non-deterministic behaviour.  \efd

  \section{Expressing and proving soundness}
\label{sect:bss-soundness}

A predicate (for instance, a typing judgment) is \emph{sound} when, informally, a program satisfying such predicate (\eg a well-typed program) cannot \emph{go wrong}, following Robin Milner's slogan \citep{Milner78}.
In small-step style, as firstly formulated by \citet{WrightF94}, this is naturally expressed as follows: well-typed programs  never reduce to 
terms which neither are values, nor can be further reduced  (called \emph{stuck} terms). 
The standard technique to ensure soundness is by subject reduction (well-typedness is preserved by reduction) and progress (a well-typed term is not stuck).

In standard (inductive) big-step semantics, soundness, as described above, cannot even be expressed, because diverging and stuck computations are not distinguishable. 

Constructions presented in the previous sections make this distinction explicit, hence they allow us to reason about soundness with respect to a big-step semantics. 
In this section, we discuss how soundness can be expressed and we will provide sufficient conditions. 
In other words, we provide a proof technique to show  the  soundness of a predicate with respect to a big-step semantics. 

It is important to highlight the following about the presented approach to soundness. 
First, even though type systems are the paradigmatic example, we will consider a generic predicate on configurations, hence our approach could be instantiated with other kinds of predicates. 
Second, depending on the kind of construction considered, we can express different flavours of soundness, which will have different proof techniques. 
Finally, and more importantly, 
as  mentioned in the introduction, the extended semantics  is only needed 
to prove the correctness of the technique, whereas to \emph{apply} the technique for a given big-step semantics it is enough to reason on the original rules.

\subsection{Expressing soundness} 
\label{sect:bss-must-may}

In the following, we assume a big-step semantics  \ple{\ConfSet,\ResSet,\RuleSet}, and 
an \emph{indexed  predicate on configurations and results}, that is,  
a family $  \Pred= \ple{\ConfPred_\idx,\ResPred_\idx}_{\idx \in \IdxSet}$, for $\IdxSet$ set of \emph{indexes}, with $\ConfPred_\idx \subseteq \ConfSet$ and $\ResPred_\idx\subseteq\ResSet$. 
A representative case is that, as in the examples of \cref{sect:bss-ex}, predicates on configurations and results are typing judgments and the indexes are types; 
however, this setting is more general and so the proof technique could be applied to other kinds of predicates. 
When there is no ambiguity, we also denote by $\ConfPred$ and $\ResPred$, respectively, the corresponding predicates $\bigcup_{\idx \in \IdxSet} \ConfPred_\idx$ and $\bigcup_{\idx\in\IdxSet} \ResPred_\idx$ on $\ConfSet$ and $\ResSet$ (\eg to be well-typed with an arbitrary type).

To discuss how to express soundness of $\Pred$, first of all note that, in the non-deterministic case (that is, there is possibly more than one computation for a configuration),  we can distinguish two flavours of soundness, see, \eg  \citep{DeNicolaH84}:
\begin{description}
\item[soundness-must] (or simply soundness) no computation can be stuck
\item[soundness-may] at least one computation is not stuck
\end{description}
Soundness-must is the standard soundness in small-step semantics, and can be expressed by the $\Wrong$ construction
as follows: 
\begin{description}
\item[soundness-must]  If $\conf \in \ConfPred$, then $\notvalidInd{\WrRuleSet}{\eval{\conf}{\Wrong}}$
\end{description}
Soundness-must \emph{cannot} be expressed by the constructions  making divergence explicit, 
because stuck computations are not explicitly modelled. 
In contrast,  soundness-may can be expressed, for instance,  by the divergence  construction as follows:
\begin{description}
\item[soundness-may] If $\conf \in \ConfPred$, then $\validFCo{\InfRuleSet}{\coRuleSet}{\eval{\conf}{\infres}}$, for some $\infres\in\InfResSet$
\end{description}
whereas it cannot be expressed by the $\Wrong$ construction, since diverging computations are not modelled.
\bfd Note that, instead, using the total semantics, we can express both flavours of soundness, as it models both diverging and stuck computations. \efd 

Of course soundness-must and soundness-may coincide  in the deterministic case. 
Finally, note that indexes (\eg the specific types of configurations and results) do not play any role in the above statements. 
However, they are relevant in the notion of \emph{strong soundness}, introduced by \citet{WrightF94}.
Strong soundness holds (in must or may flavour) if soundness holds (in must or may flavour), and, moreover,  configurations satisfying $\ConfPred_\idx$ (\eg having a given type) produce results, if any, satisfying $\ResPred_\idx$ (\eg of the same type).
Note that soundness alone does not even guarantee to obtain a result satisfying $\ResPred$ (\eg a well-typed result). 
The sufficient conditions introduced in the following subsection actually ensure strong soundness.

In \cref{sect:bss-sc}, we provide sufficient conditions for soundness-must, showing that  they  ensure soundness as stated above  (\cref{thm:sound-wrong}). 
Then, in \cref{sect:bss-smc}, we provide (weaker) sufficient conditions for soundness-may, and show that they ensure soundness-may  (\cref{thm:sound-div}).

\subsection{Conditions ensuring soundness-must} 
\label{sect:bss-sc}

 The three conditions which ensure the soundness-must property  are \emph{local preservation}, \emph{$\exists$-progress}, and \emph{$\forall$-progress}. 
The names suggest that the former plays the role of the \emph{type preservation (subject reduction)} property, and the latter two of the \emph{progress} property in small-step semantics. 
However, as we will see, the correspondence is only rough, since the reasoning here is different. 

Considering the first condition more closely, we use the name \emph{preservation} rather than type preservation since, as already mentioned, the proof technique can be applied to arbitrary predicates. 
More importantly, \emph{local} means that the condition is \emph{on single rules}  rather than on the semantic relation as a whole, as standard subject reduction; 
the semantic relation is only used in the hypotheses of the condition, so that, when checking it, one can rely on stronger assumptions. 
The same holds for the  other two conditions. 


\begin{definition}[Local preservation \LPCond] \label{def:sc-lp}
\label{sound:preservation}
For each $\rho = \inlinerule{\judg_1\ldots\judg_n}{\conf}{\res}$ in $\RuleSet$, 
if $\conf \in \ConfPred_\idx$, then there exists $\idx_1,\ldots,\idx_n \in \IdxSet$ such that 
\begin{enumerate}
\item\label{def:sc-lp:1} for all $k \in 1..n$, if, for all $h < k$, $\validInd{\RuleSet}{\judg_h}$ and $\ResSet(\judg_h) \in \ResPred_{\idx_h}$, 
then $\ConfSet(\judg_k) \in \ConfPred_{\idx_k}$, and 
\item\label{def:sc-lp:2} if, for all $k\in 1..n$, $\validInd{\RuleSet}{\judg_k}$ and $\ResSet(\judg_k) \in \ResPred_{\idx_k}$, then $\res \in \ResPred_\idx$. 
\end{enumerate}
\end{definition}

Thinking to the paradigmatic case where the indexes are types, to check that this condition holds, for each rule $\rho = \inlinerule{\judg_1\ldots\judg_n}{\conf}{\res}$ where $\conf$, the conclusion, has type $\idx$, 
we have to find types $\idx_1, \ldots, \idx_n$, which can be assigned to (configurations and results in) the premises,  and,  
when all the premises satisfy the chosen type, $\res$, the result in the conclusion, must have type $\idx$, that is, the same type of $\conf$. 
More precisely,  we will proceed as follows: 
we start finding type $\idx_1$, and successively find the type $\idx_k$ for (the configuration in) the $k$-th premise assuming that  all previous premises are derivable and their results have the expected  types, and, finally, we have to check that the final result $\res$ has type $\idx$ assuming all premises are derivable and their results have the expected type. 
Indeed, if all such previous premises are derivable, then the expected type should be preserved by their results; if some premise is not derivable, the considered rule is ``useless''. 
For instance, considering (an instantiation of) meta-rule $\metainlinerule{app}{\eval{\e_1}{\Lam{\x}\e}\ \ \eval{\e_2}{\val_2}\ \ \eval{\subst{\e}{\val_2}{\x}}{\val}}{\e_1\appop\e_2}{\val}$ in \cref{fig:bss-ex-lambda}, we prove that $\subst{\e}{\val_2}{\x}$ has the type $\T$ of $\e_1\appop\e_2$ under the assumption that $\Lam{\x}\e$ has type $\funType{\tA'}{\tA}$, and $\val_2$ has type $\T'$ (see the proof example in \cref{sect:bss-simply-typed} for more details).
A counter-example to condition \LPCond is discussed at the beginning of \cref{sect:bss-iut}. 

The following lemma states that local preservation actually implies \emph{preservation} of the semantic relation as a whole.  

\begin{lemma}[Preservation] \label{lem:sr}
Let \ple{\ConfSet,\ResSet,\RuleSet} and $\Pred = \ple{\ConfPred_\idx,\ResPred_\idx}_{\idx\in\IdxSet}$ satisfy condition \LPCond. 
If $\validInd{\RuleSet}{\eval{\conf}{\res}}$ and $\conf \in \ConfPred_\idx$, then $\res \in \ResPred_\idx$.
\end{lemma}
\begin{proof}
The proof is by  a double induction
From the hypotheses, we know that $\eval{\conf}{\res}$ has a finite derivation in $\RuleSet$ and $\conf\in \ConfPred_\idx$. 
The first induction is on the derivation of $\eval{\conf}{\res}$. 
Suppose the last applied rule is $\rho = \inlinerule{\judg_1\ldots\judg_n}{\conf}{\res}$ and denote by $RH$ the induction hypothesis. 
Then, we prove by complete arithmetic induction on $k \in 1..n$ (the second induction)  that $\ConfSet(\judg_k) \in \Pred_{\idx_k}$, for all $k \in 1..n$ and for some $\idx_1, \ldots, \idx_n \in \IdxSet$.
Let us denote by $IH$ the second induction hypothesis. 
By \LPCond, there are indexes $\idx_1, \ldots, \idx_n \in \IdxSet$, \bfd satisfying \cref{def:sc-lp:1,def:sc-lp:2} of \LPCond (\cf \cref{def:sc-lp}). \efd 
Let $k\in 1..n$, then 
by $IH$ we know that $\ConfSet(\judg_h) \in \ConfPred_{\idx_h}$, for all $h < k$. 
Then, by  $RH$,  
we get that $\ResSet(\judg_h) \in \ResPred_{\idx_h}$.  
Hence,
by \LPCond (\cf \refItem{def:sc-lp}{1}), we get  $\ConfSet(\judg_k) \in \Pred_{\idx_k}$, as needed. 

Now, since $\ConfSet(\judg_k) \in \ConfPred_{\idx_k}$, for all $k\in 1..n$, as we have just proved, again by $RH$, we get that $\ResSet(\judg_k)\in\ResPred_{\idx_k}$, for all $k\in 1..n$. 
Then, by \LPCond (\cf \refItem{def:sc-lp}{2}), we conclude that $\res \in \ResPred_\idx$, as needed. 
\end{proof}

The following proposition  is  a form of local preservation  where indexes (\eg specific types) are not relevant, simpler to use in the 
 proofs of  
\cref{thm:sound-wrong,thm:sound-div}.  

\begin{proposition} \label{prop:preservation}
Let \ple{\ConfSet,\ResSet,\RuleSet} and $\Pred = \ple{\ConfPred_\idx,\ResPred_\idx}_{\idx\in\IdxSet}$ satisfy condition \LPCond. 
For each rule $\rho = \inlinerule{\judg_1\ldots\judg_n}{\conf}{\res}$ and $k \in 1..n$, 
if $\conf \in \ConfPred$ and, for all $h < k$, $\validInd{\RuleSet}{\judg_h}$, then $\ConfSet(\judg_k) \in \ConfPred$.
\end{proposition}
\begin{proof}
By hypothesis we know that $\conf \in \ConfPred_\idx$, for some $\idx \in \IdxSet$, thus by condition \LPCond, there are indexes $\idx_1,\ldots,\idx_n \in \IdxSet$, \bfd satisfying \cref{def:sc-lp:1,def:sc-lp:2} of \LPCond (\cf \cref{def:sc-lp}). \efd 
We show by 
complete arithmetic induction that, for all $k\in 1..n$, $\ConfSet(\judg_k) \in \ConfPred_{\idx_k}$, which implies the thesis. 
Assume the thesis for all $h < k$,
then, since by hypothesis we have $\validInd{\RuleSet}{\judg_h}$ for all $ h < k$, we  get, by induction hypothesis,  $\ConfSet(\judg_h) \in \ConfPred_{\idx_h}$, for all $h < k$. 
By \cref{lem:sr}, we also get $\ResSet(\judg_h) \in \ResPred_{\idx_h}$, hence, by condition \LPCond (\cf \refItem{def:sc-lp}{1}), we get $\ConfSet(\judg_k) \in \ConfPred_{\idx_k}$, as needed. 
\end{proof}

The second condition, named \emph{$\exists$-progress}, ensures that, for configurations satisfying $\Pred$  (\eg well-typed),
we can \emph{start} the evaluation, that is, the construction of an evaluation tree. 

\begin{definition}[$\exists$-progress \ExPCond] \label{def:sc-exp} 
For each 
$\conf \in \ConfPred$, there exists a rule $\rho\in\RuleSet$ such that  $\ConfSet(\rho) = \conf$.
\end{definition}

The third condition, named \emph{$\forall$-progress}, ensures that, for configurations satisfying 
$\Pred$ (\eg well-typed), we can \emph{continue} the evaluation, that is, the construction of the evaluation tree. 
This condition uses the equivalence on rules introduced in \cref{def:rule-eq}. 

\begin{definition}[$\forall$-progress \AllPCond] \label{def:sc-allp}
For each rule $\rho = \inlinerule{\judg_1\ldots\judg_n}{\conf}{\res}$ with  $\conf \in \ConfPred$, 
for each $k\in 1..n$,
if, for all $h < k$, $\validInd{\RuleSet}{\judg_h}$ and $\validInd{\RuleSet}{\eval{\ConfSet(\judg_k)}{\res'}}$, for some $\res' \in \ResSet$, 
then there is a rule $\rho' \sim_k \rho$ such that $\ResSet(\rho', k) = \res'$. 
\end{definition}

We have to check, for each rule $\rho = \inlinerule{\judg_1\ldots\judg_n}{\conf}{\res}$, the following: 
if the configuration $\conf$ in the conclusion satisfies the predicate (\eg is well-typed), then, for each $k \in 1..n$, 
if the configuration in the $k$-th premise evaluates to some result $\res'$ (that is, $\validInd{\RuleSet}{\eval{\ConfSet(\judg_k)}{\res'}}$), 
then there is a rule ($\rho$ itself or another rule with the same  configuration in the conclusion and the same first $k-1$ premises) with such judgement as $k$-th premise.  
This check  can be done  under the assumption that all the previous premises are derivable. 
For instance, consider again (an instantiation of) the meta-rule $\metainlinerule{app}{\eval{\e_1}{\Lam{\x}\e}\ \ \eval{\e_2}{\val_2}\ \ \eval{\subst{\e}{\val_2}{\x}}{\val}}{\e_1\appop\e_2}{\val}$. 
Assuming that $\e_1$ evaluates to some $\val_1$,  we have to check that there is a rule  with first premise  $\eval{\e_1}{\val_1}$, in practice,  that $\val_1$ is a $\lambda$-abstraction; 
in general, checking \AllPCond for a (meta-)rule amounts to show that configurations in the premises evaluate to results  with  the required shape (see also the proof example in \cref{sect:bss-simply-typed}). \label{page:sc-allp}

We now prove the claim of soundness-must  expressed by means of the $\Wrong$ construction (\cf \cref{sect:bss-wrong}). 

\begin{theorem}[Soundness-must]  \label{thm:sound-wrong}
Let \ple{\ConfSet,\ResSet,\RuleSet} and $\Pred = \ple{\ConfPred_\idx,\ResPred_\idx}_{\idx\in\IdxSet}$ satisfy conditions \LPCond, \ExPCond and \AllPCond. 
If $\conf \in \ConfPred$, then $\notvalidInd{\WrRuleSet}{\eval{\conf}{\Wrong}}$. 
\end{theorem}
\begin{proof}
To prove the statement, we assume $\validInd{\WrRuleSet}{\eval{\conf}{\Wrong}}$  and look for a contradiction. 
The proof is by induction on the derivation of $\eval{\conf}{\Wrong}$. 
We split cases on the last applied rule in such derivation. 
\begin{proofcases}
\item [\wrongax{\conf}] 
By construction (\cf \cref{def:wr-rules}), we know that there is no rule $\rho \in \RuleSet$ such that $\ConfSet(\rho) = \conf$, and this violates  condition \ExPCond, since $\conf \in \ConfPred$, by hypothesis. 

\item [\wrongrule{\rho}{i}{\res'}]
Suppose $\rho = \inlinerule{\judg_1\ldots\judg_n}{\conf}{\res}$, hence $i\in 1..n$, 
then, by hypothesis, for all $k < i$, we have  $\validInd{\WrRuleSet}{\judg_k}$, and $\validInd{\WrRuleSet}{\eval{\ConfSet(\judg_i)}{\res'}}$, and these judgments can also be derived in $\RuleSet$ by conservativity (\cf \cref{thm:wr-conservative}). 
Furthermore, by construction (\cf \cref{def:wr-rules}), we know that there is no other rule $\rho' \sim_i \rho$ such that $\ResSet(\rho', i) = \res'$, and this violates condition \AllPCond, since $\conf \in \ConfPred$ by hypothesis. 

\item [\proprule{\rho}{i}{\Wrong}] 
Suppose $\rho = \inlinerule{\judg_1\ldots\judg_n}{\conf}{\res}$, hence $i\in 1..n$, 
then,  by hypothesis, for all $k < i$, we have $\validInd{\WrRuleSet}{\judg_k}$, 
and these judgments 
can also be derived in $\RuleSet$ by conservativity (\cf \cref{thm:wr-conservative}). 
Then, by \cref{prop:preservation} (which requires condition \LPCond), since $\conf \in \ConfPred$, we have
$\ConfSet(\judg_i) \in \ConfPred$, hence we get the thesis by induction hypothesis, because $\validInd{\WrRuleSet}{\eval{\ConfSet(\judg_i)}{\Wrong}}$ holds by hypothesis. 
\end{proofcases}
\end{proof}

Note that conditions \LPCond, \ExPCond and \AllPCond, actually ensure strong soundness, because, by \cref{lem:sr}, which is applicable since we assume \LPCond, we have that converging computations preserve indexes of the predicate.

\subsection{Conditions ensuring soundness-may}
\label{sect:bss-smc}

As discussed in \cref{sect:bss-must-may}, if we explicitly model divergence rather than stuck computations, we can only express a weaker form of soundness: at least one computation is not stuck (\emph{soundness-may}). 
Actually, we will state soundness-may in a different, but equivalent, way, which is simpler to prove, that is, a configuration that does not converge, diverges. 

As the reader can expect, to ensure this property weaker sufficient conditions are enough: namely, condition \LPCond, and another condition,
 named \emph{may-progress}, defined below. 
We write ``$\notvalidInd{\RuleSet}{\eval{\conf}{}}$''  if $\conf$ \emph{does not converge} (there is no $\res$ such that $\validInd{\RuleSet}{\eval{\conf}{\res}}$). 

\begin{definition}[May-progress \MayPCond] \label{def:sc-mayp} 
For each $\conf \in \ConfPred$, there is a rule $\rho = \inlinerule{\judg_1\ldots\judg_n}{\conf}{\res}$ such that, 
if there is a (first) $k \in 1..n$ 
such that $\notvalidInd{\RuleSet}{\judg_k}$ and, 
for all $h < k$, 
$\validInd{\RuleSet}{\judg_h}$,
then $\notvalidInd{\RuleSet}{\eval{\ConfSet(\judg_k)}{}}$. 
\end{definition}

This condition can be informally understood as follows: 
we have to show that there is an either finite or infinite computation for $\conf$. 
If we find a rule where all premises are derivable (there is no $k$), then there is a finite computation. 
Otherwise, $\conf$ cannot converge. 
In this case, we should find a rule where the configuration in the first non-derivable premise $k$ cannot converge as well. 
Indeed, by coinductive reasoning (\cf \cref{thm:sound-div}), this implies that $\conf$ diverges. 
The following proposition states that this condition is indeed a weakening of \ExPCond and \AllPCond. 

\begin{proposition} \label{prop:progress-to-may}
Conditions \ExPCond and \AllPCond imply condition \MayPCond. 
\end{proposition}
\begin{proof}
For each $\conf \in \ConfSet$, let us define $b_\conf  \in \N$ as $\max \{ \#\rho  \mid \ConfSet(\rho) = \conf \}$, which is well-defined and finite by condition (\BP) in \cref{def:bss}. 
For each rule $\rho$ with $\ConfSet(\rho) = \conf$,  let us denote by $nd(\rho)$ the index of the first premise of $\rho$ which is not derivable, if any, 
otherwise set  $nd(\rho) = b_\conf + 1$. 
For each $\conf \in \ConfPred$, we first prove the following fact: 
$(\star)$ for each rule $\rho$, with $\ConfSet(\rho) = \conf$, there exists a rule $\rho'$ such that $\ConfSet(\rho') = \conf$,  $nd(\rho') \ge nd(\rho)$ and,
if $nd(\rho') \le b_\conf$, then $\notvalidInd{\RuleSet}{\eval{\ConfSet(\rho', nd(\rho'))}{ }}$. 
Note that the requirement in $(\star)$ is the same as that of condition \MayPCond. 
The proof is by complete arithmetic induction on $h(\rho) = b_\conf + 1 - nd(\rho)$. 
If $h(\rho) = 0$, hence $nd(\rho) = b_\conf + 1$, then the thesis follows by taking $\rho' = \rho$. 
Otherwise, we have two cases: 
if there is no $\res \in \ResSet$ such that $\validInd{\RuleSet}{\eval{\ConfSet(\rho, nd(\rho))}{\res}}$, then we have the thesis taking $\rho' = \rho$;
otherwise, by condition \AllPCond, there is a rule $\rho'' \sim_{nd(\rho)} \rho$ such that $\ResSet(\rho'', nd(\rho)) = \res$, hence $nd(\rho'') > nd(\rho)$. 
Then, we have $h(\rho'') < h(\rho)$, hence we get the thesis by induction hypothesis.

Now, by \ExPCond, there is a rule $\rho$ with $\ConfSet(\rho) = \conf$, and applying $(\star)$ to $\rho$ we get \MayPCond. 
\end{proof} 

We now prove the claim of soundness-may expressed by means of the divergence  construction (\cf \cref{sect:bss-div}). 

\begin{theorem}[Soundness-may]  \label{thm:sound-div}  
Let \ple{\ConfSet,\ResSet,\RuleSet} and $\Pred = \ple{\ConfPred_\idx,\ResSet_\idx}_{\idx\in\IdxSet}$ satisfy conditions \LPCond and \MayPCond. 
If $\conf \in \ConfPred$, then $\validFCo{\InfRuleSet}{\coRuleSet}{\eval{\conf}{\infres}}$, for some $\infres \in\InfResSet$. 
\end{theorem}
\begin{proof}
First note that, thanks to  \cref{thm:div-conservative}, the statement is equivalent to the following:
\begin{quote}
If $\conf \in \ConfPred$ and $\notvalidInd{\RuleSet}{\eval{\conf}{}}$, then $\validFCo{\InfRuleSet}{\coRuleSet}{\eval{\conf}{\divres}}$. 
\end{quote}
Then, the thesis follows by \cref{lem:div-proof}. 
We set $\Spec = \{ \conf\in\ConfSet  \mid \conf \in \ConfPred\mbox{ and } \notvalidInd{\RuleSet}{\eval{\conf}{}} \}$,  and show that, for all $\conf \in \Spec$, 
there are $\rho = \inlinerule{\judg_1\ldots\judg_n}{\conf}{\res}$ and $k \in 1..n$ such that, for all $h < k$, $\validInd{\RuleSet}{\judg_h}$ and $\ConfSet(\judg_k) \in \Spec$. 

Consider $\conf \in \Spec$, then, by \MayPCond (\cf \cref{def:sc-mayp}), there is $\rho = \inlinerule{\judg_1\ldots\judg_n}{\conf}{\res}$. 
By definition of $\Spec$, we have $\notvalidInd{\RuleSet}{\eval{\conf}{}}$, 
hence there exists a  (first)  $k \in 1..n+1$ such that $\notvalidInd{\RuleSet}{\judg_k}$, 
since, otherwise, we would have $\validInd{\RuleSet}{\eval{\conf}{\res}}$. 
Then, since $k$ is the  first  index with such property, for all $h < k$, we have $\validInd{\RuleSet}{\judg_h}$, hence, again by condition \MayPCond (\cf \cref{def:sc-mayp}), we have that  $\notvalidInd{\RuleSet}{\eval{\ConfSet(\judg_k)}{}}$. 
Finally, since $\conf\in\ConfPred$ and, for all $h < k$, we have $\validInd{\RuleSet}{\judg_h}$, by \cref{prop:preservation} we get  $\ConfSet(\judg_k) \in \ConfPred$,  hence $\ConfSet(\judg_k) \in \Spec$, as needed. 
\end{proof}

Note that conditions \LPCond and \MayPCond actually ensure strong soundness, because, by \cref{lem:sr}, which is applicable since we assume \LPCond, we have that converging computations preserve indexes of the predicate.

  \section{Examples of soundness proofs} 
\label{sect:bss-ex}

In this section, we show how to use the technique introduced in \cref{sect:bss-soundness} to prove soundness of a type system with respect to a big-step semantics, by several examples. 
We focus on the technique for soundness-must, as it is the usual notion of soundness for type systems. 
\cref{sect:bss-simply-typed} explains in detail how a typical soundness proof can be rephrased in terms of our technique,  by reasoning directly  on big-step rules. 
\cref{sect:bss-fjl} shows a case where  this  is advantageous, since the property to be checked is \emph{not preserved} by intermediate computation steps, whereas it holds for the whole computation.
\cref{sect:bss-iut}  considers  a more sophisticated type system, with intersection and union types. 
\cref{sect:bss-fjos} shows another example where types are  not preserved,  whereas soundness can be proved  with our technique. 
This example is intended as a preliminary step towards a more challenging case. 
In \cref{sect:bss-ifj} we show how our technique can also deal with imperative features. 

\subsection{Simply-typed $\lambda$-calculus with recursive types}
\label{sect:bss-simply-typed}

As  a first example, we take the $\lambda$-calculus with natural constants, successor,  and non-deterministic choice introduced in \cref{fig:bss-ex-lambda}. 
We consider a standard simply-typed version with (equi)recursive types, obtained by interpreting the production in the top section of \cref{fig:bss-ex-lambda-types}
coinductively.  
Introducing  recursive types makes the calculus non-normalising and  permits  
to write interesting programs such as $\Omega$ (see \cref{sect:bss-traces}).

The typing rules are recalled in the bottom section of \cref{fig:bss-ex-lambda-types} and, as usual,  they are interpreted inductively. 
Type environments, written $\Gamma$, are finite maps from variables to types, and $\SubstFun{\Gamma}{\T}{\x}$ denotes the map which returns $\T$ on $\x$  and  coincides with $\Gamma$ elsewhere.  
We write $\HasType{\es}{\e}{\T}$ for $\HasType{\emptyset}{\e}{\T}$.

\begin{figure} 
\begin{center}
\begin{small}
\begin{math} 
\begin{grammatica}
\produzione{\T}{\natType\mid\funType{\T_1}{\T_2}}{types}
\end{grammatica}
\end{math}
\HSep
\begin{math}
\begin{array}{c}
\MetaRule
{t-var}
{}
{\HasType{\Gamma}{\x}{\T}}{\Gamma(\x)=\T}
\BigSpace
\MetaRule{t-const}{}{\HasType{\Gamma}{\natconst}{\natType}}{}
\\[4ex]
\MetaRule
{t-abs}
{\HasType{\SubstFun{\Gamma}{\T'}{\x}}{\e}{\T}}
{\HasType{\Gamma}{\Lam{\x}\e}{\funType{\T'}{\T}}}{} 
\BigSpace
\MetaRule
{t-app}
{\HasType{\Gamma}{\e_1}{\funType{\T'}{\T}}\Space\HasType{\Gamma}{\e_2}{\T'}}
{\HasType{\Gamma}{\e_1\appop\e_2}{\T}}{} 
\\[4ex]
\MetaRule
{t-succ}{\HasType{\Gamma}{\e}{\natType}}
{\HasType{\Gamma}{\SuccExp{\e}}{\natType}}{}
\BigSpace
\MetaRule{t-choice}{\HasType{\Gamma}{\e_1}{\T}\Space\HasType{\Gamma}{\e_2}{\T}}{\HasType{\Gamma}{\Choice{\e_1}{\e_2}}{\T}}{}
\end{array}
\end{math}
\end{small}
\end{center}
\caption{$\lambda$-calculus: type system}\label{fig:bss-ex-lambda-types}
\end{figure}

Let \ple{\ConfSet_1,\ResSet_1,\RuleSet_1} be the big-step semantics described in \cref{fig:bss-ex-lambda} ($\ConfSet_1$ is the set of expressions and $\ResSet_1$ is the set of values), and let 
$\ConfPredOne_T = \{\e \in \ConfSet_1 \mid \HasType{\es}{\e}{\T} \}$ and 
$\ResPredOne_T = \{\val \in \ResSet_1 \mid \HasType{\es}{\val}{\T} \}$, 
where $\T$ is a type, defined in \cref{fig:bss-ex-lambda-types}, that is, 
$\ConfPredOne_\T$ and $\ResPredOne_\T$ are the sets of expressions and values of type $\T$, respectively. 
To prove the three conditions \LPCond, \ExPCond and \AllPCond  of \cref{sect:bss-sc}, we need lemmas of inversion, substitution and canonical  forms, as in the standard technique for small-step semantics. 

\begin{lemma}[Inversion]\label{lem:ilr}
The following hold:
\begin{enumerate}
\item \label{lem:ilr:1}If $\HasType{\Gamma} \x\tA$, then $\Gamma(\x)= \tA$.
\item \label{lem:ilr:2} If $\HasType{\Gamma}{\natconst}\T$, then $\T=\natType$.
\item \label{lem:ilr:3} If $\HasType\Gamma{\LambdaExp{\x}{\e}}\tA$, then $\tA=\funType{\tA_1}{\tA_2}$ and  $\HasType{\SubstFun\Gamma{\tA_1}\x} \e {\tA_2}$.
\item \label{lem:ilr:4} If $\HasType\Gamma{\e_1\appop\e_2}\tA$, then $\HasType\Gamma{\e_1} {\funType{\tA'}{\tA}}$ and $\HasType\Gamma{\e_2} {\tA'}$.
\item \label{lem:ilr:5} If $\HasType{\Gamma}{\SuccExp{\e}}{\T}$, then $\T=\natType$ and $\HasType{\Gamma}{\e}{\natType}$.
\item \label{lem:ilr:6} If $\HasType\Gamma {\e_1\oplus\e_2}\tA$, then $\HasType\Gamma{\e_i}\tA$ with $i\in1,2$.
\end{enumerate}
\end{lemma}

\begin{lemma}[Substitution]\label{lem:s}
If $\HasType{\SubstFun\Gamma{\tA'}\x} \e {\tA}$ and $\HasType\Gamma{\e'} {\tA'}$, then \mbox{$\HasType\Gamma{\subst\e{\e'}\x} {\tA}$}.
\end{lemma}

\begin{lemma}[Canonical Forms]\label{lem:cf}\
The following hold: 
\begin{enumerate}
\item \label{lem:cf:1}
If $\HasType{\es} \val {\funType{\tA'}{\tA}}$, then $\val=\LambdaExp{\x}{\e}$.
\item \label{lem:cf:2}
If $\HasType{\es} \val \natType$, then $\val=\natconst$.
\end{enumerate}
\end{lemma}

\begin{theorem}[Soundness]\label{thm:sd}
The big-step semantics \ple{\ConfSet_1,\ResSet_1,\RuleSet_1}  and the indexed predicate $\PredOne$ satisfy the conditions \LPCond, \ExPCond and \AllPCond  of \cref{sect:bss-sc}.
\end{theorem}
\begin{proof}
Since the aim of this first example is to illustrate the proof technique, we provide a proof where we explain the reasoning  in detail.

\proofsect{Proof of \LPCond} 
We should prove this condition for each (instantiation of meta-)rule in \cref{fig:bss-ex-lambda}. 
\begin{proofcases}
\item [\rn{app}] 
Assume that $\HasType{\es}{\e_1\appop\e_2}{\T}$ holds.  
We have to find types for the premises.  
We  proceed as follows:
\begin{enumerate}
\item First premise: by \refItem{lem:ilr}{4}, $\HasType{\es}{\e_1} {\funType{\tA'}{\tA}}$.
\item Second premise:  again by \refItem{lem:ilr}{4}, $\HasType{\es}{\e_2}{\T'}$  (without needing the assumption $\HasType{\es}{\Lam{\x}\e}{\funType{\tA'}{\tA}}$).   
\item Third premise:  $\HasType{\es}{\subst{\e}{\val_2}{\x}}{\T}$ should hold (assuming $\HasType{\es}{\Lam{\x}\e}{\funType{\tA'}{\tA}}$, $\HasType{\es}{\val_2}{\T'}$).  
Since $\HasType{\es}{\Lam{\x}\e}{\funType{\tA'}{\tA}}$, by \refItem{lem:ilr}{3} we have $\HasType{\x{:}\T'} \e {\tA}$, so by \cref{lem:s} and $\HasType{\es}{\val_2}{\T'}$ we have $\HasType{\es}{\subst\e{\val_2}\x} {\tA}$. 
\end{enumerate}
Finally, we have to show $\HasType{\es}{\val}{\tA}$, assuming $\HasType{\es}{\Lam{\x}\e}{\funtype{\tA'}{\tA}}$, $\HasType{\es}{\val_2}{\tA'}$ and  $\HasType{\es}{\val}{\tA}$, which is trivial from the third assumption. 

\item [\rn{succ}] 
Assume that $\HasType{\es}{\SuccExp{\e}}{\T}$ holds. 
By \refItem{lem:ilr}{5}, $\T=\natType$, and $\HasType{\es}{\e}{\natType}$, hence we find $\natType$ as type for the premise. 
Moreover, $\HasType{\es}{\natconst+1}{\natType}$ holds by rule \rn{t-const}. 

\item [\rn{choice}] 
Assume that $\HasType{\es}{\Choice{\e_1}{\e_2}}{\T}$ holds. 
By \refItem{lem:ilr}{6}, we have $\HasType{\es}{\e_i}{\T}$, with $i\in 1,2$. 
Hence we  find  $\T$ as type for the  premise.
Finally, we have to show $\HasType{\es}{\val}{\tA}$, assuming  $\HasType{\es}{\val}{\tA}$, which is trivial. 

\item [\rn{val}] Trivial by assumption. 
\end{proofcases}

\proofsect{Proof of \ExPCond} 
We should prove that, for each configuration (here, expression $\e$) such that $\HasType{\es}{\e}{\T}$ holds for some $\T$, there is a rule with this configuration in the conclusion.  
The expression $\e$ cannot be a variable, since a variable cannot be typed in the empty environment. 
Application, successor, choice, abstraction and constants appear as consequence in the big-step rules \rn{app}, \rn{succ}, \rn{choice} and \rn{val}. 

\proofsect{Proof of \AllPCond}
We should prove this condition for each (instantiation of meta-)rule. 
\begin{proofcases}
\item [\rn{app}] 
Assuming $\HasType{\es}{\e_1\appop\e_2}{\T}$, again  by \refItem{lem:ilr}{4}  we get  ${\HasType{\es}{\e_1} {\funType{\tA'}{\tA}}}$. 
\begin{enumerate}
\item First premise: if $\eval{\e_1}{\val}$ is derivable, then there should be a rule with $\e_1\appop\e_2$ in the conclusion and $\eval{\e_1}{\val}$ as first premise. 
Since we proved \LPCond, by preservation (\cref{lem:sr})  
$\HasType{\es}{\val}{\funType{\tA'}{\tA}}$ holds. 
Then, by \refItem{lem:cf}{1}, $\val$ has shape $\Lam{\x}\e$, hence the required rule exists. 
As noted at page \pageref{page:sc-allp}, in practice checking \AllPCond for a (meta-)rule amounts to show that configurations in the premises evaluate to results which have the required shape (to be a $\lambda$-abstraction in this case).
\item Second premise: if $\eval{\e_1}{\Lam{\x}\e}$, and $\eval{\e_2}{\val}$, then there should be a rule  with $\e_1\appop\e_2$ in the conclusion and $\eval{\e_1}{\Lam{\x}\e}$, $\eval{\e_2}{\val}$ as first two premises. 
This is trivial since the meta-variable $\val_2$ can be freely instantiated in the meta-rule. 
\item Third premise: trivial as the previous one. 
\end{enumerate}

\item [\rn{succ}] 
Assuming $\HasType{\es}{\SuccExp{\e}}{\T}$, again by \refItem{lem:ilr}{5} we get $\HasType{\es}{\e}{\natType}$. 
If $\eval{\e}{\val}$ is derivable, there should be a rule with $\SuccExp{\e}$ in the conclusion and $\eval{\e}{\val}$ as first premise. 
Indeed, by preservation (\cref{lem:sr}) and \refItem{lem:cf}{2}, $\val$ has shape $\natconst$. 

\item [\rn{choice}] Trivial since the meta-variable $\val$ can be freely instantiated.

\item [\rn{val}] Empty, because there are no premises. 
\end{proofcases}
\end{proof}

An interesting remark is that, differently from the standard approach, there is  \emph{no induction}  in the proof: everything is \emph{by cases}. 
This is a consequence of the fact that, as discussed in \cref{sect:bss-sc}, the three conditions are \emph{local}, that is, they are conditions on single rules. 
Induction is ``hidden'' once and for all in the proof that  those  three conditions  are sufficient to  ensure soundness.

If we drop in \cref{fig:bss-ex-lambda} rule \rn{succ}, then condition \ExPCond fails, since  there is no longer a rule for the well-typed configuration $\SuccExp{\natconst}$. 
If we add the \rn{fool} rule $\HasType {\es}{0\appop 0}{\natType}$, then condition \AllPCond  fails for rule \rn{app}, since $\eval{0}{0}$ is derivable, but there is no rule with $0\appop 0$  in the conclusion and $\eval{0}{0}$ as first premise.\EZComm{e la prima condizione?}


\subsection{$\MiniFJLambda$}
\label{sect:bss-fjl}

In this example, the language is a subset of $\FJLambda$ \citep{BettiniBDGV18}, a calculus extending Featherweight Java ($\FJ$) with $\lambda$-abstractions and intersection types, introduced in Java 8.  
To keep the example small, we do not consider intersections and focus  on one key typing feature: $\lambda$-abstractions can only be typed when occurring in a context requiring a given type (called  the  \emph{target type}). 
In a small-step semantics, this poses a problem: reduction can move $\lambda$-abstractions into arbitrary contexts, leading to intermediate terms which would be ill-typed. 
To  maintain  subject reduction, \citet{BettiniBDGV18} decorate $\lambda$-abstractions with their initial target type. 
In a big-step semantics, there is no need of intermediate terms and annotations.

\begin{figure} 
\begin{small}
\begin{math}
\begin{grammatica}
\produzione{\e}{\x\mid\FieldAccess{\e}{\f}\mid\ConstrCall{\CC}{\e_1,\ldots,\e_n}\mid
\MethCall{\e}{\m}{\e_1,\ldots,\e_n}\mid\LambdaExp{\xs}{\e}\mid\Cast{\T}{\e}}{expression}\\
\produzione{\T}{\CC\mid\II}{type}
\end{grammatica}
\end{math}

\HSep

\begin{math}
\begin{grammatica}
\produzione{\conf}{\Conf{\env}{\e}}{configuration}\\
\produzione{\val}{\obj{\CC}{\vals}\mid\Lam{\xs}\e}{result (value)}\\
\end{grammatica}
\end{math}

\HSep

\begin{math}
\begin{array}{c}
\MetaRule{var}{}{\eval{\Conf{\env}{\x}}{\val}}{\env(\x)=\val}
\BigSpace
\MetaRule{new}{
\eval{\Conf{\env}{\e_i}}{\val_i}\Space \forall i\in 1..n}{\eval{\Conf{\env}{\ConstrCall{\CC}{\e_1,\ldots,\e_n}}}{\obj{\CC}{\val_1,\ldots,\val_n}}} 
{} 
\\[4ex]
\MetaRule{field-access}{\eval{\Conf{\env}{\e}}{\obj{\CC}{\val_1,\ldots,\val_n}}}{
\eval{\Conf{\env}{\FieldAccess{\e}{\f_i}}}{\val_i}}
{\begin{array}{l}
\fields{\CC}=\Field{\T_1}{\f_1}\ldots\Field{\T_n}{\f_n}\\
i\in 1..n\\
\end{array}
}\\[4ex]
\MetaRule{invk}{
\begin{array}{l}
\eval{\Conf{\env}{\e_0}}{\obj{\CC}{{\vals}}}\\ 
\eval{\Conf{\env}{\e_i}}{\val_i}\Space \forall i\in 1..n\\ 
\eval{\Conf{\x_1{:}\val_1,\ldots,\x_n{:}\val_n,\kwthis{:}\obj{\CC}{\vals}}{\e}}{\val}
\end{array}}
{\eval{\Conf{\env}{\MethCall{\e_0}{\m}{\e_1,\ldots,\e_n}}}{\val}}
{ \mbody{\CC}{\m}=\Pair{\x_1\ldots\x_n}{\e}  }
\\[4ex]
\MetaRule{$\lambda$-invk}{
\begin{array}{l}
\eval{\Conf{\env}{\e_0}}{\LambdaExp{\x_1\ldots\x_n}{\e}}\\ 
\eval{\Conf{\env}{\e_i}}{\val_i}\Space \forall i\in 1..n\\ 
\eval{\Conf{\x_1{:}\val_1,\ldots,\x_n{:}\val_n}{\e}}{\val}
\end{array}}
{ \eval{\Conf{\env}{\MethCall{\e_0}{\m}{\e_1,\ldots,\e_n}}}{\val}}
{}
\\[4ex]
\MetaRule{$\lambda$}{}{\eval{\Conf{\env}{\Lam{\xs}\e}}{\Lam{\xs}\e}}{}
\BigSpace
\MetaRule{upcast}{
\eval{\Conf{\env}{\e}}{\val}
}{\eval{\Conf{\env}{\Cast{\T}{\e}}}{\val}}
{}
\end{array}
\end{math}
\end{small}
\caption{$\MiniFJLambda$: syntax  and  big-step semantics}\label{fig:FJ-lambda-big-step}\label{fig:FJ-lambda-syntax}
\end{figure}

The syntax is given in the first part of \cref{fig:FJ-lambda-syntax}. 
We assume sets of \emph{variables} $\x$, \emph{class names} $\CC$, \emph{interface names} $\II$, $\JJ$, \emph{field names} $\f$, and \emph{method names} $\m$. 
As usual, we assume a special variable $\kwthis$, used in method bodies to refer to the receiver object. 
Interfaces which have \emph{exactly} one method (dubbed \emph{functional interfaces}) can be used as target 
types.
Expressions are those of $\FJ$, plus $\lambda$-abstractions, and types are class and interface names. 
Throughout this section $\xs$  and $\vals$ denote lists of variables and values, respectively. 
In $\Lam{\xs}\e$ we assume that $\xs$ is not empty and  $\e$ is not a $\lambda$-abstraction.
For simplicity, we only consider \emph{upcasts}, which have no runtime effect, but are important to allow the programmer to use $\lambda$-abstractions,
as exemplified in discussing typing rules.

To be concise, the class table is abstractly modelled  as follows: 
\FDComm{sono quelli dichiarati o anche ereditati?} \EZComm{anche quelli ereditati}
\begin{itemize}
\item $\fields{\CC}$ gives the sequence of  field declarations $\Field{\T_1}{\f_1}..\Field{\T_n}{\f_n}$ for class $\CC$  
\item $\mtype{\T}{\m}$ gives, for each method $\m$ in class or interface $\T$, the pair $\funtype{\T_1\ldots\T_n}{\T'}$ 
consisting of the parameter types and return type
\item $\mbody{\CC}{\m}$ gives, for each method $\m$ in class $\CC$, the pair %
$\Pair{\x_1\ldots\x_n}{\e}$ 
consisting of the parameters and body
\item $\leqfj$ is the reflexive and transitive closure of the union of the $\aux{extends}$ and $\aux{implements}$ relations, stating that two class or interface names are related iff they occur in the class table connected by the keywords $\aux{extends}$ or $\aux{implements}$ 
\item $\umtype{\II}$ gives, for each \emph{functional} interface $\II$, $\mtype{\II}{\m}$,  where $\m$ is the only method of $\II$. 
\end{itemize}

The big-step semantics is given in the last part of \cref{fig:FJ-lambda-big-step}. $\MiniFJLambda$ shows an example of instantiation of the framework where configurations include an auxiliary structure, rather than being just language terms.  
In this case, the structure is an \emph{environment} $\env$ (a finite map from variables to values) modelling the current stack frame. 
Furthermore, results are not particular configurations: they are either \emph{objects}, of shape $\obj{\CC}{\vals}$, or $\lambda$-abstractions. 

Rules for $\FJ$ constructs are straightforward. 
Note that, since we only consider upcasts, casts have no runtime effect. 
Indeed, they are guaranteed to succeed on well-typed expressions.
Rule \rn{$\lambda$-invk} shows that, when the receiver of a method is a $\lambda$-abstraction, the method name is not significant at runtime, and the effect is that the body of the function is evaluated as in  the  usual application.

The type system, consisting of judgements for configurations, expressions and values,  is given in \cref{fig:FJ-bss-ex-lambda-types}. 
The following assumptions formalize standard $\FJ$ typing constraints on the class table.
\begin{description}
\item[\MBodyCond] Method bodies are well-typed with respect to method types:
\begin{itemize}
\item either $\mbody{\CC}{\m}$ and $\mtype{\CC}{\m}$ are both undefined
\item or $\mbody{\CC}{\m}=\Pair{\x_1\ldots\x_n}{\e}$,  
$\mtype{\CC}{\m}=\funtype{\T_1\ldots\T_n}{\T}$, and ${\HasType{\x_1{:}\T_1,\ldots,\x_n{:}\T_n,\kwthis{:}\CC}{\e}{\T}}$. 
\end{itemize}
\item[\InhFldCond]  Fields are inherited, no field hiding:\\
if  $\T\leqfj\T'$, and $\fields{\T'}=\Field{\T_1}{\f_1}\ldots\Field{\T_n}{\f_n}$, then
\Space $\fields{\T}=\Field{\T_1}{\f_1}\ldots\Field{\T_m}{\f_m}$, $m\geq n$, and $\f_i\neq\f_j$ for $i\neq j$.
\item[\InhMethCond] Methods are inherited, no method overloading, invariant overriding:\\
if $\T\leqfj\T'$, and $\mtype{\T'}{\m}$ is defined, then $\mtype{\T}{\m}=\mtype{\T'}{\m}$.
\end{description}

\EZComm{Ho messo dei nomi simbolici, che si potrebbero anche rendere pi\`u espressivi, nel caso si riesca a fare riferimento a queste assunzioni dove servono nelle prove}

\begin{figure}
\begin{small}
\begin{math}
\begin{array}{c}
\MetaRule{t-conf}{\HasType{\es}{\val_i}{\T_i}\Space \forall i\in 1..n\Space\Space \HasType{\x_1{:}\T'_1,\ldots,\x_n{:}\T'_n}{\e}{\T}}{\HasType{\es}{\Conf{\x_1{:}\val_1,\ldots,\x_n{:}\val_n}{\e}}{\T}}
{\T_i\leqfj\T'_i\Space \forall i\in 1..n }
\\[4ex]
\MetaRule{t-var}{}{\HasType{\Gamma}{\x}{\T}}{\Gamma(\x)=\T}
\BigSpace
\MetaRule{t-upcast}{\HasType{\Gamma}{\e}{\T}}{\HasType{\Gamma}{\Cast{\T}{\e}}{\T}}{}
\\[4ex] 
\MetaRule{t-field-access}{\HasType{\Gamma}{\e}{\CC}}{\HasType{\Gamma}{\FieldAccess{\e}{\f}}{\T_i}}
{\begin{array}{l}
\fields{\CC}=\Field{\T_1}{\f_1}\ldots\Field{\T_n}{\f_n}\\
i\in 1..n
\end{array}
}\\[4ex]
\MetaRule{t-new}{\HasType{\Gamma}{\e_i}{\T_i}\Space \forall i\in 1..n}{\HasType{\Gamma}{\ConstrCall{\CC}{\e_1,\ldots,\e_n}}{\CC}}
{\begin{array}{l}
\fields{\CC}=\Field{\T_1}{\f_1}\ldots\Field{\T_n}{\f_n}\\
\end{array}
}
\\[4ex]
\MetaRule{t-invk}{\HasType{\Gamma}{\e_i}{\T_i}\Space \forall i\in 0..n}{\HasType{\Gamma}{\MethCall{\e_0}{\m}{\e_1,\ldots,\e_n}}{\T}}
{\begin{array}{l}
\e_0\ \mbox{not of shape}\ \LambdaExp{\xs}{\e}\\
\mtype{\T_0}{\m}=\funtype{\T_1\ldots\T_n}{\T}\\\
\end{array}
}
\\[4ex]
\MetaRule{t-$\lambda$}{\HasType{
\x_1{:}\T_1,\ldots,\x_n{:}\T_n}{\e}{\T}}{\HasType{\Gamma}{\LambdaExp{\x_1\ldots\x_n}{\e}}{\II}}
{\begin{array}{l}
\umtype{\II}=\funtype{\T_1\ldots\T_n}{\T}
\end{array}
}
\\[4ex]
\MetaRule{t-object}{\HasType{\Gamma}{\val_i}{\T'_i}\Space\forall i\in 1..n}{\HasType{\Gamma}{\obj{\CC}{\val_1,\ldots,\val_n}}{\CC}{}}
{\begin{array}{l}
\fields{\CC}=\Field{\T_1}{\f_1}\ldots\Field{\T_n}{\f_n}\\
\T'_i\leqfj\T_i\Space \forall i\in 1..n
\end{array}}
\\[4ex]
\MetaRule{t-sub}{\HasType{\Gamma}{\e}{\T}}{\HasType{\Gamma}{\e}{\T'}}{
\begin{array}{l}
\e\ \mbox{not of shape}\ \LambdaExp{\xs}{\e}\\
\T\leqfj\T'
\end{array}}
\end{array}
\end{math}
\end{small}
\caption{$\MiniFJLambda$: type system}\label{fig:FJ-bss-ex-lambda-types}
\end{figure}

Besides the standard typing features of $\FJ$, the $\MiniFJLambda$ type system ensures the following. 
\begin{itemize}
\item A functional interface $\II$ can be assigned as type to a $\lambda$-abstraction which has the 
functional type of the method,
see rule \rn{t-$\lambda$}.
\item A $\lambda$-abstraction should have a \emph{target type} determined by the context where the $\lambda$-abstraction occurs. 
More precisely, as described by \citet[p.~602]{GoslingEtAl14}, a $\lambda$-abstraction in our calculus can only occur as  return expression of a method or argument of constructor, method call or cast.  
Then, in some contexts a $\lambda$-abstraction cannot be typed,
in our calculus when  occurring as receiver in field access or method invocation, hence these cases should be prevented.
This is implicit in rule \rn{t-field-access}, since the type of the receiver should be a class name, whereas it is explicitly forbidden in rule \rn{t-invk}. 
Finally, 
a $\lambda$-abstraction cannot be the main expression of a program, as also in this case the target type is not well defined. 
For simplicity, this requirement is not enforced by typing rules, but it can be easily recovered as an assumption on the source program. 
\item A $\lambda$-abstraction with a given target type $\JJ$ should have type \emph{exactly} $\JJ$: a subtype $\II$ of $\JJ$ is not enough. 
Consider, for instance, the following class table:
\begin{lstlisting}
interface J {}
interface I extends J { A m(A x); }
class C { J f; } 
class D {
  D m(I y) { return new D().n(y); }
  D n(J y) { return new D(); }
}
\end{lstlisting} 
In the main expression $\MethCall{\ConstrCall{\DD}{}}{\n}{\lambda\x.\x}$, 
the $\lambda$-abstraction has target type $\JJ$, which is \emph{not} a functional interface, hence the expression is ill-typed in Java (the compiler has no functional type against which to typecheck the {$\lambda$-abstraction}). On the other hand, in the body of method $\m$, the parameter $y$ of type $\II$ can be passed, as usual, to method $\n$ expecting a supertype. For instance, the main expression $\MethCall{\ConstrCall{\DD}{}}{\m}{\lambda\x.\x}$ is well-typed, since the $\lambda$-abstraction has target type $\II$, and can be safely passed to method $\n$, since it is not used as function there. To formalise this behaviour, it is forbidden to apply subsumption to {$\lambda$-abstraction}s, see rule \rn{t-sub}.
\item However, $\lambda$-abstractions occurring as results rather than  in source code (that is, in the environment and as fields of objects) are allowed to have a subtype of the required type, see the explicit side condition in rules \rn{t-conf} and \rn{t-object}. 
For instance, in the above class table, 
the expression 
$\ConstrCall{\CC}{\Cast{\II}{\Lam{\x}\x}}$ is well-typed, whereas $\ConstrCall{\CC}{\Lam{\x}\x}$ is ill typed, since rule \rn{t-sub} cannot be applied to {$\lambda$-abstraction}s. When the expression is evaluated, the result is $\obj{\CC}{\Lam{\x}\x}$, which is well-typed.  
\end{itemize}
As mentioned at the beginning, the obvious small-step semantics would produce not typable expressions. {In the above example, } we get 
\[ \ConstrCall{\CC}{\Cast{\II}{\Lam{\x}\x}}\longrightarrow \ConstrCall{\CC}{\Lam{\x}\x}\longrightarrow\obj{\CC}{\Lam{\x}\x} \]
and $\ConstrCall{\CC}{\Lam{\x}\x}$ has no type, while $\ConstrCall{\CC}{\Cast{\II}{\Lam{\x}\x}}$ and $\obj{\CC}{\Lam{\x}\x}$ have type $\CC$.

As expected, to show soundness (\cref{thm:sdfjl}) lemmas of inversion and canonical forms are handy: they can be easily proved as usual. 
Instead, we do not need a substitution lemma, since environments associate variables \DA{with} values. 

\begin{lemma}[Inversion]\label{lem:il} 
The following hold:
\begin{enumerate}
\item\label{lem:il:1} If $\HasType{ }{\Conf{\x_1{:}\val_1,\ldots,\x_n{:}\val_n}{\e}}{\T}$, then $\HasType{\x_1{:}\T_1,\ldots,\x_n{:}\T_n}{\e}{\T}$, $\HasType{\es}{\val_i}{\T'_i}$ and  $\T'_i\leqfj\T_i$ for all $i \in 1..n$. 
\item\label{lem:il:3} If  $\HasType{\Gamma}{\x}{\T}$, then $\Gamma(\x)\leqfj\T$.
\item\label{lem:il:4} If  $\HasType{\Gamma}{\FieldAccess{\e}{\f_i}}{\T}$, then $\HasType{\Gamma}{\e}{\CC}$ and {\em $\fields{\CC}=\Field{\T_1}{\f_1}\ldots\Field{\T_n}{\f_n}$} and $\T_i\leqfj T$ where 
$i\in 1..n$.
\item\label{lem:il:5} If  $\HasType{\Gamma}{\ConstrCall{\CC}{\e_1,\ldots,\e_n}}{\T}$, then $\CC\leqfj\T$ and {\em $\fields{\CC}=\Field{\T_1}{\f_1}\ldots\Field{\T_n}{\f_n}$} and $\HasType{\Gamma}{\e_i}{\T_i}$ for all $i\in 1..n$.
\item\label{lem:il:6} If  $\HasType{\Gamma}{\MethCall{\e_0}{\m}{\e_1,\ldots,\e_n}}{\T}$, then $\e_0$ not of shape $\LambdaExp{\xs}{\e}$ and 
$\HasType{\Gamma}{\e_i}{\T_i}$ for all $i\in 0..n$ and {\em $\mtype{\T_0}{\m}=\funtype{\T_1\ldots\T_n}{\T'}$} with $\T'\leqfj\T$.
\item\label{lem:il:7} If  $\HasType{\Gamma}{\LambdaExp{\xs}{\e}}{\T}$, then $\T=\II$ and {\em $\umtype{\II}=\funtype{\T_1\ldots\T_n}{\T'}$} and $\HasType{\x_1{:}\T_1,\ldots,\x_n{:}\T_n}{\e}{\T'}$.
\item\label{lem:il:8} If  $\HasType{\Gamma}{\Cast{\T'}{\e}}{\T}$, then $\HasType{\Gamma}{\e}{\T'}$ and $\T'\leqfj \T$.
\item\label{lem:il:9} If  $\HasType{\Gamma}{\obj{\CC}{\val_1,\ldots,\val_n}}{\T}{}$, then $\CC\leqfj\T$ and {\em $\fields{\CC}=\Field{\T_1}{\f_1}\ldots\Field{\T_n}{\f_n}$} and $\HasType{\Gamma}{\val_i}{\T'_i}$ and $\T'_i\leqfj\T_i $ for all $i\in 1..n$.
\end{enumerate}
\end{lemma}

\begin{lemma}[Canonical Forms]\label{lem:cfj} 
The following hold:
\begin{enumerate}
\item \label{lem:cfj:1}
If $\HasType\es \val {\CC}$, then $\val=\obj{\DD}{\vals}$ and $\DD\leqfj\CC$.
\item \label{lem:cfj:2}
If $\HasType\es \val \II$, then either $\val=\obj{\CC}{\vals}$ and $\CC\leqfj\II$ or $\val=\LambdaExp{\xs}{\e}$ and $\II$ is a functional interface.
\end{enumerate}
\end{lemma}

We write $\HasType{\Gamma}{\e}{\leqfj\T}$ as short for $\HasType{\Gamma}{\e}{\T'}$ and $\T'\leqfj\T$ for some $\T'$. 
In order to state soundness, set \ple{\ConfSet_2,\ResSet_2,\RuleSet_2} the big-step semantics defined in \cref{fig:FJ-lambda-big-step}, and 
let $\ConfPredTwo_\T = \{ \Conf{\env}{\e} \in \ConfSet_2 \mid \HasType{\es}{\Conf{\env}{\e}}{\leqfj\T} \}$ and 
$\ResPredTwo_\T = \{ \val \in \ResSet_2 \mid \HasType{\es}{\val}{\leqfj\T} \}$, for $\T$ defined  in \cref{fig:FJ-lambda-syntax}.  

\begin{theorem}[Soundness]\label{thm:sdfjl}
The big-step semantics \ple{\ConfSet_2,\ResSet_2,\RuleSet_2} and the indexed predicate $\PredTwo$ satisfy the conditions \LPCond, \ExPCond and \AllPCond of \cref{sect:bss-sc}.
\end{theorem}
\begin{proof}
\proofsect{Proof of \LPCond} 
The proof is by cases on instantiations of meta-rules. 
In all such cases, we have a configuration $\Conf{\env}\e$ in the conclusion, with $\env = y_1{:}\hat\val_1,\ldots,y_p{:}\hat\val_p$,  such that 
$\HasType{}{\Conf{y_1{:}\hat\val_1,\ldots,y_p{:}\hat\val_p}\e}{\leqfj\hat\T}$, hence, 
by \refItem{lem:il}{1},  we get $\HasType\es{\hat\val_\ell}{\leqfj\hat\T_\ell}$ for all $\ell\in 1..p$ and $\HasType{\Gamma}{\e}{\T}$ with $\Gamma = y_1{:}\hat\T_1,\ldots,y_p{:}\hat\T_p$ and $\T\leqfj\hat\T$, for some $\hat\T_1,\ldots,\hat\T_p$. 
\begin{proofcases}
\item [\rn{var}] 
\refItem{lem:il}{3} applied to $\HasType{\Gamma}{\x}{\T}$ implies $\x = y_i$ and $\hat\T_i\leqfj\T$ for some $i \in 1..p$. 
Then, the thesis follows by transitivity of subtyping since $\env(\x) = \hat\val_i$ and $\HasType{}{\hat\val_i}{\leqfj\hat\T_i}$. 
\item [\rn{field-access}] \refItem{lem:il}{4} applied to $\HasType{\Gamma}{\FieldAccess{\e}{\f_i}}{\T}$ implies $\HasType{\Gamma}{\e}{\DD}$ and $\fields{\DD}=\Field{\T_1}{\f_1}\ldots\Field{\T_m}{\f_m}$ and $\T_i\leqfj T$ where 
$i\in 1..m$.  
Since $\eval{\Conf{\env}\e}{\obj{\CC}{\val_1,\ldots,\val_n}}$ is a premise we assume $\HasType{}{\obj{\CC}{\val_1,\ldots,\val_n}}{\leqfj\DD}$, which implies $\CC\leqfj\DD$ and $\fields{\CC}=\Field{\T'_1}{\f'_1}\ldots\Field{\T'_n}{\f'_n}$ and $\HasType{\Gamma}{\val_j}{\leqfj\T'_j}$ for all $j\in 1..n$ by \refItem{lem:il}{9}. From $\CC\leqfj\DD$ and assumption \InhFldCond we have $m\leq n$ and $\T_j=\T'_j$ and $\f_j=\f'_j$ for all $j\in 1..m$.  
We conclude $\HasType{}{\val_i}{\leqfj\T}$. 

\item [\rn{new}] 
\refItem{lem:il}{5} applied to $\HasType{\Gamma}{\ConstrCall{\CC}{\e_1,\ldots,\e_n}}{\T}$ implies $\CC\leqfj\T$ and $\fields{\CC}=\Field{\T_1}{\f_1}\ldots\Field{\T_n}{\f_n}$ and $\HasType{\Gamma}{\e_i}{\T_i}$ for all $i\in 1..n$. 
Since $\eval{\Conf{\env}{\e_i}}{\val_i}$ is a premise we assume $\HasType{}{\val_i}{\leqfj\T_i}$  for all $i\in 1..n$. 
Using rule \rn{t-object} we derive $\HasType{}{\obj{\CC}{\val_1,\ldots,\val_n}}{\leqfj\T}$. 

\item [\rn{invk}] 
\refItem{lem:il}{6} applied to $\HasType{\Gamma}{\MethCall{\e_0}{\m}{\e_1,\ldots,\e_n}}{\T}$ implies $\e_0$ not of shape $\LambdaExp{\xs}{\e}$ and 
$\HasType{\Gamma}{\e_i}{\T_i}$ for all $i\in 0..n$ and $\mtype{\T_0}{\m}=\funtype{\T_1\ldots\T_n}{\T'}$ with $\T'\leqfj\T$. 
Since  $\eval{\Conf{\env}{\e_0}}{\obj{\CC}{\vals'}}$ is a premise we assume $\HasType{}{\obj{\CC}{\vals'}}{\leqfj\T_0}$, which implies $\CC\leqfj\T_0$ by \refItem{lem:il}{9}. 
Since $\eval{\Conf{\env}{\e_i}}{\val_i}$ is a premise we assume 
$\HasType{}{\val_i}{\leqfj\T_i}$ for all $i\in 1..n$. 
We have $\mtype{\CC}{\m}=\funtype{\T_1\ldots\T_n}{\T'}$ since $\mtype{\T_0}{\m}=\funtype{\T_1\ldots\T_n}{\T'}$ and $\CC\leqfj\T_0$ by assumption \InhMethCond. 
By assumption \MBodyCond, $\HasType{x_1{:}\T_1,\ldots,\x_n{:}\T_n,\kwthis{:}\CC}{\e}{\T'}$. 
Therefore, by rule \rn{t-conf} and since $\T'\leqfj \T$,  we can derive $\HasType\es{\Conf{\x_1{:}\val_1,\ldots,\x_n{:}\val_n,\kwthis{:}\obj{\CC}{\vals'}}{\e}}{\leqfj\T}$. 

\item [\rn{$\lambda$-invk}] 
\refItem{lem:il}{6} applied to $\HasType{\Gamma}{\MethCall{\e_0}{\m}{\e_1,\ldots,\e_n}}{\T}$ implies 
$\HasType{\Gamma}{\e_i}{\T_i}$ for all $i\in 0..n$ and $\mtype{\T_0}{\m}=\funtype{\T_1\ldots\T_n}{\T'}$ with $\T'\leqfj\T$. 
Since $\eval{\Conf{\env}{\e_0}}{\LambdaExp{\xs}{\e}}$ is a premise we assume $\HasType{}{\LambdaExp{\xs}{\e}}{\leqfj\T_0}$, which implies $\II\leqfj\T_0$ and $\umtype{\II}=\funtype{\T_1\ldots\T_n}{\T'}$ and $\HasType{\x_1{:}\T_1,\ldots,\x_n{:}\T_n}{\e}{\T'}$ by \refItem{lem:il}{7}. 
Since  $\eval{\Conf{\env}{\e_i}}{\val_i}$ is a premise we assume 
$\HasType{}{\val_i}{\leqfj\T_i}$ for all $i\in 1..n$.  
Therefore we derive $\HasType\es{\Conf{\x_1{:}\val_1,\ldots,\x_n{:}\val_n}{\e}}{\leqfj\T}$.

\item [\rn{$\lambda$}] 
The thesis is trivial as the configuration and the final result are the same. 

\item [\rn{upcast}] 
\refItem{lem:il}{8} applied to $\HasType{\Gamma}{\Cast{\T'}{\e}}{\T}$ implies $\HasType{\Gamma}{\e}{\leqfj \T}$. 
From $\eval{\Conf{\env}\e}\val$ we conclude $\HasType\es\val{\leqfj\T}$. 
\end{proofcases}

\proofsect{Proof of \ExPCond}  
It is easy to verify that if $\HasType{\es}{\Conf{\env}{\e}}{\leqfj\T}$, then there is a rule in  \cref{fig:FJ-lambda-big-step}, whose conclusion is $\Conf\env\e$, just because for every syntactic construct there is a corresponding rule and side conditions in typing rules imply those of big-step rules. 
The only less trivial case is that of variables: 
if $\HasType{\es}{\Conf{\env}{\x}}{\leqfj\T}$, then by \cref{lem:il}~(\ref{lem:il:1},\ref{lem:il:3}), $\x \in \dom(\env)$, hence rule \rn{var} is applicable, as the side condition is satisfied. 

\proofsect{Proof of \AllPCond}  
Rule \rn{field-access} requires that $\Conf{\env}{\e}$ reduces to an object with a field $\f_i$, and this is assured by the typing rule \rn{t-field-access}, which prescribes a class type for the expression $\e$ with the field $\f_i$, together with the validity of condition \LPCond (which assures type preservation by \cref{lem:sr}) and  \refItem{lem:cfj}{1}. 
For a well-typed method call $\MethCall{\e_0}{\m}{\e_1,\ldots,\e_n}$ the configuration $\Conf{\env}{\e_0}$ can reduce either to an object or to a $\lambda$-expression. 
In the first case we can apply rule \rn{invk} and in the second case rule \rn{$\lambda$-invk}. 
In both cases the typing assures that the arguments are in the right number, while the condition is trivial for the last premise. 
\end{proof}

\subsection{Intersection and union types}
\label{sect:bss-iut}

We enrich the type system of \cref{fig:bss-ex-lambda-types} by adding intersection and union type constructors and the corresponding typing rules, see \cref{fig:iutypes}. 
Intersection types for the $\lambda$-calculus have been widely studied, \eg by \citet{BarendregtDS13}. 
Union types naturally model conditionals \citep{Grudzinski00} and non-deterministic choice \citep{DezaniLP98}. 

The production in the top section of \cref{fig:iutypes} is again interpreted coinductively to allow possibly infinite types, but, 
as usual with recursive types, we only consider \emph{contractive} types \cite{Pierce02}, that is, we require an infinite number of arrows in each infinite path in a type (viewed as a tree). 
On the other hand, typing rules are still interpreted inductively. 

\begin{figure}
\begin{small}
\begin{math}
\begin{grammatica}
\production{\tA}{\natType\mid \tA_1\to \tA_2\mid \tA_1\wedge \tA_2\mid \tA_1\vee \tA_2}{type}
\end{grammatica}
\end{math}
\HSep
\begin{math}
\begin{array}{c}
\MetaRule
{$\wedge$ I}
{\HasType{\Gamma}{\e}{\T}\Space\HasType{\Gamma}{\e}{\Ts}}
{\HasType{\Gamma}{\e}{\T\wedge\Ts}}{}
\BigSpace
\MetaRule{$\wedge$ E}{\HasType{\Gamma}{\e}{\T\wedge\Ts}}{\HasType{\Gamma}{\e}{\T}}{}\BigSpace
\MetaRule{$\wedge$ E}{\HasType{\Gamma}{\e}{\T\wedge\Ts}}{\HasType{\Gamma}{\e}{\Ts}}{}\\[3ex]
\MetaRule{$\vee$ I}{\HasType{\Gamma}{\e}{\T}}{\HasType{\Gamma}{\e}{\T\vee\Ts}}{}\BigSpace
\MetaRule{$\vee$ I}{\HasType{\Gamma}{\e}{\Ts}}{\HasType{\Gamma}{\e}{\T\vee\Ts}}{}
\end{array}
\end{math}
\end{small}
\caption{Intersection and union types: syntax and typing rules}\label{fig:iutypes}
\end{figure}

The typing rules for the introduction and the elimination of intersection and union are standard,  except  for the absence of the union elimination rule:
\[
\MetaRule{$\vee E$} {\HasType{\SubstFun\Gamma\tA\x}
\e\tC\Space \HasType{\SubstFun\Gamma\tB\x} \e\tC\Space
\HasType\Gamma {\e'}{\tA\vee \tB}}{\HasType\Gamma 
{\subst\e{\e'}\x}\tC}{} 
\]
As a matter of fact, rule \rn{$\vee E$} is unsound for $\oplus$. 
For example, let split the type $\natType$ into $\even$ and $\odd$ and add the expected typings for natural numbers. 
The prefix addition $\mathtt+$ has type
$(\funtype{\funtype\even\even}\even)\wedge(\funtype{\funtype\odd\odd}\even)$ 
and we  derive
\[\footnotesize
\Infer{
  \HasTypeNarrow{x{:}\even}{+\,x\,x}{\even}
  \Space 
  \HasTypeNarrow{x{:}\odd}{+\,x\,x}{\even}
  \Space 
  \Infer{ 
    \Infer{ 
      \HasType{}1\odd
    }{ \HasType{}1{\even\vee\odd} }{} 
    \Space 
    \Infer{ 
      \HasType{}2\even
    }{ \HasType{}2{\even\vee\odd} }{}
  }{ \HasType{}{(1\oplus2)}{\even\vee\odd} }{} 
}{ \HasType{}{{+}(1\oplus2)(1\oplus2)}\even }{} 
\]

We cannot assign the type $\even$ to $3$, which is a possible result, so strong soundness is lost. 
In addition, in the small-step approach, we cannot assign $\even$ to the intermediate term ${+}\,1\,2$, so subject reduction fails.  
In the big-step approach, there is no such intermediate term; however, condition \LPCond fails for the big-step rule for $+$. 
Indeed, considering the following instantiation of the rule:
\[
\MetaRule{$+$}{\eval{1\oplus2}{1}\Space \eval{1\oplus2}{2} }{ \eval{{+}(1\oplus2)(1\oplus2)}{3}}{}
\]
and the type $\even$ for the conclusion: we cannot assign this type to the final result as required by \LPCond (\cf \refItem{def:sc-lp}{2}). 

Intersection types allow to derive meaningful types also for expressions containing variables applied to themselves, for example we can derive 
$\HasType{}{\Lam{\x}\x\appop\x}{\funType{(\funType{\tA}{\tB})\wedge\tA}{\tB}}$. 
With union types all non-deterministic choices between typable expressions can be typed too, since we can derive $\HasType\Gamma {\e_1\oplus\e_2}{\T_1\vee\T_2}$ from $\HasType\Gamma {\e_1}{\T_1}$ and  $\HasType\Gamma {\e_2}{\T_2}$.

We now state standard lemmas for the type system, which are handy towards the soundness proof. 
We first define the \emph{subtyping} relation $\tA\leq \tB$ as the smallest preorder such that:
\begin{itemize}
\item $\tB\leq \tA_1$ and $\tB\leq \tA_2$ imply $\tB\leq \tA_1\wedge \tA_2$;
\item $\tA\wedge \tB\leq \tA$ and $\tA\wedge \tB\leq \tB$;
\item $\tA\leq \tA \vee \tB$ and $\tA\leq \tB\vee \tA$.
\end{itemize}
It is easy to verify that $\tA\leq\tB$ iff $\HasType{\Gamma,\x{:}\tA}{\x}{\tB}$ for an arbitrary variable $\x$, using rules  \rn{$\wedge$I}, \rn{$\wedge$E} and \rn{$\vee$I}.

\begin{lemma}[Inversion]\label{lem:iliu}
The following hold: 
\begin{enumerate}
\item \label{lem:iliu:1}If $\HasType\Gamma\x\tA$, then $\Gamma(\x)\leq \tA$.
\item \label{lem:iliu:2} If $\HasType{\Gamma}{\natconst}\T$, then $\natType\leq\T$.
\item \label{lem:iliu:3} If $\HasType\Gamma{\LambdaExp{\x}{\e}}\T$, then $\HasType{\SubstFun\Gamma{\tB_i}\x}\e {\tC_i}$ for $i\in 1..m$ and
$\bigwedge_{i\in 1..m}(\funType{\tB_i} {\tC_i})\leq \tA$.
\item \label{lem:iliu:4} If $\HasType\Gamma{\e_1\appop\e_2}\tA$, then $\HasType\Gamma  {\e_1} {\funType{\tB_i}{\tC_i}}$  and $\HasType\Gamma{\e_2}{\tB_i}$    for $i\in 1..m$ and $\bigwedge_{i\in 1..m}\tC_i\leq \tA$.
\item \label{lem:iliu:5} If $\HasType{\Gamma}{\SuccExp{\e}}\T$, then $\natType\leq\T$ and $\HasType{\Gamma}{\e}{\natType}$.
\item \label{lem:iliu:6} If $\HasType\Gamma {\e_1\oplus\e_2}\tA$, then $\HasType\Gamma{\e_i}\tA'$ with $\tA' \leq\tA$ and $i\in1..2$.
\end{enumerate}
\end{lemma}

\begin{lemma}[Substitution]\label{lem:siu}
If $\HasType{\SubstFun\Gamma{\tA'}\x} \e {\tA}$ and $\HasType\Gamma{\e'} {\tA'}$, then $\HasType\Gamma{\subst\e{\e'}\x} {\tA}$.
\end{lemma}

\begin{lemma}[Canonical Forms]\label{lem:cfiu}
The following hold: 
\begin{enumerate}
\item \label{lem:cfiu:1}
If $\HasType\es \val {\funType{\tA'}{\tA}}$, then $\val=\LambdaExp{\x}{\e}$.
\item \label{lem:cfiu:2}
If $\HasType\es \val \natType$, then $\val=\natconst$.
\end{enumerate}
\end{lemma}

In order to state soundness,  let
$\ConfPredThr_\T = \{ \e \in \ConfSet_1 \mid \HasType{\es}{\e}{\T} \}$ and 
$\ResPredThr_\T = \{ \val \in \ResSet_1 \mid \HasType{\es}{\val}{\T} \}$, for $\T$  defined in \cref{fig:iutypes}. 

\begin{theorem}[Soundness]\label{thm:sdiu}
The big-step semantics \ple{\ConfSet_1,\ResSet_1,\RuleSet_1}  and the indexed predicate $\PredThr$ satisfy the conditions \LPCond, \ExPCond and \AllPCond of \cref{sect:bss-sc}.
\end{theorem}
\begin{proofsketch} 
We prove conditions only for rule \rn{app}, the other cases are similar (\cf proof of \cref{thm:sd}). 

\proofsect{Proof of \LPCond} 
The proof is by cases on instantiations of  meta-rules. 
For rule \rn{app} \refItem{lem:iliu}{4}  applied to $\HasType\es{\e_1\appop\e_2}\tA$ implies $\HasType{\es}  {\e_1} {\funType{\tB_i}{\tC_i}}$  and $\HasType{\es}{\e_2}{\tB_i}$  for $i\in 1..m$ and $\bigwedge_{i\in 1..m}\tC_i\leq \tA$. 
Now, from assumptions of \LPCond, we get 
$\HasType{\es}{\Lam{\x}\e}{\funType{\tB_i}{\tC_i}}$ and $\HasType{\es}{\val_2} {\tB_i}$ for $i\in 1.. m$. 
\refItem{lem:iliu}{3} implies $\HasType{\x {\ : \ } \tB_i} \e {\tC_i}$, so by \cref{lem:siu} we have $\HasType{\es}{\subst\e{\val_2}\x} {\tC_i}$ for $i\in 1.. m$. 
We can derive $\HasType{\es}{\subst\e{\val_2}\x} {\tA}$ using rules \rn{$\wedge$I}, \rn{$\wedge$E} and \rn{$\vee$I}.

\proofsect{Proof of \ExPCond} The proof is as in  \cref{thm:sd}.

\proofsect{Proof of \AllPCond}  
The proof is by cases on instantiations of  meta-rules. 
For rule \rn{app} \refItem{lem:iliu}{4}  applied to $\HasType{\es}{\e_1\appop\e_2}\tA$ implies $\HasType{\es}  {\e_1} {\funType{\tB_i}{\tC_i}}$  for $i\in 1..m$. 
If $\eval{\e_1}\val$ we get $\HasType{\es}  {\val} {\funType{\tB_i}{\tC_i}}$  for $i\in 1..m$ by \LPCond and \cref{lem:sr}. 
\refItem{lem:cfiu}{1} applied to $\HasType{\es}{\val} {\funType{\tB_i}{\tC_i}}$ implies $\val=\LambdaExp{\x}{\e}$ as needed. 
\end{proofsketch}

\subsection{$\MiniFJOS$}
\label{sect:bss-fjos}

A well-known example in  which  
proving soundness with respect to small-step semantics is extremely challenging is  the standard type system with intersection and union types \citep{BarbaneraDL95}
w.r.t. the pure $\lambda$-calculus with full reduction.  
Indeed, the standard subject reduction technique fails\footnote{For this reason, \citet{BarbaneraDL95} prove soundness by an  ad-hoc  
technique, that is, by considering parallel reduction and an equivalent type system \`a la Gentzen, which enjoys the cut elimination property. }, 
since, for instance, we can derive  the type 
\[ (\tA\to \tA\to \tC)\wedge (\tB\to \tB\to \tC)\to (\tD\to \tA \vee \tB)\to \tD \to \tC \] 
for both 
$\Lam{x} \Lam{y} \Lam{z} x\appop((\Lam{t}t)\appop(y\appop z))\appop((\Lam{t}t)\appop(y\appop z))$ 
and
$\Lam{x} \Lam{y} \Lam{z} x\appop(y\appop z)\appop (y\appop z)$, 
 but the intermediate  expressions  
 $\Lam{x} \Lam{y} \Lam{z} x\appop((\Lam{t}t)\appop(y\appop z))\appop(y\appop z)$ and $\Lam{x} \Lam{y} \Lam{z} x\appop(y\appop z)\appop((\Lam{t}t)\appop(y\appop z))$ do not have this type. 

As the example shows, the key problem is that rule \rn{$\vee$E} can be applied to  expression  
$\e$ where the same  subexpression  
$\e'$ occurs more than once. 
In the non-deterministic case, as shown by the example in the previous section, this is unsound, since $\e'$ can reduce to different values. 
In the deterministic case, instead, this is sound, but cannot be proved by subject reduction. 
Since using big-step semantics there are no intermediate steps to be typed, our approach seems very promising to investigate an alternative 
proof of soundness.
Whereas we leave this challenging problem to future work, here as first step we describe a calculus with a much simpler version of the problematic feature.

The calculus is a variant of $\FJ\vee$, introduced by \citet{IgarashiN07}, an extension of $\FJ$  \citep{IgarashiPW01} with union types. 
As discussed more extensively by \citet{IgarashiN07}, this gives  the ability to define a supertype even after a class hierarchy is fixed, grouping independently developed classes with similar interfaces. 
In fact, given some types, their union type can be viewed as an interface type that ``factors out'' their common features.
With respect to $\FJ\vee$, we do not consider cast and type-case constructs and, more importantly, in the typing rules we handle differently union types, taking inspiration directly from rule \rn{$\vee$E} of the $\lambda$-calculus. 
With this approach, we enhance the expressivity of the type system, since it becomes possible to eliminate unions simultaneously for an arbitrary number of arguments, including the receiver,  in a method invocation, provided that they are all equal to each other.  
We dub this calculus $\MiniFJOS$. 

\begin{figure} 
\begin{small}
\begin{math}
\begin{grammatica}
\produzione{\e}{\x \mid \FieldAccess{\e}{\f} \mid \ConstrCall{\CC}{\e_1,\ldots,\e_n} \mid
\MethCall{\e}{\m}{\e_1,\ldots,\e_n} }{expression}\\
\seguitoproduzione{ \ifte{\e}{\e_1}{\e_2} \mid \kwtt \mid \kwff }{}\\
\produzione{\val}{\ConstrCall{\CC}{\val_1,\ldots,\val_n}\mid\kwtt\mid\kwff}{value}\\
\produzione{\T}{\CC\mid\boolType\mid \T_1\vee\T_2}{type}\\
\produzione{\ectx}{\FieldAccess{[]}{\f} \mid \MethCall{[]}{\m}{[],\ldots,[],\e_1,\ldots,\e_n}}{elimination context} 
\end{grammatica}
\end{math}
\HSep
\begin{math}
\begin{array}{c}
\MetaRule{field}{\eval{\e}{\ConstrCall{\CC}{\val_1,\ldots,\val_n}}}
{ \eval{\FieldAccess{\e}{\f_i}}{\val_i}}
{ \fields{\CC}=\Field{\T_1}{\f_1}\ldots\Field{\T_n}{\f_n}\\
  i\in 1..n  }
\\[4ex]
\MetaRule{new}{
\eval{\e_i}{\val_i}\Space \forall i\in 1..n}{\eval{\ConstrCall{\CC}{\e_1,\ldots,\e_n}}{\ConstrCall{\CC}{\val_1,\ldots,\val_n}}}
{}\\[4ex]
\MetaRule{invk}{
\begin{array}{l}
\eval{\e_0}{\ConstrCall{\CC}{\vals'}}\\
\eval{\e_i}{\val_i}\Space \forall i\in 1..n\\
\eval{\subst{\subst{\subst\e{\val_1}{\x_1}\ldots}{\val_n}{\x_n}}{\ConstrCall{\CC}{\vals'}}{\kwthis}}{\val}
\end{array}}
{ \eval{\MethCall{\e_0}{\m}{\e_1,\ldots,\e_n}}{\val} }
{ \mbody{\CC}{\m}=\Pair{\x_1\ldots\x_n}{\e} }
\\[4ex]
\MetaRule{true}{}{\eval{\kwtt}{\kwtt}}{} 
\BigSpace
\MetaRule{false}{}{\eval{\kwff}{\kwff}}{} 
\\[4ex]
\MetaRule{if-t}{
  \eval{\e}{\kwtt} \Space 
  \eval{\e_1}{\val} 
}{ \eval{\ifte{\e}{\e_1}{\e_2}}{\val} }{}
\BigSpace 
\MetaRule{if-f}{
  \eval{\e}{\kwff} \Space 
  \eval{\e_2}{\val} 
}{ \eval{\ifte{\e}{\e_1}{\e_2}}{\val} }{}
\end{array}
\end{math}

\HSep

\begin{math}
\begin{array}{c}
\MetaRule{t-var}{}{\HasType{\Gamma}{\x}{\T}}{\Gamma(\x)=\T}
\BigSpace
\MetaRule{t-bool}{}{\HasType{\Gamma}{\bl}{\boolType}}{ \bl\in\{\kwtt,\kwff\} }
\\[4ex]
\MetaRule{t-fld}{\HasType{\Gamma}{\e}{\CC}}{\HasType{\Gamma}{\FieldAccess{\e}{\f_i}}{\T_i}}
{ \fields{\CC}=\Field{\T_1}{\f_1}\ldots\Field{\T_n}{\f_n}\\
  i\in 1..n }
\\[4ex]
\MetaRule{t-new}{\HasType{\Gamma}{\e_i}{\T_i}\Space \forall i\in 1..n}{\HasType{\Gamma}{\ConstrCall{\CC}{\e_1,\ldots,\e_n}}{\CC}}
{ \fields{\CC}=\Field{\T_1}{\f_1}\ldots\Field{\T_n}{\f_n}\\ }
\\[4ex]
\MetaRule{t-invk}{
  \HasType{\Gamma}{\e}{\CC} \Space 
  \HasType{\Gamma}{\e_i}{\T_i} \quad \forall i\in 1..n   
}{\HasType{\Gamma}{\MethCall{\e}{\m}{\e_1,\ldots,\e_n}}{\T}}  
{ \mtype{\CC}{\m} = \funtype{\T_1\ldots\T_n}{\T}    
}
\\[4ex]
\MetaRule{t-if}{\HasType{\Gamma}{\e}{\boolType}\Space\HasType{\Gamma}{\e_1}{\T}\Space\HasType{\Gamma}{\e_2}{\T}}{\HasType{\Gamma}{\ifte\e{\e_1}{\e_2}}{\T}}{}\BigSpace
\MetaRule{t-sub}{\HasType{\Gamma}{\e}{\T}}{\HasType{\Gamma}{\e}{\T'}}
{ \T\leqfj\T' }
\\[4ex] 
\MetaRule{t-$\vee$-elim}{
  \HasType{\Gamma}{\e}{\bigvee_{i\in 1..m} \CC_i} \Space 
  \HasType{\Gamma,\x{:}\CC_i}{\ectx[x]}{\T} \quad \forall i\in 1..m 
}{ \HasType{\Gamma}{\ectx[\e]}{\T} }{ \x\mbox{ fresh} } 
\end{array}
\end{math}
\end{small}

\caption{$\MiniFJOS$: syntax, big-step semantics and type system}\label{fig:FJOS-typesystem}
\end{figure}

\cref{fig:FJOS-typesystem} gives the syntax, big-step semantics and typing rules of $\MiniFJOS$.  
The subtyping relation $\leqfj$ is the reflexive and transitive closure of the union of the $\aux{extends}$ relation and the standard rules for union:
\[ \Rule{}{\T_1\leqfj\T_1\vee\T_2} \BigSpace
   \Rule{}{\T_2\leqfj\T_1\vee\T_2} \BigSpace 
   \Rule{\T_1\leqfj\T\Space\T_2\leqfj\T}{\T_1\vee\T_2\leqfj\T}  
\]

%
The functions \aux{mtype}, \aux{fields} and \aux{mbody} are defined as for $\MiniFJLambda$, apart that here fields, method parameters and return types can be union types as well, still assuming the conditions on the class table \MBodyCond, \InhFldCond, and \InhMethCond. 

Clearly rule \rn{t-$\vee$-elim} is inspired by rule \rn{$\vee$E}, but restricted only to some specific contexts, named \emph{(union) elimination contexts}. 
Elimination contexts are field access and method invocation, where the latter has $n>0$ holes corresponding to the receiver and (for simplicity the first)  $n-1$ parameters. 
Thanks to this restriction,  we are able to prove a standard inversion lemma, which is not known for the general rule in the $\lambda$-calculus. 

Given an elimination context $\ectx$, we denote by $\ectx[\e]$ the expression obtained by filling all holes of $\ectx$ by $\e$. 

This rule allows us to make the type system more ``structural'', with respect to $\FJ$, similarly to what happens in $\FJ\vee$.  
Let us consider the following classes:
\begin{lstlisting}
class C { 
  A f; Object g; 
  C update(A x) {...}
  Bool eq(C x) {..} 
}
class D { 
  A f; 
  D update(A x) {...}
  Bool eq(D x) {...} 
}
\end{lstlisting}
They share a common structure, but they are not related by inheritance (there is no common superclass abstracting shared features), hence in standard $\FJ$ they cannot be handled uniformly. 
By means of \rn{t-$\vee$-elim} this is possible: for instance, we can write a wrapper class that, in a sense, provides the common interface of $\CC$ and $\DD$ ``ex-post''
\begin{lstlisting}
class CorD {
  C $\vee$ D el; 
  A getf() { this.el.f } 
  CorD update(A x) { new CorD(this.el.update(x)) } 
}
\end{lstlisting}
Bodies of methods \lstinline{getf} and \lstinline{update} in class \lstinline{CorD} are well-typed thanks to rule \rn{t-$\vee$-elim}, 
as shown by the following derivation for \lstinline{update}, where $\Gamma = x{:}\aux{A},\kwthis{:}\aux{CorD}$. 
\[\footnotesize 
\Infer{
  \Infer{
    \Infer{}{ \HasType{\Gamma}{\FieldAccess{\kwthis}{\aux{el}}}{\CC\vee\DD} }{} 
    \Space 
    \Infer{
      \HasType{\Gamma,y{:}\CC}{ \MethCall{y}{\aux{update}}{x} }{\CC} 
    }{ \HasType{\Gamma,y{:}\CC}{ \MethCall{y}{\aux{update}}{x} }{\CC\vee\DD} }{} 
    \Space 
    \Infer{
      \HasType{\Gamma,y{:}\DD}{ \MethCall{y}{\aux{update}}{x} }{\DD} 
    }{ \HasType{\Gamma,y{:}\DD}{ \MethCall{y}{\aux{update}}{x} }{\CC\vee\DD} }{} 
  }{ \HasType{\Gamma}{ \MethCall{\FieldAccess{\kwthis}{\aux{el}}}{\aux{update}}{x} }{\CC\vee\DD} }{}
}{ \HasType{\Gamma}{ \ConstrCall{\aux{CorD}}{ \MethCall{\FieldAccess{\kwthis}{\aux{el}}}{\aux{update}}{x}} }{\aux{CorD}} }{}
\]

The above example can be typed in $\FJ\vee$ as well, even though with a different technique.\footnote{When the receiver of a method call has a union type, look-up (function \aux{mtype}) is directly performed and gives a set of method signatures; arguments should comply all parameter types and the type of the call is the union of return types.}
On the other hand, with our more uniform approach inspired by rule \rn{$\vee$E}, we can type examples where the same subexpression having a union type occurs more than once, and soundness relies on the determinism of evaluation, as in the example at the beginning of this section.  

To illustrate this, 
let us consider an example.
Assuming the above class table, 
consider the expression  $\e = \ifte{\kwff}{\ConstrCall{\CC}{\ldots}}{\ConstrCall{\DD}{\ldots}}$. 
By rule \rn{t-if}, the expression $\e$ has type $\CC\vee\DD$, and, by rule \rn{t-$\vee$-elim}, the expression $\MethCall{\e}{\aux{eq}}{\e}$ has type $\boolType$, as shown by the following derivation:
\[
\Infer{
  \Infer{}{\HasType{\es}{\e}{\CC\vee\DD}}{} \Space 
  \Infer{}{\HasType{\x{:}\CC}{\MethCall{\x}{\aux{eq}}{\x}}{\boolType} }{} 
  \Infer{}{\HasType{\x{:}\DD}{\MethCall{\x}{\aux{eq}}{\x}}{\boolType} }{}
}{ \HasType{\es}{\MethCall{\e}{\aux{eq}}{\e}}{\boolType} }{}
\]
This expression cannot be typed in $\FJ\vee$, because there is no way to eliminate the union type assigned to $\e$ when it occurs as an argument. 

Quite surprisingly, subject reduction fails for the expected small-step semantics, even if there are no intersection types, which are the source, together with the \rn{$\vee$E} rules, of the problems in the $\lambda$-calculus. 
Indeed, we have the following small-step reduction:
\[
\MethCall{\e}{\aux{eq}}{\e} \longrightarrow \MethCall{\ConstrCall{\DD}{\ldots}}{\aux{eq}}{\e} \longrightarrow \MethCall{\ConstrCall{\DD}{\ldots}}{\aux{eq}}{\ConstrCall{\DD}{\ldots}} 
\]
where the intermediate expression cannot be typed, because $\e$ has a union type. 
This happens because intersection types are in a sense hidden in the class table: the method \aux{eq} occurs in two different classes with different types, hence, roughly, we could assign it the intersection type
$(\funtype{\CC\,\CC}{\boolType})\wedge(\funtype{\DD\,\DD}{\boolType})$. 

As in previous examples, the soundness proof uses an inversion lemma and a substitution lemma. 
The canonical forms lemma is trivial since the only values 
of type $\CC$ are objects (constructor calls with values as arguments) instances of a subclass. 
In addition, we need a lemma (dubbed ``key'') which assures that a value typed by a union of classes can also be typed by one of these classes. 
The proof of this lemma is straightforward, since values having class types are just new constructors, as shown by canonical forms. 

\begin{lemma}[Substitution]\label{lem:sos}
If $\HasType{\SubstFun\Gamma{\tA'}\x} \e {\tA}$ and $\HasType\Gamma{\e'} {\tA'}$, then $\HasType\Gamma{\subst\e{\e'}\x} {\tA'}$.
\end{lemma}

\begin{lemma}[Canonical forms] \label{lem:cfos}
The following hold:
\begin{enumerate}
\item\label{lem:cfos:1} If $\HasType{\Gamma}{\val}{\boolType}$, then $\val = \kwtt$ or $\val = \kwff$.
\item\label{lem:cfos:2} If $\HasType{\Gamma}{\val}{\CC}$, then $\val = \ConstrCall{\DD}{\val_1,\ldots,\val_n}$ and $\DD\leqfj\CC$. 
\end{enumerate}
\end{lemma}

\begin{lemma}[Inversion]\label{lem:ilos} 
The following hold: 
\begin{enumerate}
\item\label{lem:ilos:1} If  $\HasType{\Gamma}{\x}{\T}$, then $\Gamma(\x)\leqfj\T$.
\item\label{lem:ilos:2} If  $\HasType{\Gamma}{\FieldAccess{\e}{\f}}{\T}$, then $\HasType{\Gamma}{\e}{\bigvee_{i\in 1..m} \CC_i}$ and, for all $i\in 1..m$,  
$\fields{\CC_i} = \Field{\\T_{i1}}{\f_{i1}}\ldots\Field{\T_{in_i}}{\f_{in_i}}$ and $\f = \f_{ik_i}$ and   $\T_{ik_i}\leqfj T$ for some 
$k_i\in 1..n_i$.
\item\label{lem:ilos:3} If  $\HasType{\Gamma}{\ConstrCall{\CC}{\e_1,\ldots,\e_n}}{\T}$, then $\CC\leqfj\T$ and $\fields{\CC}=\Field{\T_1}{\f_1}\ldots\Field{\T_n}{\f_n}$ and $\HasType{\Gamma}{\e_i}{\T_i}$ for all $i\in 1..n$.
\item\label{lem:ilos:4} If  $\HasType{\Gamma}{\MethCall{\e_0}{\m}{\e_1,\ldots,\e_n}}{\T}$, then $\HasType{\Gamma}{\e_0}{\bigvee_{i\in 1..m} \CC_i}$ and, 
there is $p \in 0..n$ such that $\e_0 = \ldots = \e_p$ and, 
for all $i \in 1..m$, 
\begin{itemize} 
\item $\mtype{\CC_i}{\m} = \funtype{\T_{i1}\ldots\T_{in}}{\T_i}$, and  
\item for all $k\in 1..p$, $\CC_i \leqfj \T_{ik}$, and 
\item for all $k\in p+1..n$, $\HasType{\Gamma}{\e_k}{\T_{ik}}$, and 
\item $\T_i\leqfj \T$. 
\end{itemize}
\item\label{lem:ilos:5} If   $\HasType{\Gamma}{\ifte\e{\e_1}{\e_2}}{\T}$, then $\HasType{\Gamma}{\e}{\boolType}$  and  $\HasType{\Gamma}{\e_i}{\T'}$ with $\T'\leqfj\T$ and $i \in 1..2$. 
\end{enumerate}
\end{lemma}
\begin{proofsketch}
We prove only points \ref{lem:ilos:2} and \ref{lem:ilos:4}. 
\begin{enumerate} 
\setcounter{enumi}{1}
\item The proof is by induction on the derivation of $\HasType{\Gamma}{\FieldAccess{\e}{\f}}{\T}$. 
For rule \rn{t-fld}, we have $\HasType{\Gamma}{\e}{\CC}$, $\fields{\CC} = \Field{\T_1}{\f_1}\ldots\Field{\T_n}{\f_n}$, $\f_i = \f$ and $\T_i = \T$, for some $i \in 1..n$. 
For rule \rn{t-sub}, the thesis is immediate by induction hypothesis. 
For rule \rn{t-$\vee$-elim}, we have $\ectx = \FieldAccess{[]}{\f}$, $\HasType{\Gamma}{\e}{\bigvee_{i\in 1..m} \CC_i}$ and $\HasType{\Gamma,\x{:}\CC_i}{\ectx[\x]}{\T}$, for all $i\in 1..m$, then, by induction hypothesis, 
for all $i \in 1..m$, we get $\HasType{\Gamma,\x{:}\CC_i}{\x}{\bigvee_{j\in 1..m_i} \DD_{ij}}$ and, for all $j\in 1..m_i$, 
$\fields{\DD_{ij}} = \Field{\T_{j11}}{\f_{j1}}\ldots\Field{\T_{jn_j}}{\f_{jn_j}}$ and $\T_{jk_j}\leqfj \T$, for some $k_j \in 1..n_j$. 
Since $\HasType{\Gamma,\x{:}\CC_i}{\x}{\bigvee_{j\in 1..m_j}\DD_{ij}}$, we have $\CC_i \leqfj \bigvee_{j\in 1..m_j} \DD_{ij}$\EZComm{per il punto 1 no?}, hence 
$\CC_i \leqfj \DD_{ij_i}$, for some $j_i \in 1..m_i$, by definition of subtyping. 
Then the thesis follows easily by assumption \InhFldCond. 

\setcounter{enumi}{3}
\item The proof is by induction  on the derivation of ${\HasType{\Gamma}{\MethCall{\e_0}{\m}{\e_1,\ldots,\e_n}}{\T}}$. 
For rule \rn{t-invk}, we have $\HasType{\Gamma}{\e_0}{\CC_0}$, $p = 0$,  ${\mtype{\CC_0}{\m} = \funtype{\T_1\ldots\T_n}{\T}}$, and, 
for all $k \in 1..n$, $\HasType{\Gamma}{\e_k}{\T_k}$. 
For rule \rn{t-sub}, the thesis is immediate by induction hypothesis. 
For rule \rn{t-$\vee$-elim}, we have $\ectx = \MethCall{[]}{\m}{[],\ldots,[],\e_{p+1},\ldots,\e_n}$, hence $p$ is the number of holes in $\ectx$ and $\e_0 = \ldots = \e_p$, 
and $\HasType{\Gamma}{\e_0}{\bigvee_{i\in 1..m} \CC_i}$ and, for all $i \in 1..m$, $\HasType{\Gamma,\x{:}\CC_i}{\ectx[\x]}{\T}$, with $\x$ fresh. 
By induction hypothesis, we know that, for all $i \in 1..m$, 
$\HasType{\Gamma,\x{:}\CC_i}{\x}{\bigvee_{j \in 1..m_i} \DD_{ij}}$ and there is $p_i \in 1..n$ such that 
the first $p_i$ arguments of $\ectx[\x]$ are equal to the receiver, namely $\x$ and this implies $p_i \le p$ because $\x$ is fresh.
Let $i \in 1..m$. 
Since $\HasType{\Gamma,\x{:}\CC_i}{\x}{\bigvee_{j \in 1..m_i} \DD_{ij}}$, we get $\CC_i \leqfj \bigvee_{j\in 1..m_j} \DD_{ij}$\EZComm{per il punto 1 no?}, thus $\CC_i \leqfj \DD_{ij_i}$, for some $j_i \in 1..m_i$, by definition of subtyping. 
Therefore, by induction hypothesis and assumption \InhMethCond, we get 
$\mtype{\CC_i}{\m} = \funtype{\T_{i1}\ldots\T_{in}}{\T_i}$ and, 
for all $k \in 1..p_i$, $\DD_{ij_i} \leqfj \T_{ik}$, hence $\CC_i\leqfj \T_{ik}$, and, 
for all $k \in p_i+1..p$, $\HasType{\Gamma,\x{:}\CC_i}{\x}{\T_{ik}}$, hence \EZComm{per il punto 1 no?} $\CC_i \leqfj \T_{ik}$ and, 
for all $k \in p+1..n$, $\HasType{\Gamma,\x{:}\CC_i}{\e_k}{\T_{ik}}$, hence, because $\x$ does not occur in $\e_k$ as it is fresh, by contraction we get $\HasType{\Gamma}{\e_k}{\T_{ik}}$, 
and, finally,  $\T_i \leqfj \T$. 
\end{enumerate}
\end{proofsketch}

\begin{lemma}[Key]\label{lem:key}
If $\HasType{\Gamma}{\val}{\bigvee_{1\leq i\leq n}\CC_i}$, then $\HasType{\Gamma}{\val}{\CC_i}$ for some $i\in 1\ldots n$. 
\end{lemma}

In order to state soundness,  
let \ple{\ConfSet_4,\ResSet_4,\RuleSet_4} be the big-step semantics defined in \cref{fig:FJOS-typesystem} ($\ConfSet_4$ is the set of expressions and $\ResSet_4$ is the set of values), and let $\ConfPredFour_\T = \{ \e \in \ConfSet_4 \mid \HasType{\es}{\e}{\T} \}$ and 
$\ResPredFour_\T = \{ \val \in \ResSet_4 \mid \HasType{\es}{\val}{\T} \}$, for $\T$  defined  in \cref{fig:FJOS-typesystem}. 
We need a last lemma to prove soundness:

\begin{lemma}[Determinism] \label{lem:dos}
If $\validInd{\RuleSet_4}{\eval{\e}{\val_1}}$ and $\validInd{\RuleSet_4}{\eval{\e}{\val_2}}$, then ${\val_1 = \val_2}$. 
\end{lemma}
\begin{proof}
Straightforward induction on rules in $\RuleSet_4$, because every syntactic construct has a unique big-step meta-rule. 
\end{proof}

\begin{theorem}[Soundness]\label{thm:sdos}
The big-step semantics \ple{\ConfSet_4,\ResSet_4,\RuleSet_4}  and the indexed predicate $\PredFour$ satisfy the conditions \LPCond, \ExPCond and \AllPCond of \cref{sect:bss-sc}.
\end{theorem}
\begin{proofsketch}
We sketch the proof only of \LPCond for rule \rn{invk}, other cases and conditions are similar to previous proofs. 

For rule \rn{invk}, \refItem{lem:ilos}{4}  applied to $\HasType\es{\MethCall{\e_0}{\m}{\e_1,\ldots,\e_n}}{\T}$ implies 
$\HasType{\es}{\e_0}{\bigvee_{i\in 1..m} \CC_i}$ and, there is $p \in 0..n$ such that 
$\e_0 = \ldots = \e_p$ and, for all $i \in 1..m$, 
$\mtype{\CC_i}{\m} = \funtype{\T_{i1}\ldots\T_{in}}{\T_i}$, and
for all $k \in 1..p$, $\CC_i \leqfj \T_{ik}$, and 
for all $k \in p+1..n$, $\HasType\es{\e_i}{\T_{ik}}$, and 
$\T_i \leqfj \T$. 
Assuming 
$\HasType{\es}{\ConstrCall{\CC}{\vals}}{\bigvee_{i\in 1..m} \CC_i }$, by \cref{lem:key} and \cref{lem:cfos}, we get 
$\CC \leqfj \CC_i$ for some $i \in 1..m$. 
Since $\mtype{\CC_i}{\m} = \funtype{\T_{i1}\ldots\T_{in}}{\T_i}$ and $\mbody{\CC}{\m} = \Pair{\x_1\ldots\x_n}{\e}$, {by assumption \InhMethCond and \MBodyCond,}
$\HasType{\kwthis{:}\CC,\x_1{:}\T_{i1},\ldots,\x_n{:}\T_{in}}{\e}{\T_i}$. 
Assume, for all $k \in 1..p$, $\HasType{\es}{\val_k}{\bigvee_{i\in 1..m} \CC_i}$ and, 
for all $k \in p+1..n$, $\HasType{\es}{\val_k}{\T_{ik}}$, then, 
since $\e_0 = \ldots = \e_p$, by \cref{lem:dos}, we get $\val_1 = \ldots = \val_p = \ConstrCall{\CC}{\vals}$, hence $\HasType{\es}{\val_k}{\T_{ik}}$, for all $k \in 1..p$, because $\CC \leqfj \CC_i \leqfj \T_{ik}$ for all $k \in 1..p$. 
\cref{lem:sos} gives $\HasType{\es}{\subst{\subst{\subst\e{\val_1}{\x_1}\ldots}{\val_n}{\x_n}}{\ConstrCall{\CC}{\vals}}{\kwthis}}{\T_i}$. 
Finally, we can conclude $\HasType{\es}{\val}{\T}$ by rule \rn{t-sub}, as $\T_i \leqfj \T$. 
\end{proofsketch}

\subsection{Imperative $\FJ$}
\label{sect:bss-ifj}

We show here how our technique behaves in an imperative setting. 
In \cref{fig:FJ-imp} and \cref{fig:FJ-impT} we show a minimal imperative extension of $\FJ$.  
We assume a well-typed class table and we use the notations introduced in \cref{sect:bss-fjl}. 
Expressions are enriched with field assignment and \emph{object identifiers} $\oid$, which only occur in runtime expressions.  
A \emph{memory} $\mem$ maps object identifiers to \emph{object states}, which are expressions of shape $\ConstrCall{\CC}{\oid_1,\ldots\oid_n}$. 
Results are configurations of shape $\Conf{\mem}{\oid}$. 
We denote by $\UpdateMem{\mem}{\oid}{i}{\oid'}$ the memory obtained from $\mem$ by replacing by $\iota'$ the $i$-th field of the object state associated \DA{with} $\oid$. 
The \emph{type assignment} $\Sigma$ maps object identifiers into types (class names). 
We write $\HasType{\Sigma}{\e}{\CC}$ for $\HasTypeMem{\emptyset}{\Sigma}{\e}{\CC}$. 

\begin{figure}
\begin{small}
\begin{math}
\begin{grammatica}
\produzione{\e}{\x\mid\FieldAccess{\e}{\f}\mid\ConstrCall{\CC}{\e_1,\ldots,\e_n}\mid
\MethCall{\e}{\m}{\e_1,\ldots,\e_n}\mid\FieldAssign{\e}{\f}{\e'}\mid\oid}{expressions}\\
\produzione{\conf}{\Conf{\mem}{\e}}{configurations}\\ 
\produzione{\res}{\Conf{\mem}{\oid}}{results} 
\end{grammatica}
\end{math}

\HSep

\begin{math}
\begin{array}{c}
\MetaRule{obj}{}{\eval{\Conf{\mem}{\oid}}{\Conf{\mem}{\oid}}}{} \BigSpace 
\MetaRule{fld}{\eval{\Conf{\mem}{\e}}{\Conf{\mem'}{\oid}}}
{ \eval{\Conf{\mem}{\FieldAccess{\e}{\f_i}}}{\Conf{\mem'}{\oid_i}} }
{ \mem'(\oid)=\ConstrCall{\CC}{\oid_1,\ldots,\oid_n}\\
  \fields{\CC}=\Field{\CC_1}{\f_1}\ldots\Field{\CC_n}{\f_n}\\
  i\in 1..n }
\\[2ex]
\MetaRule{new}{
  \eval{\Conf{\mem_i}{\e_i}}{\Conf{\mem_{i+1}}{\oid_i}}\Space \forall i\in 1..n
}{ \eval{\Conf{\mem}{\ConstrCall{\CC}{\e_1,\ldots,\e_n}}}{\Conf{\mem'}{\oid}} }
{ \mem_1=\mem\\
  \mem'=\SubstFun{\mem_{n+1}}{\ConstrCall{\CC}{\oid_1,\ldots,\oid_n}}{\oid}\\
  \oid\ \mbox{fresh} }
\\[4ex]
\MetaRule{invk}{
  \begin{array}{l}
  \eval{\Conf{\mem_i}{\e_i}}{\Conf{\mem_{i+1}}{\oid_i}}\Space \forall i\in 0..n\\
  \eval{\Conf{\mem_{n+1}}{\subst{\subst{\subst{\e}{\oid_1}{\x_1}\ldots}{\oid_n}{\x_n}}{\oid_0}{\kwthis}}}{\Conf{\mem'}{\oid}}
  \end{array}
}{ \eval{\Conf{\mem}{\MethCall{\e_0}{\m}{\e_1,\ldots,\e_n}}}{\Conf{\mem'}{\oid}} }
{ \mem_0=\mem\\
  \mem_1(\oid_0) =\ConstrCall{\CC}{\_}\\
  \mbody{\CC}{\m}=\Pair{\x_1\ldots\x_n}{\e} } 
\\[5ex]
\MetaRule{fld-up}{
  \eval{\Conf{\mem}{\e}}{\Conf{\mem'}{\oid}}\Space\eval{\Conf{\mem'}{\e'}}{\Conf{\mem''}{\oid'}}
}{ \eval{\Conf{\mem}{\FieldAssign{\e}{\f_i}{\e'}}}{\Conf{\UpdateMem{\mem''}{\oid}{i}{\oid'}}{\oid'}} }
{ \mem(\oid)=\ConstrCall{\CC}{\oid_1,\ldots,\oid_n}\\
  \fields{\CC}=\Field{\CC_1}{\f_1}\ldots\Field{\CC_n}{\f_n}\\
  i\in 1..n }
\end{array}
\end{math}

\HSep

\end{small}
\caption{Imperative $\FJ$: syntax and big-step semantics}\label{fig:FJ-imp}
\end{figure}

\begin{figure}
\begin{small}
\begin{math}
\begin{array}{c}
\MetaRule{t-conf}{
  \HasType{\Sigma}{\mem(\oid)}{\Sigma(\oid)}\ \forall\oid\in\dom(\mem)\Space 
  \HasType{\Sigma}{\e}{\CC}
}{ \HasType{\Sigma}{\Conf{\mem}{\e}}{\CC} }
{ \dom(\Sigma)=\dom(\mem) }
\\[3ex]
\MetaRule{t-var}{}{\HasTypeMem{\Gamma}{\Sigma}{\x}{\CC}}{\Gamma(\x)=\CC}
\\[3ex]
\MetaRule{t-fld}{
  \HasTypeMem{\Gamma}{\Sigma}{\e}{\CC}
}{ \HasTypeMem{\Gamma}{\Sigma}{\FieldAccess{\e}{\f_i}}{\CC_i} }
{ \fields{\CC}=\Field{\CC_1}{\f_1}\ldots\Field{\CC_n}{\f_n}\\
  i\in 1..n } 
\\[4ex]
\MetaRule{t-new}{
  \HasTypeMem{\Gamma}{\Sigma}{\e_i}{\CC_i}\Space \forall i\in 1..n
}{ \HasTypeMem{\Gamma}{\Sigma}{\ConstrCall{\CC}{\e_1,\ldots,\e_n}}{\CC} }
{ \fields{\CC}=\Field{\CC_1}{\f_1}\ldots\Field{\CC_n}{\f_n} }
\\[3ex]
\MetaRule{t-invk}{
  \HasTypeMem{\Gamma}{\Sigma}{\e_i}{\CC_i}\Space \forall i\in 0..n
}{\HasTypeMem{\Gamma}{\Sigma}{\MethCall{\e_0}{\m}{\e_1,\ldots,\e_n}}{\CC}}
{ \mtype{\CC_0}{\m}=\funtype{\CC_1\ldots\CC_n}{\CC} }
\\[3ex]
\MetaRule{t-fld-up}{
  \begin{array}{l}
  \HasTypeMem{\Gamma}{\Sigma}{\e}{\CC}\\
  \HasTypeMem{\Gamma}{\Sigma}{\e'}{\CC_i}
  \end{array}
}{\HasTypeMem{\Gamma}{\Sigma}{\FieldAssign{\e}{\f_i}{\e'}}{\CC_i}}
{ \fields{\CC}=\Field{\CC_1}{\f_1}\ldots\Field{\CC_n}{\f_n}\\
  i\in 1..n }
\\[4ex]
\MetaRule{t-oid}{}{\HasTypeMem{\Gamma}{\Sigma}{\oid}{\CC}}{\Sigma(\oid)=\CC}
\BigSpace
\MetaRule{t-sub}{\HasTypeMem{\Gamma}{\Sigma}{\e}{\CC}}{\HasTypeMem{\Gamma}{\Sigma}{\e}{\CC'}}
{\CC\leqfj\CC'}
\end{array}
\end{math}
\end{small}
\caption{Imperative $\FJ$: typing rules}\label{fig:FJ-impT}
\end{figure}

As for the other examples, to prove soundness we need some standard properties of the typing rules: inversion and substitution lemmas. 

\begin{lemma}[Inversion]\label{lem:ilimp}\
The following hold: 
\begin{enumerate}
\item\label{lem:ilimp:0} If $\HasType{\Sigma}{\Conf{\mem}{\e}}{\CC}$, then $\HasType{\Sigma}{\mem(\oid)}{\Sigma(\oid)}$ for all $\oid\in\dom(\mem)$ and  $\HasType{\Sigma}{\e}{\CC}$ and $\dom(\Sigma)=\dom(\mem)$.
\item\label{lem:ilimp:1} If  $\HasTypeMem{\Gamma}{\Sigma}{\x}{\CC}$, then $\Gamma(\x)\leqfj\CC$.
\item\label{lem:ilimp:2} If  $\HasTypeMem{\Gamma}{\Sigma}{\FieldAccess{\e}{\f_i}}{\CC}$, then $\HasTypeMem{\Gamma}{\Sigma}{\e}{\DD}$ and {\em $\fields{\DD}=\Field{\CC_1}{\f_1}\ldots\Field{\CC_n}{\f_n}$} and $\CC_i\leqfj \CC$ where 
$i\in 1..n$.
\item\label{lem:ilimp:3} If  $\HasTypeMem{\Gamma}{\Sigma}{\ConstrCall{\CC}{\e_1,\ldots,\e_n}}{\DD}$, then $\CC\leqfj\DD$ and {\em $\fields{\CC}=\Field{\CC_1}{\f_1}\ldots\Field{\CC_n}{\f_n}$} and $\HasTypeMem{\Gamma}{\Sigma}{\e_i}{\CC_i}$ for all $i\in 1..n$.
\item\label{lem:ilimp:4} If  $\HasTypeMem{\Gamma}{\Sigma}{\MethCall{\e_0}{\m}{\e_1,\ldots,\e_n}}{\CC}$, then  
$\HasTypeMem{\Gamma}{\Sigma}{\e_i}{\CC_i}$ for all $i\in 0..n$ and  {\em $\mtype{\CC_0}{\m}=\funtype{\CC_1\ldots\CC_n}{\DD}$} with $\DD\leqfj\CC$.
\item\label{lem:ilimp:5} If  $\HasTypeMem{\Gamma}{\Sigma}{\FieldAssign{\e}{\f_i}{\e'}}{\CC}$, then $\HasTypeMem{\Gamma}{\Sigma}{\e}{\DD}$ and $\fields{\DD}=\Field{\CC_1}{\f_1}\ldots\Field{\CC_n}{\f_n}$, with $i \in 1..n$, and $\HasTypeMem{\Gamma}{\Sigma}{\e'}{\CC_i}$  and $\CC_i\leqfj\CC$.
\item\label{lem:ilimp:6} If  $\HasTypeMem{\Gamma}{\Sigma}{\oid}{\CC}$, then $\Sigma(\oid)\leqfj\CC$.
\end{enumerate}
\end{lemma}

\begin{lemma}[Substitution]\label{lem:sosimp}
If $\HasTypeMem{\SubstFun\Gamma{\CC'}\x}{\Sigma} \e {\CC}$ and $\HasTypeMem{\Gamma}{\Sigma}{\e'} {\CC'}$, then $\HasTypeMem{\Gamma}{\Sigma}{\subst\e{\e'}\x} {\CC}$.
\end{lemma}

Let \ple{\ConfSet_5,\ResSet_5,\RuleSet_5} be the big-step semantics defined in \cref{fig:FJ-imp}. 
We can prove the soundness of the indexed predicate $\PredFive$ defined by: 
$\ConfPredFive_\Pair{\Sigma}{\CC} = \{ \Conf{\mem}{\e} \in \ConfSet_5 \mid  \HasType{\Sigma'}{\Conf{\mem}{\e}}{\CC} \mbox{ for some $\Sigma'$ s.t. $\Sigma\subseteq\Sigma'$} \}$ and 
$\ResPredFive_\Pair{\Sigma}{\CC} = {\ResSet_5 \cap \ConfPredFive_\Pair{\Sigma}{\CC}}$.  
The type assignment $\Sigma'$
is needed, since memory can grow during evaluation. 

\begin{theorem}[Soundness]\label{thm:sdimp}
The big-step semantics \ple{\ConfSet_5,\ResSet_5,\RuleSet_5}  and the indexed predicate $\PredFive$ satisfy the conditions \LPCond, \ExPCond and \AllPCond of \cref{sect:bss-sc}.
\end{theorem}
\begin{proof}
We prove separately the three conditions. 
The most interesting aspect here is that the presence of a memory induces a dependency between subsequent premises in each big-step rule and the hypotheses provided by the soundness conditions are essential to handle such a dependency. 

\proofsect{Proof of \LPCond} 
The proof is by cases on instantiations of meta-rules. 
\begin{proofcases}
\item [\rn{obj}] Trivial from the hypothesis. 

\item[\rn{fld}]  
\refItem{lem:ilimp}{0} applied to  $\HasType{\Sigma}{\Conf{\mem}{\FieldAccess{\e}{\f_i}}}{\CC}$ implies $\HasType{\Sigma}{\mem(\oid)}{\Sigma(\oid)}$ for all $\oid\in\dom(\mem)$ and  $\HasType{\Sigma}{\FieldAccess{\e}{\f_i}}{\CC}$ and $\dom(\Sigma)=\dom(\mem)$.
\refItem{lem:ilimp}{2} applied to $\HasType{\Sigma}{\FieldAccess{\e}{\f_i}}{\CC}$ implies $\HasType{\Sigma}{\e}{\DD}$ and $\fields{\DD}=\Field{\CC_1}{\f_1}\ldots\Field{\CC_n}{\f_n}$ and $\CC_i\leqfj \CC$ where 
$i\in 1..n$.  
Since $\eval{\Conf{\mem}{\e}}{\Conf{\mem'}{\oid}}$ is a premise we assume $\HasType{\Sigma'}{\Conf{\mem'}{\oid}}\DD$ with $\Sigma\subseteq\Sigma'$. 
\refItem{lem:ilimp}{0} and \refItem{lem:ilimp}{6} imply
$\Sigma'(\oid)\leqfj\DD$. 
\refItem{lem:ilimp}{3} allows us to get 
$\mem'(\oid)=\ConstrCall{\CC'}{\oid_1,\ldots\oid_m}$ with $n\leq m$ and $\CC'\leqfj\DD$ and  $\HasType{\Sigma'}{\oid_i}{\CC_i}$. 
So we conclude $\HasType{\Sigma'}{\Conf{\mem'}{\oid_i}}\CC$ by rules \rn{t-sub} and \rn{t-conf}.

\item [\rn{new}] 
\refItem{lem:ilimp}{0} applied to  $\HasType{\Sigma}{\Conf{\mem}{\ConstrCall{\CC}{\e_1,\ldots,\e_n}}}{\DD}$ implies $\HasType{\Sigma}{\mem(\oid)}{\Sigma(\oid)}$ for all $\oid\in\dom(\mem)$ and  $\HasType{\Sigma}{\ConstrCall{\CC}{\e_1,\ldots,\e_n}}{\DD}$ and $\dom(\Sigma)=\dom(\mem)$.
\refItem{lem:ilimp}{3} applied to $\HasType{\Sigma}{\ConstrCall{\CC}{\e_1,\ldots,\e_n}}{\DD}$ implies $\CC\leqfj\DD$ and $\fields{\CC}=\Field{\CC_1}{\f_1}\ldots\Field{\CC_n}{\f_n}$ and $\HasType{\Sigma}{\e_i}{\CC_i}$ for all $i\in 1..n$. 
Since $\eval{\Conf{\mem}{\e_i}}{\Conf{\mem_{i+1}}{\oid_i}}$ is a premise we assume 
$\HasType{\Sigma_i}{\Conf{\mem_{i+1}}{\oid_i}}{\CC_i}$ for all $i\in 1..n$ with $\Sigma\subseteq\Sigma_1\subseteq\cdots\subseteq\Sigma_n$.
\refItem{lem:ilimp}{0} and \refItem{lem:ilimp}{6} imply
$\Sigma_i(\oid_i)\leqfj\CC_i$  for all $i\in 1..n$. 
Using rules \rn{t-oid}, \rn{t-new} and \rn{t-sub} we derive $\HasType{\Sigma_n}{\ConstrCall{\CC}{\oid_1,\ldots,\oid_n}}{\DD}$. 
We then conclude $\HasType{\Sigma_n,\oid:\DD}{\Conf{\mem_{n+1}}\oid}{\DD}$ by rules \rn{t-oid} and \rn{t-conf}.

\item [\rn{invk}] 
\refItem{lem:ilimp}{0} applied to  $\HasType{\Sigma_0}{\Conf{\mem_0}{\MethCall{\e_0}{\m}{\e_1,\ldots,\e_n}}}{\CC}$ implies $\HasType{\Sigma_0}{\mem_0(\oid)}{\Sigma_0(\oid)}$ for all $\oid\in\dom(\mem_0)$ and  $\HasType{\Sigma_0}{\MethCall{\e_0}{\m}{\e_1,\ldots,\e_n}}{\CC}$ and $\dom(\Sigma_0)=\dom(\mem_0)$.
\refItem{lem:ilimp}{4} applied to $\HasType{\Sigma_0}{\MethCall{\e_0}{\m}{\e_1,\ldots,\e_n}}{\CC}$ implies  
$\HasType{\Sigma_i}{\e_i}{\CC_i}$ for all $i\in 0..n$ and  $\mtype{\CC_0}{\m}=\funtype{\CC_1\ldots\CC_n}{\DD}$ with $\DD\leqfj\CC$. Since $\eval{\Conf{\mem_i}{\e_i}}{\Conf{\mem_{i+1}}{\oid_i}}$ is a premise we assume $\HasType{\Sigma_i}{\Conf{\mem_{i+1}}{\oid_i}}{\CC_i}$ for all $i\in 0..n$ with $\Sigma_0\subseteq\cdots\subseteq\Sigma_n$. \refItem{lem:ilimp}{0} gives $\HasType{\Sigma_i}{\oid_i}{\CC_i}$ for all $i\in 0..n$. 
The typing of the class table implies ${\HasType{\x_1{:}\CC_1,\ldots,\x_n{:}\CC_n,\kwthis{:}\CC_0}{\e}{\DD}}$.
\cref{lem:sosimp} gives $\HasType{\Sigma_n}{\e'}\DD$ where $\e'=\subst{\subst{\subst{\e}{\oid_1}{\x_1}\ldots}{\oid_n}{\x_n}}{\oid_0}{\kwthis}$. Using rules \rn{t-sub} and 
\rn{t-conf} we derive $\HasType{\Sigma_n}{\Conf{\mem_{n+1}}{\e'}}\CC$.
Since $\eval{\Conf{\mem_{n+1}}{\e'}}{\Conf{\mem'}{\oid}}$ is  a premise we conclude $\HasType{\Sigma'}{\Conf{\mem'}{\oid}}\CC$ with $\Sigma_n\subseteq\Sigma'$.

\item [\rn{fld-up}] 
\refItem{lem:ilimp}{0} applied to  $\HasType{\Sigma}{\Conf{\mem}{\FieldAssign{\e}{\f_i}{\e'}}}{\CC}$ implies $\HasType{\Sigma}{\mem(\oid)}{\Sigma(\oid)}$ for all $\oid\in\dom(\mem)$ and  $\HasType{\Sigma}{\FieldAssign{\e}{\f_i}{\e'}}{\CC}$ and $\dom(\Sigma)=\dom(\mem)$. \refItem{lem:ilimp}{5} applied to $\HasType{\Sigma}{\FieldAssign{\e}{\f_i}{\e'}}{\CC}$ implies $\HasType{\Sigma}{\e}{\DD}$ and $\fields{\DD}=\Field{\CC_1}{\f_1}\ldots\Field{\CC_n}{\f_n}$ and $\HasType{\Sigma}{\e'}{\CC_i}$  and $\CC_i\leqfj\CC$. Since $\eval{\Conf{\mem}{\e}}{\Conf{\mem'}{\oid}}$ and $\eval{\Conf{\mem'}{\e'}}{\Conf{\mem''}{\oid'}}$ are premises we assume  $\HasType{\Sigma'}{\Conf{\mem'}{\oid}}\DD$  and $\HasType{\Sigma''}{\Conf{\mem''}{\oid'}}\CC_i$, with $\Sigma\subseteq\Sigma'\subseteq\Sigma''$. 
Notice that $\mem''(\oid)$ and $\UpdateMem{\mem''}{\oid}{i}{\oid'}(\oid)$ have the same types for all $\oid$ by construction. We conclude $\HasType{\Sigma''}{\Conf{\UpdateMem{\mem''}{\oid}{i}{\oid'}}{\oid'}}\CC_i$.
\end{proofcases}

\proofsect{Proof of \ExPCond} All the closed expressions appear as conclusions in the reduction rules. 

\proofsect{Proof of \AllPCond}  
Since the only values are configurations with object identifiers it is easy to verify that the premises of the reduction rules are satisfied, being the conditions on memory and object identifiers assured by the typing rules. 
\end{proof}

  \section{Concluding discussions}
\label{sect:conclu} 

The big-step style can be useful for abstracting details or 
directly deriving the implementation of an interpreter. 
However, reasoning on properties involving infinite computations, such as the soundness of a type system, is non-trivial, 
because standard big-step semantics is able only to capture finite computations, hence it cannot distinguish between stuck and infinite \bez ones. \eez

In this paper, we address this problem, providing a systematic analysis of big-step semantics. 
The first, and fundamental, methodological feature of our analysis is that we want to be \emph{independent from specific languages}, developing an abstract study of big-step semantics in itself. 
Therefore, we provide a definition of what a big-step semantics is, so our results  will be applicable, as we show by several examples, to all concrete big-step semantics matching our definition. 

A second important building block of our approach is that we take seriously the fact that big-step rules implicitly define an \emph{evaluation algorithm}.
Indeed, we make such intuition formal by showing that \bez starting from the \eez rules we can define a transition relation on incomplete derivations, abstractly modeling such evaluation algorithm. 
Relying on this transition relation, we are able to define computations in the big-step semantics in the usual way, as possibly infinite sequences of transition steps; 
thus we can distinguish converging, diverging and stuck computations, even \bez though \eez big-step rules only define convergence. 
This shows that diverging and stuck computations are, in a sense, implicit in standard big-step rules, and the transition relation makes them explicit. 

Finally, the third feature of our approach is that we provide \emph{constructions} that, starting from a usual big-step semantics, produce an extended one where the distinction between diverging and stuck computation is explicit. 
Such constructions show that we can distinguish stuckness and divergence directly by a big-step semantics, without resorting to a transition relation: 
we rely on the above described transition relation on incomplete derivations only to prove that the constructions are correct. 
Corules are crucial to define extended big-step semantics precisely modelling divergence just as a special result, thus avoiding the redundancy introduced by traces. 

Building on this systematic study,
we show how one can reason about soundness of a predicate directly on a big-step semantics. 
To this end, we design proof techniques for two flavours of soundness, based on sufficient conditions on big-step rules. 


  \subsection{Related work}

The research presented in this paper 
follows a stream of work dating back to \citet{CousotC92}, who proposed a stratified approach, investigated by \citet{LeroyG09} as well, with a separate judgment for divergence, defined coinductively. 
In this way, however, there is no unique formal definition of the behaviour of the modelled system. 
An alternative possibility, also investigated by \citet{LeroyG09}, is to interpret coinductively the standard big-step rules (\emph{coevaluation}).
Unfortunately, coevaluation is non-deterministic, allowing the derivation of spurious judgements,  and, thus, may fail to correctly
capture the infinite behavior of a configuration: a diverging term, such as $\Omega$, evaluates to any value, hence it cannot be properly distinguished from  converging terms.
Furthermore, in coevaluation there are still configurations, such as $\Omega\appop(0\appop 0)$, for which no judgment can be derived, here because  no judgment can be derived for the subterm $0\appop 0$;  
basically, this is due to the fact that divergence of a premise should be propagated and this cannot be correctly handled by coevaluation as divergence is not explicitly modelled. 

\emph{Pretty big-step semantics} by  \citet{Chargueraud13} handles the issue of duplication of meta-rules by a
unified judgment with a unique set of (meta-)rules and divergence modelled by a special value. 
Rules are interpreted coinductively, hence they allow the derivation of spurious judgements, but, thanks to the use of a special value for divergence and the particular structure of rules, they can solve most of the issues of coevaluation. 
However, this particular structure of rules is not as natural as usual big-step rules and, more importantly, 
it requires the introduction of new specific syntactic forms representing intermediate computation steps, as in small-step semantics, hence making the big-step semantics less abstract. 
This may be a problem, for instance, when proving soundness of a type system, as such intermediate configurations may be ill-typed. 

\citet{PoulsenM17} subsequently present \emph{flag-based big-step semantics},
which further streamlines the approach by combining it with the M-SOS technique
(modular structural operational semantics), thereby reducing
the number of (meta-)rules and premises, avoiding the need for intermediate configurations. 
The key idea is to extend configurations and results by flags explicitly modelling convergence and divergence, used to properly handle divergence propagation. 
To model divergence, they interpret rules coinductively, hence they allow the derivation of spurious judgements. 

Differently from all the previously cited papers, which consider specific examples, the work by \citet{Ager04} shares with us the aim of providing a generic construction to model non-termination, basing on an arbitrary big-step semantics. 
Ager considers a big-step judgement of shape $\rho\vdash t \Downarrow v$ where $\rho$ is an environment, $t$ a syntactic term and $v$ a final value, and 
values, environments and the signature for terms are left unspecified. 
Then, given a big-step semantics, he describes a method to extract an abstract machine from it, which models a proof-search algorithm. 
In this way, converging, diverging and stuck computations are distinguished. 
This approach is somehow similar to our transition relation on partial evaluation trees, even tough a different style is used: 
we have no syntactic components and 
the transition system we propose is directly defined on evaluation trees and corresponds to a partial order on them, modelling refinement. 
Moreover, Ager's notion of big-step semantics is not fully formal, in particular, it is not clear whether he works with plain rules or meta-rules. 

Another piece of work whose aim is to define a general framework for operational semantics specification is the one by \citet{BodinGJS19} on \emph{skeletal semantics}. 
Here the key idea is to specify the semantics of a language by a set of skeletons, one for each syntactic construct, which describe how to evaluate each of them. 
Skeletons are very much like big-step rules, indeed they can be regarded as an ad-hoc syntax for specifying them. 
This syntax is quite unusual, but probably better suited for the Coq implementation which the framework comes with. 
This approach is not specifically tailored for big-step semantics: 
a skeletal specification can give rise to semantics in different styles, such as big-step, small-step or abstract machines. 
However, given the similarity between skeletons and big-step rules, it may be possible to adapt the proof technique we propose to this setting, but this is matter for future work. 

\citet{AnconaDZ@oopsla17} firstly show that with corules one can define a unified big-step judgment with a unique set of rules avoiding spurious evaluations. 
This can be seen as \emph{constrained coevaluation}. 
Indeed, corules add constraints on the infinite derivations to filter out spurious results, so that,  
for diverging terms, it is only possible to get $\divres$ as result. 
This is extended to include observations as traces by \citet{AnconaDZ@ecoop18}. 
A further step is done by \citet{AnconaDRZ20}, where observations are modelled by an arbitrary monoid and a variant of the construction described in \cref{sect:bss-div} is considered. 

Other proposals, by \citet{OwensMKT16,AminRompf17},
are inspired by \emph{definitional interpreters} \citep{Reynolds72}, based on a step-indexed approach (a.k.a.``fuel''-based semantics)
where computations are approximated to some finite amount of steps (typically with a counter);
in this way divergence can be modeled by induction.
\citet{OwensMKT16} investigates functional big-step semantics for proving by induction compiler
correctness.
\citet{AminRompf17} explore inductive proof strategies for type soundness properties 
for the polymorphic type systems $F_{<:}$, and equivalence with small-step semantics.
An inductive proof of type soundness for the big-step semantics of a Java-like language is proposed by \citet{Ancona14}.

Coinductive trace semantics in big-step style have been studied by \citet{NakataU09,NakataU10,NakataU10a}.
Their investigation started with the semantics of an imperative While language  with no I/O \citep{NakataU09},
where traces are possibly infinite sequences of states; semantic rules are all coinductive and define two mutually dependent judgments.
Based on such a semantics, they define a Hoare logic \citep{NakataU10}. 
They provide a constructive theory and metatheory, together with a Coq formalization of their results. 
Differently from our approach, 
weak bisimilarity between traces is needed 
for proving that programs exhibit equivalent observable behaviors. 
This is due to the fact that ``silent effects'' (that is, non-observable internal steps) must be explicitly represented
to guarantee guardedness conditions which ensure productivity of corecursive definitions. 
This is a natural consequence of having computable definitions. 
By using corules, we can avoid bisimilarity, accepting an approach which is not fully constructive. 

This semantics has been subsequently extended with interactive I/O \citep{NakataU10a}, by exploiting
the notion of resumption monad: a tree representing possible runs of a program to model its non-deterministic behavior due to input values.
Also in this case a big-step trace semantics is defined with two mutually recursive coinductive judgments, and
weak bisimilarity is needed; however, the definition of the observational equivalence is more involved,
since it requires nesting inductive definitions in coinductive ones.
A generalised notion of resumption has been introduced later by \citet{PirogG14}  in a category-theoretic and coalgebraic context.

\citet{Danielsson12}, inspired by \citet{LeroyG09},
relying on the coinductive partiality monad, defines  big-step semantics for
$\lambda$-calculi and virtual machines as total, computable functions able to capture divergence.
 

The resumption monad of  \citet{NakataU10a} and the partiality monad of \citet{Danielsson12} are inspired by the seminal work of \citet{Capretta05}
on the \emph{delay monad}, where coinductive types are exploited to model infinite computations
by means of a type constructor for partial elements, which allows the formal definition of convergence and divergence
and a type-theoretic representation of general recursive functions; this type constructor
is proved to constitute a strong monad, upon which subsequent related 
papers \citep{AbelC14,McBride15,ChapmanUV19} elaborated to define other monads for managing divergence.
In particular, \citet{McBride15} has proposed a more general approach based on a free monad
for which the delay monad is an instantiation obtained through a monad morphism. 
All these proposals are based on the step-indexed approach.

More recently, interaction trees (ITrees) \citep{itrees20} have been presented as a coinductive variant of free monads with the main aim
of defining the denotational semantics for effectful and possibly nonterminating computations,  to allow
compositional reasoning for mutually recursive components of an interactive system, with fully mechanized proofs in Coq.
Interaction trees are coinductively defined trees which directly support a more general fixpoint combinator
which does need a step-indexed approach, as happens for the general monad of McBride.
A \lstinline{Tau} constructor is introduced to represent a silent step of computation, 
to express silently diverging computations without violating Coq's guardedness condition;
as a consequence, \bez a \eez generic definition of weak bisimulation on ITrees is required to
remove any finite number of \lstinline{Tau}s, similarly as what happens in the approach of Nakata and Uustalu.


  \subsection{Future work} 
\label{sect:bss-future} 

There are several directions for further research. 
A first direction is to study other approaches to model divergence in big-step semantics using our general meta-theory, that is, 
defining yet other constructions, such as adding a counter and timeout, as done by  \citet{OwensMKT16,AminRompf17}, or adding flags, as done by \citet{PoulsenM17}. 
This would provide a general account of these approaches, allowing to study their properties in general, abstracting away particular features of concrete languages. 
A further direction is to consider other computational models such as probabilistic computations, which are quite difficult to model in big-step style, as shown by \citet{DalLagoZ12}. 

Concerning proof techniques for soundness, 
we also plan to compare our proof technique with the standard one for small-step semantics: 
if a predicate satisfies progress and subject reduction with respect to a small-step semantics, does it satisfy our soundness conditions  with respect to  an equivalent big-step semantics? 
To formally prove such a statement, the first step will be to express equivalence between small-step and big-step  semantics, and such equivalence has to be expressed  at the level of big-step rules, as it needs to be extendible to stuck and infinite computations. 
Note that, as a by-product, this will provide us with a proof technique to show equivalence between small-step and big-step semantics. 
\citet{AnconaDRZ20} make a first attempt to express such an equivalence for a more restrictive class of big-step semantics. 
On the other hand, 
the converse does not hold, as  shown 
by the examples in \cref{sect:bss-fjl} and \cref{sect:bss-fjos}. 

Furthermore, it would be interesting to extend such techniques for soundness to big-step semantics with observations, taking inspiration from type and effect systems \citep{MarinoM09,Tate13}. 


Last but not least, to support reasoning by our framework on concrete examples, such as those in \cref{sect:bss-ex}, 
it is desirable to have a mechanisation of our meta-theory and related techniques. 
A necessary preliminary step in this direction is to provide support for corules in proof assistants. 
An Agda library supporting (generalised) inference systems is described by \citet{CicconeDZ21} and can be found at 
\url{https://github.com/LcicC/inference-systems-agda}. 
Moreover, in the paper we lazily relied on the usual setting of classical logic (even though we try not to abuse of it), 
however, towards a formalisation, we will have to carefully rearrange definitions and proofs to fit the logic of the choosen proof assistant.

\begin{acks}
Special thanks go to Elena Zucca, Mariangiola Dezani-Ciancaglini and Viviana Bono for collaborating on this work with many interesting discussions and useful suggestions, which have greatly improved the paper. 
\end{acks}

\bibliographystyle{ACM-Reference-Format}
\bibliography{biblio}



\end{document}
\endinput